\documentclass[11pt]{article}
\pdfoutput=1
\usepackage[margin=1in]{geometry} 
\usepackage{amsmath,amsthm,amssymb,amsfonts,titling,bbm, titlesec, bm, physics, enumitem, accents, xcolor, setspace, fancyhdr, natbib, geometry,
	pdflscape, mathrsfs}
\usepackage{graphicx}
\usepackage{float}
\usepackage[font=footnotesize,labelfont=bf]{caption}
\usepackage{array}
\usepackage[title]{appendix}
\usepackage{afterpage}
\usepackage{listings}
\usepackage{prodint}
\usepackage{algorithm}
\usepackage{algpseudocode}
\usepackage{algorithmicx}

\usepackage{xcolor}
\definecolor{Bleu}{RGB}{0,0,204}
\usepackage[hypertexnames=false]{hyperref}
\hypersetup{
colorlinks,
  citecolor=Bleu,
  linkcolor=Bleu,
  urlcolor=Bleu}

\usepackage{booktabs}

\usepackage{subcaption}
\newlength{\continueindent}
\usepackage{etoolbox}
\makeatletter
\newcommand*{\ALG@customparshape}{\parshape 2 \leftmargin \linewidth \dimexpr\ALG@tlm+\continueindent\relax \dimexpr\linewidth+\leftmargin-\ALG@tlm-\continueindent\relax}
\apptocmd{\ALG@beginblock}{\ALG@customparshape}{}{\errmessage{failed to patch}}
\makeatother

\algnewcommand\algorithmicstack{\textit{Stack:}}
\algnewcommand\Stack{\item[\algorithmicstack]}
\algnewcommand\algorithmicchoosetimegrid{\textit{Choose time grid:}}
\algnewcommand\Choosetimegrid{\item[\algorithmicchoosetimegrid]}
\algnewcommand\algorithmicchoosetimebasis{\textit{Choose time basis:}}
\algnewcommand\Choosetimebasis{\item[\algorithmicchoosetimebasis]}
\algnewcommand\algorithmicbuilddiscretedata{\textit{Build data at time $t_i$:}}
\algnewcommand\Builddiscretedata{\item[\algorithmicbuilddiscretedata]}
\algnewcommand\algorithmicfit{\textit{Fit:}}
\algnewcommand\Fit{\item[\algorithmicfit]}
\algnewcommand\algorithmicpredict{\textit{Predict:}}
\algnewcommand\Predict{\item[\algorithmicpredict]}

\providecommand{\keywords}[1]
{
  \small	
  \textbf{\textit{Keywords---}} #1
}

\DeclareMathOperator*{\argmin}{argmin}

\allowdisplaybreaks

\newtheorem{theorem}{Theorem}
\newtheorem{lemma}{Lemma}
\newtheorem{corollary}{Corollary}

\theoremstyle{definition}
\newtheorem{example}{Example} 
\newtheorem*{example*}{Example} 
 \usepackage{setspace}
\newcommand{\slightspacing}{\setstretch{1.175}}

\newtheorem{remark}{Remark}

\titleformat{\section}{\large\scshape\bfseries}{\thesection.}{1em}{}
\titleformat{\subsection}{\normalfont\bfseries}{\thesubsection.}{1em}{}

\title{Adaptive debiased machine learning\\ using data-driven model selection techniques}

\author{Lars van der Laan, Marco Carone,  Alex Luedtke, Mark van der Laan}

\begin{document}

\singlespacing
\maketitle
\begin{abstract}
Debiased machine learning estimators for smooth functionals in nonparametric models can exhibit substantial variability and instability, often leading practitioners to instead rely on parametric or semiparametric working models. Such models, however, may be misspecified and can therefore introduce bias. We study how data-driven model selection can be combined with debiased machine learning to construct estimators that adapt to structure in the data-generating distribution. To this end, we propose Adaptive Debiased Machine Learning (ADML), a nonparametric framework for constructing superefficient estimators of pathwise differentiable parameters. The framework unifies a broad class of previously proposed adaptive estimators, including methods based on variable selection, learned feature representations, and collaborative targeted learning. It requires only high-level conditions and approximate validity of the selection procedure, which are implied by lower-level conditions already assumed in important settings, including sieve-based selection, sparsity-based methods such as the Lasso, and data-adaptive feature representations. We show that ADML estimators yield regular and efficient root-\(n\) inference for an oracle projection parameter induced by a data-adaptive oracle submodel. This oracle parameter coincides with the target parameter at the true distribution but typically has a smaller efficiency bound, thereby yielding superefficiency for the target parameter. As a practical illustration, we introduce a broad class of automatic ADML estimators for continuous linear functionals of the outcome regression, in which model selection is performed directly on the regression itself. Motivated by overlap challenges in causal inference, we develop new superefficient plug-in estimators for the average treatment effect based on calibration in semiparametric regression models.

\end{abstract}


\doublespacing

\slightspacing

\section{Introduction}


For many scientific applications, including treatment effect estimation and policy learning, it is important to infer real-valued summaries, that is, functionals, of probability distributions. Several debiased machine learning (DML) frameworks are available for this purpose, including one-step estimation \citep{levit1975efficiency, pfanzagl1985contributions, hasminskii1979nonparametric, bickel1993efficient}, estimating equations and double machine learning \citep{robins1995analysis, robinsCausal, vanderlaanunified, DoubleML}, targeted maximum likelihood estimation \citep{laan_rubin_2006, vanderLaanRose2011}, and sieve-based plug-in estimation \citep{shen1997methods, spnpsieve, sieveOneStepPlugin, SieveQiu, discussionMLEMark,undersmoothedHAL}. These frameworks typically involve two stages: a preliminary estimation step, in which flexible machine learning methods are used to estimate the data-generating distribution, and a debiasing step, which enables valid uncertainty quantification under a prespecified statistical model. When the model is correctly specified, these methods yield regular, asymptotically linear estimators that are root-\(n\) consistent and efficient among all regular estimators \citep{bickel1993efficient, vandervaart2000asymptotic}. Moreover, such efficient estimators are locally asymptotically minimax over all estimators, including irregular ones, relative to the statistical model \citep{vandervaart2000asymptotic}. That is, asymptotically, they minimize the worst-case mean squared error over local perturbations of the data-generating distribution that remain within the model.

While debiasing methods are effective for constructing efficient and locally asymptotically minimax estimators, they have an important limitation: valid debiasing and uncertainty quantification require \textit{a priori} specification of a correct statistical model and therefore do not adapt to the complexity of the true data-generating distribution. In particular, the truth may lie in an unknown but learnable sparse, smooth, or otherwise structured submodel of the prespecified model, which we call the \emph{oracle submodel}. Standard debiasing methods do not exploit this structure: their limiting variance does not decrease simply because the truth lies in a simpler submodel. This lack of adaptivity is a consequence of their regularity and local asymptotic minimaxity over the full prespecified model, including distributions that may be substantially more complex than the truth \citep{hajek1972local, vandervaart2000asymptotic}. In practice, however, inference based on a smaller model, whether specified \textit{a priori} \citep{imbensOverlapEstimand2006} or learned from the data \citep{moosavi2023costs}, can yield substantial gains in stability and efficiency. This is particularly relevant in settings such as causal effect estimation under limited overlap, where fully nonparametric regular estimators may perform poorly \citep{CTMLE, ATEsupereff, d2021deconfounding, moosavi2023costs}.

A common strategy is to impose a simpler parametric or semiparametric model \textit{a priori} \citep{imbensOverlapEstimand2006, li2019overlapWeights, vansteelandt2022assumption, vansteelandt2024assumption}. Such models may be used as working models for estimation of the original parameter, while acknowledging possible misspecification \citep{white1982MLErobust, buja2019models, buja2019models2}. They may also induce an alternative target parameter defined through projection onto the simpler model, such as the overlap-weighted ATE under a homogeneous partially linear model \citep{robinson1988root, imbensOverlapEstimand2006, li2019overlapWeights}. The appeal is that simpler models can yield more stable and efficient estimators, as reflected in the semiparametric efficiency bound for the projection parameter. However, these restrictions are often motivated by convenience or subjective judgments about the complexity of the data-generating distribution, and they may induce bias under misspecification. Moreover, when the target is redefined through projection, the resulting estimand may be less interpretable and may differ substantially from the original target \citep{whitney2020comment}. Data-adaptive approaches that learn a simpler model from the data may therefore offer a more favorable bias--variance tradeoff \citep{CTMLE, moosavi2023costs, van2024adaptive, li2025regularized}. At the same time, inference after data-driven model selection is delicate, since standard theoretical guarantees need not remain valid \citep{bauer1988model, potscher1991effects, buhlmann1999efficient, hjort2003frequentist, bunea2004consistent, LeebModelSelect2005, claeskens2007asymptotic}.

A growing body of work has shown, both empirically and, in some cases, theoretically, that incorporating data-adaptive model selection into DML---for example, through variable selection or learned feature representations for confounding adjustment---can substantially improve efficiency and stability relative to nonadaptive approaches \citep{brookhart2006variable, CTMLE, outcomeadaptivelasso, cui2019selective, ATEsupereff, d2021deconfounding, assaad2021counterfactual, moosavi2023costs}. Broadly, these methods aim to calibrate model complexity, such as the variables used for confounding adjustment, jointly with the strength of the debiasing step, so that any increase in variance is offset by a meaningful reduction in bias. For example, \citet{moosavi2023costs} showed that automatic confounder selection using the Lasso can substantially improve the efficiency and stability of DML estimators, particularly under limited overlap; one reason is that adjusting for instrumental variables can inflate variance without reducing bias, so such variables may be screened out based on their association with the outcome. More broadly, there is a rich literature on outcome-adaptive variable selection for propensity score models and covariate adjustment \citep{brookhart2006variable, westreich2011role, wyss2013variable, zhu2015variable, dukes2020obtain, outcomeadaptivelasso, CTMLE, CTMLELasso, HALoutcomeAdapt}, including the outcome-adaptive Lasso \citep{outcomeadaptivelasso} and outcome-adaptive Highly Adaptive Lasso \citep{HALoutcomeAdapt}.  Collaborative targeted learning (CTMLE) \citep{CTMLE} similarly improves performance by adaptively tuning the debiasing step and selecting increasingly complex models for orthogonal nuisance parameters, such as the propensity score \citep{gruber2010application, pirracchio2018collaborative, schnitzer2018collaborative, CTMLELasso, ju2018collaborative, ju2019scalable, adaptiveTruncationCTMLE, ATEsupereff}. Augmented minimax linear estimation \citep{augmentedMinimax} adjusts for features that predict the residual error in the outcome regression. More generally, variable selection can be viewed as a special case of learning low-dimensional representations for confounding adjustment. In this vein, \citet{d2021deconfounding} proposed deconfounding scores as a data-adaptive dimension reduction of the covariate space that can improve efficiency and stability when overlap is weak, while \citet{ATEsupereff} introduced a superefficient TMLE for the average treatment effect based on a one-dimensional deconfounding score derived from the outcome regression, and \citet{wang2023super} considered dimension reductions constructed from multiple estimated regression models in a Super Learner ensemble.

\subsection{Contributions of this work}

We propose a nonparametric framework for adaptive inference on pathwise differentiable parameters that combines data-driven model selection with debiased machine learning, which we call \textit{adaptive debiased machine learning} (ADML). A central contribution of this work is conceptual: ADML provides a unified perspective on how data-driven model selection can be combined with semiparametric efficiency theory to construct asymptotically linear, superefficient estimators that adapt to learnable structure in the data-generating distribution, without requiring restrictive parametric or semiparametric assumptions \textit{a priori}. The framework unifies a broad class of previously proposed adaptive estimators, including methods based on selection from sieves of increasing complexity \citep{CTMLE, ju2019scalable}, variable- or basis-selection procedures such as the (outcome-adaptive) Lasso \citep{moosavi2023costs, outcomeadaptivelasso, HALoutcomeAdapt, apfel2024learning}, learned feature representations \citep{ATEsupereff, d2021deconfounding, assaad2021counterfactual, wang2023super}, augmented minimax linear estimation \citep{augmentedMinimax}, and collaborative targeted learning \citep{CTMLE, ATEsupereff}. We further develop high-level conditions on the model-selection procedure and show that they can be verified under lower-level assumptions in several important settings.

Our main contributions are as follows:
\begin{enumerate}
    \item We propose a unified framework for adaptive inference on pathwise differentiable functionals that integrates data-driven model selection with debiased machine learning. Under the nonparametric model, we establish asymptotic linearity, superefficiency, and regularity properties of ADML estimators.

    \item We formalize the working and oracle targets of ADML as projection parameters induced by optimization over the working and oracle models, and we derive their efficient influence functions and semiparametric efficiency bounds.

\item We develop a new decomposition of the model approximation error induced by data-driven model selection for projection parameters, and characterize its second-order structure.

 \item We propose an adaptive extension of the autoDML framework of \cite{chernozhukov2022automatic} for regression functionals and introduce new superefficient calibrated plug-in estimators.
\end{enumerate}

Taken together, these results suggest a general principle for adaptive inference: if a data-driven working model approximates an oracle submodel sufficiently well, then the effect of learning that model is higher order. Thus, a broad class of adaptive estimators can be viewed as nonparametrically valid procedures for an oracle projection target, with efficiency determined by simpler oracle structure.

Related work includes post-model-selection inference, inference for data-adaptive target parameters, and calibration-based methods in causal inference; see Appendix~\ref{section::AMLEGeneral3Lit} for a broader review. A large literature has shown that model selection can improve efficiency at selected distributions, often at the cost of regularity and uniform validity \citep{LeebModelSelect2005, claeskens2007asymptotic, wu2019hodges, moosavi2023costs}. Our results complement this literature by showing that, for pathwise differentiable parameters, replacing an oracle submodel with a learned working model induces only second-order error under suitable conditions. Unlike approaches based on consistent model selection or oracle properties, our framework requires only that the selected working model approximate a fixed oracle submodel at the relevant sample size, analogously to approximate sparsity conditions in post-Lasso inference \citep{belloni2012sparse, belloni2014inference}. Our work also extends the literature on data-adaptive target parameters \citep{van2013AdaptTarget, rinaldo2019bootstrapping} by characterizing when inference for a data-adaptive projected target remains valid for a fixed oracle target. Finally, our examples connect ADML to recent calibration-based methods in causal inference and debiased machine learning \citep{van2024automatic, van2024stabilized}.

This paper is organized as follows. Section~\ref{section::AMLEGeneral} introduces the ADML framework, presents the main results, illustrates the approach for causal inference with linear functionals of the outcome regression, and reviews related literature. Section~\ref{section::oracleparameter} defines a class of oracle projection parameters and examines their role in the construction of superefficient estimators. Section~\ref{section::modelapprox} studies the approximation error induced by data-driven model selection and provides conditions under which it is asymptotically negligible. Section~\ref{section::theory} presents the main theoretical results for ADML. Finally, Section~\ref{adml::linearfunc} specializes the theory to linear functionals of a regression function. We provide a summary of notation in Appendix~\ref{app:notation}.

\section{Adaptive debiased machine learning}
\label{section::AMLEGeneral}

\subsection{Preliminaries}

\label{sec::prelim}

We review core concepts in semiparametric efficiency theory \citep{bickel1993efficient}. Let $\mathcal{M}$ be a statistical model, a collection of probability distributions dominated by a common sigma-finite measure $\mu$. For a given $P\in\mathcal{M}$, we denote the $L^2(P)$ norm as $\|\cdot\|_P$. A one-dimensional submodel $\{P_t: t \in \mathbb{R}\} \subseteq \mathcal{M}$ through $P$ at $t=0$ is called \textit{regular} if it is differentiable in quadratic mean at $t=0$. The tangent space $T_\mathcal{M}(P)$ of $\mathcal{M}$ at $P$ is the $L^2(P)$-closure of scores generated by regular one-dimensional submodels of $\mathcal{M}$ through $P$. We assume that $\mathcal{M}$ (and all models considered in this work) are smooth, meaning that for every $P \in \mathcal{M}$, the tangent space $T_\mathcal{M}(P)$ is a nonempty linear space. A statistical model $\mathcal{M}_{\mathrm{np}}$ is nonparametric if $T_{\mathcal{M}_{\mathrm{np}}}(P) = L^2_0(P)$ for every $P \in \mathcal{M}_{\mathrm{np}}$. 

A parameter (or functional) $\Psi: \mathcal{M} \rightarrow \mathbb{R}$ is pathwise differentiable at $P$ if there exists a bounded linear operator $d\Psi(P): T_\mathcal{M}(P) \rightarrow \mathbb{R}$ such that $d\Psi(P)[s] = \frac{d}{dt} \Psi(P_t) \big|_{t=0}$ for all regular submodels with score $s \in T_\mathcal{M}(P)$ at $P$. By the Riesz representation theorem, $d\Psi(P)$ can be expressed in terms of an inner product as $s\mapsto \langle s,  D_P \rangle_{L^2(P)}$ for some (mean-zero) element $D_P \in L^2_0(P_0)$ referred to as a gradient. There exists a unique canonical gradient $D_P^*\in T_\mathcal{M}(P)$, referred to as the efficient influence function as its squared $L^2(P)$-norm is the generalized Cramer-Rao (CR) lower bound for estimating $\Psi(P)$ relative to $\mathcal{M}$ \citep{bickel1993efficient}.

For $h \in \mathbb{R}$ and a regular submodel $\{P_t: t \in \mathbb{R}\} \subseteq \mathcal{M}_{\mathrm{np}}$ through $P$ at $t=0$, $P_{0,hn^{-1/2}}$ is a local perturbation (or local alternative) of $P_0$. An estimator $\widehat{\psi}_n$ is regular for a parameter $\Psi$ with respect to the local perturbation $P_{0,hn^{-1/2}}$ if $\sqrt{n}\,\{\widehat{\psi}_n - \Psi(P_{0,hn^{-1/2}})\}$ converges in distribution when sampling from $P_{0,hn^{-1/2}}$ with a limit that does not depend on  $h$. An estimator $\widehat{\psi}_n$ is $P_0$--regular for $\Psi$ over $\mathcal{M}$ if it is regular with respect to all local perturbations of $P_0$ in $\mathcal{M}$, and it is regular for $\Psi$ over $\mathcal{M}$ if it is $P$--regular for each $P \in \mathcal{M}$. An estimator $\widehat{\psi}_n$ is $P_0$--asymptotically linear for a parameter $\Psi$ with influence function $\phi_0$ if $\widehat{\psi}_n = \psi_0 + P_n \phi_0 + o_p(n^{-1/2})$ under sampling from $P_0$, and it is asymptotically linear for $\Psi$ over $\mathcal{M}$ if it is $P$--asymptotically linear under sampling from each $P\in\mathcal{M}$. A $P_0$--asymptotically linear estimator $\widehat{\psi}_n$ is $P_0$--efficient for $\Psi$ with respect to model $\mathcal{M}$ if its influence function, under sampling from $P_0$, equals the efficient influence function of $\Psi$ at $P_0$.  An estimator is efficient for $\Psi$ with respect to $\mathcal{M}$ if it is $P$--efficient for each $P \in \mathcal{M}$. Similarly, an estimator $\widehat{\psi}_n$ is $P_0$--superefficient for $\Psi$ relative to $\mathcal{M}$ if its limiting variance is, under sampling from $P_0$, smaller than the corresponding CR lower bound of $\Psi$ at $P_0$.

\subsection{General framework and overview of results}
\label{section::AMLEGeneral1}

In this section, we introduce our general framework and summarize the main results developed in this paper. In the next section, we apply this framework to adaptive inference on linear functionals of a regression function and introduce the concrete estimators studied throughout the paper.

We observe \(n\) independent and identically distributed observations \(O_1, O_2, \dots, O_n\) drawn from a distribution \(P_0\) known to lie in a convex nonparametric model \(\mathcal{M}_{\mathrm{np}}\). We assume that \(P_0 \in \mathcal{M} \subseteq \mathcal{M}_{\mathrm{np}}\), where \(\mathcal{M}\) is a prespecified submodel encoding any prior knowledge about the data-generating process. Our goal is to perform inference on the feature \(\psi_0 := \Psi(P_0)\), where \(\Psi : \mathcal{M}_{\mathrm{np}} \to \mathbb{R}\) is a given real-valued target parameter defined on the nonparametric model. For notational convenience, we write \(S_0\) for any summary \(S_{P_0}\) of \(P_0\).

Let \(\mathcal{M}_n \subseteq \mathcal{M}\) be a working statistical model learned through a data-driven model-selection procedure. A key component of our framework is the assumption that there exists an unknown, fixed submodel \(\mathcal{M}_0 \subseteq \mathcal{M}\) such that \(\mathcal{M}_n\) approximates \(\mathcal{M}_0\), with vanishing error in an appropriate sense, as \(n \to \infty\). Intuitively, \(\mathcal{M}_0\) may be viewed as the model that would be learned with infinite data under the selection procedure at \(P_0\). We assume that \(P_0 \in \mathcal{M}_0\), although \(P_0\) need not belong to \(\mathcal{M}_n\) for any finite \(n\), and that \(\mathcal{M}_0\) is smooth in the sense of Section~\ref{sec::prelim}, thereby excluding degenerate cases such as \(\mathcal{M}_0 = \{P_0\}\). Since \(\mathcal{M}_0\) generally depends on unknown structural features of \(P_0\), such as sparsity, smoothness, or a low-dimensional representation, we refer to \(\mathcal{M}_0\) as the \emph{oracle submodel}. Our theory does not require \(\mathcal{M}_0\) to be known explicitly, and the precise sense in which \(\mathcal{M}_n\) approximates \(\mathcal{M}_0\) is formalized through the high-level conditions summarized below. For example, if \(\mathcal{M}_n\) is determined by an estimated dimension reduction \(\phi_n(Z)\), then \(\mathcal{M}_0\) may correspond to the model induced by a population-level reduction \(\phi_0(Z)\). If the procedure performs variable or basis-function selection, such as the Lasso, then \(\mathcal{M}_0\) may correspond to the limiting selected subset. Similarly, if \(\mathcal{M}_n = \mathcal{M}_{k(n)}\) is chosen by cross-validation from a sieve \(\mathcal{M}_1 \subseteq \mathcal{M}_2 \subseteq \cdots \subseteq \mathcal{M}_\infty := \mathcal{M}\), then the oracle model may be taken to be \(\mathcal{M}_{k_0}\), where \(k_0 := \lim_{n \to \infty} k(n)\), provided this limit exists and is nonrandom.

\begin{remark}
Our framework does not require the oracle submodel \(\mathcal{M}_0\) to be minimal in any sense. In particular, when \(\mathcal{M}_n\) is based on variable selection, \(\mathcal{M}_0\) need not coincide with the exact support or the set of truly active variables. More generally, smaller working and oracle submodels correspond to greater potential efficiency gains.
\end{remark}

The target of inference in the ADML framework is an \emph{oracle projection parameter} induced by an unknown oracle submodel \(\mathcal{M}_0 \subseteq \mathcal{M}\) containing \(P_0\). Specifically, ADML estimators implicitly target the oracle parameter \(\Psi_0 : \mathcal{M}_{\mathrm{np}} \to \mathbb{R}\) defined by
\begin{equation}
\Psi_0 := \Psi \circ \Pi_0,
\label{eqn::oracleParam}
\end{equation}
where \(\Pi_0\) denotes a loss-based projection onto \(\mathcal{M}_0\), that is, \(P \mapsto \Pi_0 P \in \arg\min_{Q \in \mathcal{M}_0} P\ell(\cdot,Q)\). For example, \(\ell\) may be taken as the negative log-likelihood loss, in which case \(\Pi_0\) is the log-likelihood projection \(P \mapsto \Pi_0 P := \arg\max_{Q \in \mathcal{M}_0} P\!\left(\log \frac{dQ}{d\mu}\right)\). Under mild conditions, \(\Psi_0\) is pathwise differentiable at \(P_0\) with nonparametric efficient influence function (EIF) \(D_{0,P_0}\), and is therefore amenable to \(\sqrt{n}\)-consistent estimation \citep{bickel1993efficient}. Section~\ref{section::oracleparameter} formally defines \(\Psi_0\) and characterizes its EIF.

A key property of the oracle parameter \(\Psi_0\) is that it agrees with the original parameter \(\Psi\) at the truth $P_0$ while generally having a smaller efficiency bound. Specifically, \(\Psi_0(P)=\Psi(P)\) for all \(P \in \mathcal{M}_0\), and hence \(\Psi_0(P_0)=\Psi(P_0)\) since \(P_0 \in \mathcal{M}_0\). At the same time, when the tangent space of \(\mathcal{M}_0\) at \(P_0\) is smaller than that of \(\mathcal{M}_{\mathrm{np}}\), the nonparametric efficiency bound \(\operatorname{var}_0\{D_{0,P_0}(O)\}\) for \(\Psi_0\) at \(P_0\) is often smaller than that of \(\Psi\). Consequently, a \(P_0\)-efficient estimator of \(\Psi_0\) is typically \(P_0\)-superefficient for \(\Psi\) relative to the prespecified model \(\mathcal{M}\) and nonparametric model \(\mathcal{M}_{\mathrm{np}}\). Its limiting variance therefore depends on the complexity of the oracle submodel \(\mathcal{M}_0\), allowing it to adapt to the structure of \(P_0\). For log-likelihood loss, this bound coincides with the semiparametric efficiency bound for \(\Psi\) under the oracle model \(\mathcal{M}_0\). For other losses, the two bounds need not coincide exactly, although the bound for \(\Psi_0\) is still typically smaller than that for \(\Psi\).

 Our proposed ADML framework  suggests obtaining $P_0$--efficient inference for $\Psi_0$ --- and thereby $P_0$--superefficient inference for $\Psi$ --- by constructing debiased estimators for the (data-adaptive) working projection parameter $\Psi_n: \mathcal{M}_{\mathrm{np}} \rightarrow \mathbb{R}$ defined as the composition 
\begin{equation}
   \Psi_n:= \Psi \circ \Pi_n
   \label{eqn::dataadaptParam}
\end{equation}evaluated pointwise as $P\mapsto \Psi(\Pi_n P)$, 
where $\Pi_n P \in \argmin_{Q \in \mathcal{M}_n} P \ell(\cdot, Q)$ is a projection of $P \in \mathcal{M}_{\mathrm{np}}$ onto the data-dependent working model $\mathcal{M}_n$. Formally, an ADML estimator \(\widehat{\psi}_n\) is an estimator that satisfies the asymptotic expansion
\[
\widehat{\psi}_n = \Psi_n(P_0) + P_n D_{n,P_0} + o_p(n^{-1/2})\ ,
\]
where \(D_{n,P_0}\) is the nonparametric efficient influence function for \(\Psi_n\). Such an estimator of $\Psi_n(P_0)$ can be constructed using standard debiasing techniques, as shown in \cite{dataAdaptTargetParam}. In particular, given an estimator \(\widehat{P}_n \in \mathcal{M}_n\) of \(P_0\), the one-step debiased estimator \(\Psi(\widehat{P}_n) + P_n D_{n,\widehat{P}_n}\) of $\Psi_n(P_0)$ satisfies this expansion under mild conditions \citep{bickel1993efficient}. If \(\Psi_n\) were fixed across \(n\), then this expansion would simply require \(\widehat{\psi}_n\) to be a nonparametric efficient estimator of \(\Psi_n\).

\begin{example}[Adaptive plug-in M-estimation]
\label{example::mestimator}

Suppose that the learned model \(\mathcal{M}_n\) is finite-dimensional. Then the plug-in M-estimator \(\widehat{\psi}_n := \Psi(\widehat{P}_n)\) for \(\Psi_n(P_0)\), where $\widehat{P}_n \in \argmin_{Q \in \mathcal{M}_n} \sum_{i=1}^n \ell(Z_i, Q),$ is an example of an ADML estimator. In particular, the first-order optimality conditions for the M-estimator imply that \(P_n D_{n,\widehat{P}_n} = 0\), so no explicit bias correction is needed for \(\Psi_n\).
\end{example}

The validity of ADML relies on the key result of Section~\ref{section::modelapprox}, which characterizes the model approximation error \(\Psi_n(P_0)-\Psi_0(P_0)\) and gives conditions under which it is \(o_p(n^{-1/2})\). In particular, replacing the oracle model \(\mathcal{M}_0\) with the working model \(\mathcal{M}_n\) induces only second-order error. At a high level, this error is negligible when (i) the loss-based projection \(P_{n,0}:=\Pi_n P_0\) of \(P_0\) onto \(\mathcal{M}_n\) approximates \(P_0\) well, and (ii) \(\mathcal{M}_n\) has sufficiently rich local tangent structure to approximate that of \(\mathcal{M}_0\). Specifically, we show that, to leading order,
\begin{align*}
    \left|\Psi_n(P_0)-\Psi_0(P_0)\right|
    &\lesssim
    H(P_{n,0}, P_0)\,
    \left\|\bar{D}_{n,0,P_{n,0}} - D_{n,0,P_{n,0}}\right\|_{P_{n,0}+P_0},
\end{align*}
where \(H(P_{n,0}, P_0)\) is the Hellinger distance between \(P_{n,0}\) and \(P_0\), \(D_{n,0,P_{n,0}} \in L_0^2(P_{n,0})\) is an appropriate influence function, and \(\bar D_{n,0,P_{n,0}} \in L_0^2(P_{n,0})\) is its projection onto an approximating subspace of the tangent space \(T_{\mathcal{M}_{\mathrm{np}}}(P_{n,0}) = L_0^2(P_{n,0})\). Thus, although ADML explicitly debiases only the working estimand \(\Psi_n(P_0)\), this suffices to remove the model-selection bias to first order, even when \(\mathcal{M}_n\) converges slowly to \(\mathcal{M}_0\).

In this paper, we establish that, under appropriate conditions, ADML estimators are regular, asymptotically linear, and nonparametrically efficient for the oracle parameter \(\Psi_0\) at \(P_0\). Consequently, they provide locally uniformly valid inference for \(\Psi_0\) in the sense of \cite{buhlmann1999efficient}, even under sampling from locally misspecified perturbations of \(P_0\) lying outside the oracle submodel \(\mathcal{M}_0\). Thus, in a local asymptotic sense, there is no penalty for data-driven model selection relative to knowing the oracle submodel or oracle parameter in advance. Furthermore, these properties of the oracle parameter \(\Psi_0\) imply the following desirable properties for the target parameter \(\Psi\):
\begin{enumerate}[label=\roman*., ref=\roman*]
\item asymptotic linearity for \(\Psi\), with influence function at \(P_0\) given by the \(P_0\)-efficient influence function of \(\Psi_0\);
\item \(P_0\)-regularity for \(\Psi\) with respect to any local perturbation of \(P_0\) within the oracle submodel \(\mathcal{M}_0\);
\item asymptotic \(P_0\)-efficiency for \(\Psi\) relative to the oracle submodel \(\mathcal{M}_0\), for suitable \(\Pi_0\).
\end{enumerate}
Consequently, in practice, if \(\widehat P_n\) is a sufficiently regular estimator of \(P_0\), the ADML estimator \(\widehat{\psi}_n\) is approximately normal under \(P_0\) with mean \(\psi_0\) and variance \(n^{-1}\sigma_n^2\), where \(\sigma_n^2 := \frac{1}{n}\sum_{i=1}^n D_{n,\widehat P_n}(O_i)^2\) is an influence-function-based variance estimator. When \(\widehat{\psi}_n\) is the adaptive plug-in M-estimator of Example \ref{example::mestimator}, \(\sigma_n^2\) reduces to the standard model-robust sandwich variance estimator. An approximate \((1-\alpha)\%\) confidence interval for \(\psi_0\) is therefore
\[
\mathcal{I}_n(\alpha)
:=
\left(
\widehat{\psi}_n - q_\alpha \sigma_n n^{-1/2},
\widehat{\psi}_n + q_\alpha \sigma_n n^{-1/2}
\right),
\]
where \(q_\alpha\) is the \((1-\alpha/2)\)-quantile of the standard normal distribution. This interval is \(P_0\)-locally uniformly valid and nonparametric \(P_0\)-efficient for the projection oracle parameter \(\Psi_0\), uniformly over local perturbations \(P_{0,hn^{-1/2}} \in \mathcal{M}_{\mathrm{np}}\). In particular, it is asymptotically valid for \(\Psi(P_{0,hn^{-1/2}})\) uniformly over local perturbations \(P_{0,hn^{-1/2}} \in \mathcal{M}_0\). Moreover, when \(\Pi_0\) is the Kullback--Leibler or Hellinger projection, the width of \(\mathcal{I}_n(\alpha)\) is \(P_0\)-efficient and locally asymptotically minimax over \(\mathcal{M}_0\) for the target parameter \(\Psi\).

 \subsection{Specialization to linear functionals of the outcome regression}
\label{section::ADMLlinear}

We now specialize the ADML framework to continuous linear functionals of the outcome regression, a broad class that includes many causal estimands, such as the average treatment effect. As discussed in the introduction, adaptivity is particularly important in these settings because standard nonparametric estimators can be highly variable and unstable under limited overlap or positivity violations \citep{petersen2012diagnosing, d2021overlap}, whereas prespecified working models may improve stability at the cost of misspecification bias \citep{imbensOverlapEstimand2006, li2019overlapWeights}. Our goal in this section is to show how ADML learns this structure from the data and to relate the resulting estimators to several existing approaches based on adaptive confounder selection and data-adaptive feature representations.

Consider observations \(O=(W,A,Y)\), where \(W \in \mathbb{R}^d\) denotes covariates, \(A \in \{0,1\}\) a binary treatment, and \(Y \in \mathbb{R}\) a bounded outcome. Our target parameter is the linear functional \(\Psi(P):=E_P[m(W,\mu_P)]\) of the outcome regression \(\mu_P(a,w):=E_P(Y \mid A=a,W=w)\), where \(m(w,\mu)\) is a known map that is linear in \(\mu\). For example, the average treatment effect corresponds to \(m(w,\mu)=\mu(1,w)-\mu(0,w)\).

To construct an ADML estimator of \(\psi_0 = E_0[m(W,\mu_0)]\), let \(\mathcal{H}_n \subseteq \mathcal{H} := L^2(P_{0,A,W})\) be a working linear submodel for \(\mu_0\) that approximates an unknown oracle linear submodel \(\mathcal{H}_0 \subseteq \mathcal{H}\) containing \(\mu_0\). Define the associated working and oracle statistical models by
\[
\mathcal{M}_n := \{P \in \mathcal{M} : \mu_P \in \mathcal{H}_n\}, \qquad
\mathcal{M}_0 := \{P \in \mathcal{M} : \mu_P \in \mathcal{H}_0\}.
\]
We consider the working parameter \(\Psi_n(P) := E_P[m(W,\mu_{P,\mathcal{H}_n})]\) and the corresponding oracle parameter \(\Psi_0(P) := E_P[m(W,\mu_{P,\mathcal{H}_0})]\), where \(\mu_{P,\mathcal{H}'}\) denotes the \(L^2(P)\)-projection of \(\mu_P\) onto \(\mathcal{H}'\). For simplicity, we assume that all linear models considered above are closed subspaces of \(L^2(P_0)\) and that \(\|\cdot\|_P\) and \(\|\cdot\|_{P_0}\) are equivalent for all \(P \in \mathcal{M}_{\mathrm{np}}\), so that these projections are well defined.

An ADML estimator of \(\psi_0\) is obtained by constructing a DML estimator of \(\Psi_n(P_0)\), for example using autoDML \citep{chernozhukov2022automatic}, as described in Section~\ref{adml::linearfunc}. When \(\mathcal{H}_n\) is finite-dimensional, a simple example is the adaptive plug-in regression estimator
\begin{equation}
\label{eqn::pluginadml}
\widehat{\psi}_n=\frac{1}{n}\sum_{i=1}^n m(W_i,\mu_n), \qquad
\mu_n \in \arg\min_{\mu \in \mathcal{H}_n}\sum_{i=1}^n \{Y_i-\mu(A_i,W_i)\}^2,
\end{equation}
a special case of Example \ref{example::mestimator}. We now present three examples of ADML estimators that serve as running examples throughout the paper; their theory is developed in Section~\ref{adml::linearfunc}.

Our first example considers model selection for the CATE within a partially linear regression (PLR) model \citep{robinson1988root, imbensOverlapEstimand2006, li2019overlapWeights}. Such models are attractive because they can yield substantially smaller semiparametric efficiency bounds than the fully nonparametric model and may be more robust under limited overlap. In particular, \cite{imbensOverlapEstimand2006} proposed a constant-CATE working model, and \cite{li2019overlapWeights} argued that such structure can weaken, and in some settings effectively eliminate, standard overlap requirements. We argue that this model can often be improved by learning a simple CATE model from the data, thereby retaining favorable efficiency properties while reducing misspecification bias.  In what follows, we write the conditional average treatment effect (CATE) as \(\tau_P(w):=\mu_P(1,w)-\mu_P(0,w)\), the propensity score as \(\pi_P(w):=P(A=1 \mid W=w)\), and the marginalized regression as \(m_P(w):=E_P(Y \mid W=w)\). We also recall Robinson's parameterization of the outcome regression, $\mu_P(a,w) = m_P(w) + (a - \pi_P(w))\tau_P(w),$
from \cite{robinson1988root}.

\begin{example}[Semiparametric model selection for the CATE]
\label{example::CATE}
The ATE depends on the outcome regression \(\mu_0\) only through the conditional average treatment effect (CATE) \(\tau_0 := \mu_0(1,\cdot)-\mu_0(0,\cdot)\). This motivates ADML estimators that model the CATE directly through PLR models. Specifically, suppose
\[
\mathcal{H}_n
=
\{(a,w) \mapsto \mu(w)+a\tau(w) : \mu \in L^2(P_{0,W}), \ \tau \in \mathcal{T}_n\},
\]
where \(\mathcal{T}_n\) is a finite-dimensional linear model for the CATE learned from the data. Given nuisance estimators \(m_n\) and \(\pi_n\) of \(m_0\) and \(\pi_0\), respectively, a semiparametric ADML estimator of the ATE is
\[
\widehat{\psi}_n := \frac{1}{n}\sum_{i=1}^n \tau_n(W_i),
\qquad
\tau_n
:=
\argmin_{\tau \in \mathcal{T}_n}
\sum_{i=1}^n
\left[
Y_i - m_n(W_i) - \{A_i-\pi_n(W_i)\}\tau(W_i)
\right]^2,
\]
where \(\tau_n\) is the well-studied \(R\)-learner for the CATE \citep{nie2021quasi}, based on Robinson's parameterization of the outcome regression \citep{robinson1988root}. For example, \(\mathcal{T}_n\) may be selected by cross-validation from a nested sequence of models of increasing complexity, \(\mathcal{T}_1 \subseteq \mathcal{T}_2 \subseteq \cdots \subseteq \mathcal{T}\), where \(\mathcal{T}_1\) is the intercept model and \(\mathcal{T}\) is the nonparametric model \citep{shen1997methods, spnpsieve, CTMLE}. Alternatively, \(\mathcal{T}_n\) may be obtained via variable selection using the Lasso \citep{Tibshirani94regressionshrinkage, AdaptLassoOracle, moosavi2023costs} or via nonparametric basis selection using the Highly Adaptive Lasso (HAL) \citep{vanderlaanGenerlaTMLEFIRST, HAL2016, HALoutcomeAdapt}. Consistent variable and model selection in related settings has been studied by \cite{bunea2004consistent, su2014variable, li2008variable, claeskens2007asymptotic, zhao2017selective}. \qed
\end{example}

Another important class of ADML estimators uses learned low-dimensional feature representations to address limited overlap in high-dimensional causal inference problems \citep{d2021overlap}. Specifically, let \(\phi_n:\mathcal{W}\times\mathcal{A}\to\mathbb{R}^m\) be a data-adaptive feature map, and define \(\mathcal{H}_n:=\mathcal{H}_{\phi_n}\) and \(\mathcal{H}_0:=\mathcal{H}_{\phi_0}\), where \(\mathcal{H}_{\phi}:=\{f\circ\phi:f:\mathbb{R}^m\to\mathbb{R}\}\). In causal effect estimation, \citet{d2021deconfounding} proposed learning deconfounding-score representations from the data to adjust for confounding when covariate overlap is limited. More generally, \(\phi_n\) may be constructed, for example, by one-hot encoding random forest leaves \citep[Section~3.1]{he2014practical}, using intermediate neural network features such as autoencoder embeddings \citep{ng2011sparse}, or combining multiple estimated regression models in a Super Learner ensemble \citep{wang2023super}. A particularly important special case arises when the feature map is taken to be the estimated outcome regression itself, that is, \(\phi_n=\mu_n\) and \(\phi_0=\mu_0\). In this case, \(\mathcal{H}_n\) and \(\mathcal{H}_0\) are the classes of all transformations of \(\mu_n\) and \(\mu_0\), respectively. Superefficient estimators based on this representation were proposed and studied by \citet{ATEsupereff}.

Our second example extends the superefficient estimator of \citet{ATEsupereff} for the ATE to general linear functionals. Rather than using the TMLE-style update from that work, we build on recent work combining calibration with DML \citep{van2024automatic,van2025automatic} and use isotonic regression for debiasing. Isotonic regression \citep{barlow1972isotonic} is a tuning-free method for monotone-constrained estimation that is widely used to calibrate predictive models in machine learning \citep{zadrozny2001obtaining,niculescu2005predicting,van2024self,van2025generalized}. We show that, in the DML setting, calibration can also be used to construct superefficient plug-in estimators.

\begin{example}[Plug-in regression with isotonic calibration]
\label{example::map}
We consider the isotonic-calibrated plug-in estimator
\[
\widehat{\psi}_n:=\frac{1}{n}\sum_{i=1}^n m(W_i,\mu_n^*)  \quad \text{where} \quad\mu_n^*:=f_n\circ\mu_n,
\]
where \(\mu_n\) is a preliminary estimator of the outcome regression and \(f_n\) is obtained by isotonic regression:
\[
f_n:=\argmin_{f\in\mathcal{F}_{iso}}\sum_{i=1}^n\{Y_i-f(\mu_n(A_i,W_i))\}^2,
\]
with \(\mathcal{F}_{iso}\) denoting the class of monotone nondecreasing functions. We show that this is an ADML estimator with working model \(\mathcal{H}_n=\{f\circ\mu_n^*: f:\mathbb{R}\to\mathbb{R}\}\) and oracle model \(\mathcal{H}_0=\{f\circ\mu_0: f:\mathbb{R}\to\mathbb{R}\}\). As shown by \citet{ATEsupereff}, a key advantage of this approach is that identification, estimator variance, and stable estimation depend primarily on treatment overlap across strata of \(\{\mu_0(a,W): a \in \mathcal{A}\}\), rather than across strata of the potentially high-dimensional covariate vector \(W\).\qed

\end{example}

Our final example combines the previous two ideas to construct a novel ADML estimator of the ATE by calibrating the CATE within a semiparametric regression model. It extends the estimators in Example~\ref{example::map} and \citet{ATEsupereff} to partially linear models.

\begin{example}[Plug-in R-learner with isotonic calibration]
\label{example::catecal}
Let \(\tau_n\) be a preliminary user-supplied CATE estimator, such as $\tau_n(w) := \mu_n(1,w) - \mu_n(0,w)$. Define the ADML estimator by $\widehat{\psi}_n:=\frac{1}{n}\sum_{i=1}^n \tau_n^*(W_i),$
where \(\tau_n^*\) is the calibrated \(R\)-learner obtained via \emph{causal isotonic calibration} \citep{van2023causal}:
\[
\tau_n^* = f_n \circ \tau_n, \qquad
f_n := \argmin_{f \in \mathcal{F}_{iso}} \sum_{i=1}^n \{Y_i - m_n(W_i) - (A_i - \pi_n(W_i))f(\tau_n(W_i))\}^2.
\]
This estimator corresponds to the partially linear models \(\mathcal{H}_n := \mathcal{H}_{\tau_n^*}\) and \(\mathcal{H}_0 := \mathcal{H}_{\tau_0}\), where
\[
\mathcal{H}_{\tau} := \{(a,w) \mapsto \mu(w) + a\,f(\tau(w)) : \mu \in L^2(P_{0,W}),\ f:\mathbb{R}\to\mathbb{R}\}.
\]
Thus, \(\mathcal{H}_n\) uses \(\tau_n^*(W)\) as a data-adaptive low-dimensional summary of how the CATE \(\tau_0\) varies with \(W\).

As we will show in Section \ref{adml::linearfunc}, this estimator is nonparametrically valid, adapts to treatment effect heterogeneity, and retains the favorable efficiency properties of semiparametric estimators based on constant or parametric CATE models \citep{robinson1988root, imbensOverlapEstimand2006, li2019overlapWeights}. In particular, identification and stable estimation depend primarily on treatment overlap across strata of the CATE rather than across strata of the covariates. Moreover, if the true CATE is constant, then \(\mathcal{H}_0\) reduces to the homogeneous partially linear model and \(\Psi_0\) reduces to the overlap-weighted ATE. In that case, the ADML estimator is asymptotically equivalent to the prespecified estimators of \cite{imbensOverlapEstimand2006, li2019overlapWeights}. It therefore improves on these estimators by remaining valid more generally while retaining their asymptotic behavior when the homogeneous model is correct. \qed 
\end{example}



\section{Oracle projection parameter and statistical efficiency}
\label{section::oracleparameter}

\subsection{Pathwise differentiability and the efficient influence function: general case}

In this section, we formally define the target of inference in ADML, namely the oracle parameter \(\Psi_0 = \Psi \circ \Pi_0\), based on the loss-based projection operator \(\Pi_0\), and characterize its efficient influence function. We then apply this theory to the linear functional settings introduced in Section~\ref{section::ADMLlinear}. Although we focus on the projection parameter induced by the oracle model \(\mathcal{M}_0\), the results continue to hold for projection parameters onto any model \(\mathcal{M}'\) in place of \(\mathcal{M}_0\).

Let $\ell:\mathbb{R}^d \times \mathcal{M}_{\mathrm{np}}\rightarrow\mathbb{R}$ be a loss function satisfying $P \in \argmin_{ Q \in \mathcal{M}_{\mathrm{np}}} P \ell(\cdot, Q)$ for each $P \in \mathcal{M}_{\mathrm{np}}$. We define the (possibly non-unique) loss-based projection operator $\Pi_0: \mathcal{M}_{\mathrm{np}}\rightarrow \mathcal{M}_0$ as any map $P \mapsto \Pi_0P$ whose image $\Pi_0P$ is contained in the solution set $\argmin_{Q \in \mathcal{M}_0} P \ell(\cdot, Q)$. Informally, the projection $\Pi_0  $ maps a given distribution $P \in \mathcal{M}_{\mathrm{np}}$ to one of its best approximations in $\mathcal{M}_0$ under the risk $Q \mapsto P \ell(\cdot, Q)$. The oracle projection parameter $\Psi_0: \mathcal{M}_{\mathrm{np}} \rightarrow \mathbb{R}$, formally defined as $\Psi_0:=\Psi \circ \Pi_0$, applies the oracle projection operator $\Pi_0$ before evaluating the target parameter mapping $\Psi$.  If $\Pi_0$ is the loglikelihood projection and $\mathcal{M}_0$ is a fixed parametric model, $\Psi_0(P)$ corresponds to the $P$-limit that a maximum likelihood estimator (MLE) would converge to, even if the MLE is computed under an incorrectly specified model \citep{white1982MLErobust, freedman2006MLErobust}.

The following theorem characterizes the efficient influence function \(D_{0,P_0}\) in terms of the oracle model and the loss function. To state it, we impose high-level smoothness conditions on \(\ell\) and \(\Psi\) over \(\mathcal{M}_0\). Unless otherwise noted, these conditions are assumed to hold at an arbitrary \(P \in \mathcal{M}\), specified in the theorem.

\begin{enumerate}[label=(A\arabic*), ref=A\arabic*,series=cond]
\item \textit{Risk function is twice differentiable:} \label{cond::losssmooth}
For all \(P \in \mathcal{M}_{\mathrm{np}}\), all \(Q \in \mathcal{M}_0\), and all regular paths \(\{Q_t : t \in \mathbb{R}\} \subset \mathcal{M}_0\) through \(Q\) with score \(s \in T_{\mathcal{M}_0}(Q)\), there exist:
\begin{enumerate}
    \item[(i)] a first G\^{a}teaux derivative \(\dot{\ell}_{Q} : \mathcal{O} \times T_{\mathcal{M}_0}(Q) \to L^2(P)\) such that $\dot{\ell}_{Q}(\cdot,s) = \frac{d}{dt} \ell(\cdot, Q_t)\Big|_{t=0};$
    \item[(ii)] a second G\^{a}teaux derivative \(P\ddot{\ell}_{Q} : T_{\mathcal{M}_0}(Q) \times T_{\mathcal{M}_0}(Q) \to \mathbb{R}\) such that
    $P\ddot{\ell}_{Q}(s,v) = \frac{d}{dt} P\dot{\ell}_{Q_t}(v)\Big|_{t=0}.$
\end{enumerate}
\item \textit{Hessian of risk is nonnegative:} \label{cond::innerproductloss} $P \ddot{\ell}_{\Pi_0 P}$ defines a positive semidefinite inner product on $T_{\mathcal{M}_0}(\Pi_0 P)$.
\item \textit{Hessian-based pathwise differentiability under oracle submodel:} \label{cond::oraclepathwise} 
$\Psi$ is pathwise differentiable at $\Pi_0 P$ on $\mathcal{M}_0$ in the sense that 
$d\Psi(\Pi_0 P): T_{\mathcal{M}}(\Pi_0 P) \rightarrow \mathbb{R}$ is a bounded linear operator with respect to the Hessian inner product $P\ddot{\ell}_{\Pi_0 P}$.
\item \textit{Existence, invariance, and smoothness of the projection:} \label{cond::invariance}    \label{cond::smoothprojection}
For all \(P \in \mathcal{M}_{\mathrm{np}}\), the following hold:
\begin{enumerate}[label=(\roman*), ref=\ref{cond::invariance}.\roman* ]
    \item \label{cond::invariance-i}
    The set \(\argmin_{Q \in \mathcal{M}_0} P\ell(\cdot, Q)\) is nonempty, and \(\Psi(Q) = \Psi(Q')\) for all
    \(Q, Q' \in \argmin_{Q \in \mathcal{M}_0} P\ell(\cdot, Q)\).

    \item \label{cond::invariance-ii}
    For each \(Q, Q' \in \argmin_{Q \in \mathcal{M}_0} P\ell(\cdot, Q)\), we have
    \(\dot{\ell}_{Q} = \dot{\ell}_{Q'}\) and \(P\ddot{\ell}_{Q} = P\ddot{\ell}_{Q'}\) \(P\)-almost everywhere.

    \item \label{cond::smoothprojection2}
    For any regular submodel \(t \mapsto P_t\) with \(P_t|_{t=0} = P\), there exists a regular submodel
    \(t \mapsto Q_{P_t} \in \mathcal{M}_0\) such that \(Q_{P_t} \in \argmin_{Q \in \mathcal{M}_0} P_t\ell(\cdot, Q)\)
    for all sufficiently small \(t\).
\end{enumerate}

\end{enumerate}

The efficient influence function (EIF) in the following theorem is expressed in terms of a so-called \textit{Hessian Riesz representer} \(s_{0,P}\) of the pathwise derivative of \(\Psi\) along the oracle submodel \(\mathcal{M}_0\). We refer to the semidefinite inner product \(P\ddot{\ell}_{\Pi_0 P}\) as the \textit{Hessian} inner product; it is also known as the generalized Fisher inner product (Section 4.1.1 of \citealp{chen2014sieve}; \citealp{sieveOneStepPlugin}; and \citealp{van2025automatic2}). We show that there exists a potentially nonunique element \(s_{0,P}\) in the completion $\overline{T}_{\mathcal{M}_0}(\Pi_0 P)$ of \(T_{\mathcal{M}_0}(\Pi_0 P)\) under \(P\ddot{\ell}_{\Pi_0 P}\) such that, for any smooth submodel \((Q_t:t)\subset \mathcal{M}_0\) with \(Q_t|_{t=0}=\Pi_0 P\) and score \(s\),
\begin{equation}
\frac{d}{dt}\Psi(Q_t)\Big|_{t=0}=P\ddot{\ell}_{\Pi_0 P}(s,s_{0,P}).
\label{eqn::autodml}
\end{equation}
In principle, \(s_{0,P}\) can be directly learned by solving $ s_{0,P} \in \argmin_{s \in \overline{T}_{\mathcal{M}_0}(\Pi_0 P)}
P\ddot{\ell}_{\Pi_0 P}(s,s)-2\,d\Psi(\Pi_0 P)(s),$
analogously to Riesz regression in autoDML \citep{chernozhukov2022automatic, van2025automatic2}.

\begin{theorem}[EIF of oracle parameter]
    \label{theorem::EIFmain} 
    Suppose that Conditions~\ref{cond::losssmooth}--\ref{cond::smoothprojection} hold at \(P \in \mathcal{M}_{\mathrm{np}}\). Then:
    \begin{enumerate}
        \item[(i)] The functional \(\Psi_0 : \mathcal{M}_{\mathrm{np}} \to \mathbb{R}\) is pathwise differentiable at \(P\) with EIF \(D_{0,P}\) given by $D_{0,P} = -\dot{\ell}_{\Pi_0 P}(s_{0,P}),$
        where \(s_{0,P}\) is any Hessian Riesz representer satisfying \eqref{eqn::autodml}.

        \item[(ii)] The EIF \(D_{0,P}\) is unique and does not depend on the choice of representer \(s_{0,P}\).

\item[(iii)] If \(\ell\) is the negative log-likelihood loss, then \(D_{0,P}\) equals the EIF of \(\Psi\) restricted to the oracle submodel \(\mathcal{M}_0\), and \(\|D_{0,P}\|_{P}^2\) is the corresponding \(P\)-efficiency (minimum variance) bound under \(\mathcal{M}_0\).

    \end{enumerate}
\end{theorem}

Theorem~\ref{theorem::EIFmain} shows that the efficiency bound for \(\Psi_0\), given by \(\|D_{0,P}\|_{P}^2\), depends on both the oracle model \(\mathcal{M}_0\) and the loss-based projection \(\Pi_0\). Specifically, the EIF of \(\Psi_0\) is \(-\dot{\ell}_{\Pi_0 P}(s_{0,P})\), the derivative of the negative loss in the direction of the Hessian Riesz representer \(s_{0,P}\). Under log-likelihood loss, the EIF of \(\Psi_0\) at \(P_0 \in \mathcal{M}_0\) is \(s_{0,P_0}\), which also equals the EIF of \(\Psi:\mathcal{M}_0 \to \mathbb{R}\) at \(P_0\). Thus, a nonparametric efficient estimator of \(\Psi_0\) at \(P_0\) is also semiparametric efficient in the oracle model \(\mathcal{M}_0\).

The conditions of Theorem~\ref{theorem::EIFmain} are mild. The main requirement is \ref{cond::oraclepathwise}, which assumes that \(\Psi\) is pathwise differentiable at \(\Pi_0 P\) with respect to the Hessian inner product \(P\ddot{\ell}_{\Pi_0 P}\). Notably, this requires differentiability only along the oracle submodel \(\mathcal{M}_0\), not over the full nonparametric model \(\mathcal{M}_{\mathrm{np}}\); for log-likelihood loss, it reduces to the standard \(L^2\) pathwise differentiability condition on \(\mathcal{M}_0\). The remaining conditions are typically satisfied for smooth losses. Condition~\ref{cond::losssmooth} requires first- and second-order Gâteaux differentiability along regular paths. Condition~\ref{cond::innerproductloss} requires the Hessian to define a positive semidefinite inner product on the tangent space. Conditions~\ref{cond::invariance-i}-\ref{cond::invariance-ii} addresses possible nonuniqueness of the loss-based projection by requiring that \(\Psi_0\), \(\dot{\ell}_Q\), and \(P\ddot{\ell}_Q\) be invariant across solutions. Finally, \ref{cond::smoothprojection2} ensures that \(\Pi_0\) is smooth in the sense that it maps regular submodels to regular submodels.

\subsection{Specialization to linear functionals of the outcome regression}
\label{section::EIFlinear}

We now apply Theorem   \ref{theorem::EIFmain}  to the linear functional setup introduced in Section~\ref{section::ADMLlinear}. Recall that the oracle parameter is \(\Psi_0(P) := E_P[m(W, \mu_{P,\mathcal{H}_0})]\), where \(\mu_{P,\mathcal{H}_0} := \argmin_{\mu \in \mathcal{H}_0} E_P[\{\mu_P(A,W) - \mu(A,W)\}^2]\) denotes the \(L^2(P)\)-projection of \(\mu_P\) onto \(\mathcal{H}_0\). Equivalently, \(\Psi_0\) can be written as \(\Psi_0 = \Psi \circ \Pi_0\), where \(P \mapsto \Pi_0 P\) is any loss-based projection satisfying \(\Pi_0 P \in \argmin_{Q \in \mathcal{M}_0} P\ell(\cdot,Q)\), with \(\ell(O,Q) := \frac{1}{2}\{Y - \mu_Q(A,W)\}^2 - \log \frac{dQ_X}{d\mu_X}(X)\) almost surely.

The conditions of Theorem~\ref{theorem::EIFmain} are straightforward to verify in this setting. Conditions~\ref{cond::invariance-i}--\ref{cond::invariance-ii} hold because both \(\Psi(P)\) and \(\ell(\cdot,P)\) depend on \(P\) only through the covariate distribution \(P_W\) and the outcome regression \(\mu_P\). Conditions~\ref{cond::losssmooth}, \ref{cond::innerproductloss}, and \ref{cond::smoothprojection2} follow from the smoothness of the least-squares loss and the Hellinger differentiability of the map \(P \mapsto \mu_{P,\mathcal{H}_0}\) \citep{luedtke2024one}. Finally, \ref{cond::oraclepathwise} holds whenever \(\mu \mapsto E_P[m(W,\mu)]\) is a bounded linear functional on \(\mathcal{H}_0\). Theorem~\ref{theorem::EIFmain} therefore yields the following corollary.

\begin{enumerate}[label=(B\arabic*), ref=B\arabic*,series=condlinear]
\item  \textit{Linear functional is continuous:} \label{cond::boundedlinearfun} $\sup_{\mu \in \mathcal{H}_0 \backslash \{0\}} \frac{E_P[m(W,\mu)]}{\|\mu\|_P} < \infty$.
\end{enumerate}

By the Riesz representation theorem, there exists an \(L^2(P)\)-Riesz representer \(\alpha_{P,\mathcal{H}_0}\) of the linear functional \(\mu \mapsto E_P[m(W,\mu)]\) on \(\mathcal{H}_0\), such that \(E_P[m(W,\mu)] = E_P[\alpha_{P,\mathcal{H}_0}(A,W)\mu(A,W)]\) for all \(\mu \in \mathcal{H}_0\). It follows that \(\alpha_{P,\mathcal{H}_0} = \argmin_{\alpha \in \mathcal{H}_0} E_P[\alpha(A,W)^2 - 2m(W,\alpha)]\) \citep{chernozhukov2022automatic}.

\begin{theorem}[Pathwise differentiability of oracle linear functional]
    \label{theorem::boundedlinearEIF} 
    Suppose \ref{cond::boundedlinearfun} holds at $P \in \mathcal{M}_{\mathrm{np}}$. Then, $\Psi_0: \mathcal{M}_{\mathrm{np}} \rightarrow \mathbb{R}$ is $P$-pathwise differentiable with EIF 
    $$D_{0,P}: o \mapsto \alpha_{P,\mathcal{H}_0}(a,w)\{y - \mu_{P,\mathcal{H}_0}(a,w)\} + m(w,\mu_{P,\mathcal{H}_0}) - \Psi_0(P).$$ Moreover, \(D_{0,P} = -\dot{\ell}_{\Pi_0 P}(s_{0,P})\) is of the form given in Theorem~\ref{theorem::EIFmain}, where \(s_{0,P}(z) := v_{P,\mathcal{H}_0}^{-1}(a,w)\alpha_{P,\mathcal{H}_0}(a,w)\{y-\mu_{P,\mathcal{H}_0}(a,w)\} + m(w,\mu_{P,\mathcal{H}_0}) - \Psi_0(P)\) and \(v_{P,\mathcal{H}_0}(a,w) := E_P[\{Y-\mu_{P,\mathcal{H}_0}(A,W)\}^2 \mid A=a, W=w]\).

\end{theorem}

Theorem~\ref{theorem::boundedlinearEIF} recovers known results on pathwise differentiability and EIFs for continuous linear functionals of regression functions \citep{chernoRegRiesz,van2024automatic,van2025automatic}. Here, however, these results follow directly from our general projection-based framework and Theorem~\ref{theorem::EIFmain}, rather than from a separate analysis specialized to linear functionals. We now study the EIFs and efficiency bounds of the oracle parameters in our examples.

\setcounter{example}{1}
\begin{example}[continued]
Recall the partially linear oracle model \(\mathcal{H}_0\) implied by the CATE model \(\mathcal{T}_0\). In Appendix~\ref{app:ADMLpartiallylinear}, we show that, for the ATE parameter,
\[
\alpha_{P,\mathcal{H}_0}(a,w)=\gamma_{P,\mathcal{H}_0}(w)\{a-\pi_P(w)\}, \qquad
\gamma_{P,\mathcal{H}_0}:=\argmin_{\gamma \in \mathcal{T}_0} E_P\!\left[\pi_P(W)\{1-\pi_P(W)\}\gamma(W)^2-2\gamma(W)\right].
\]
Under \ref{cond::boundedlinearfun}, the oracle parameter \(\Psi_0\) is pathwise differentiable at \(P\) with \(P\)-EIF $D_{0,P}(z)=\tau_{\Pi_0 P}(w)-E_P\{\tau_{\Pi_0 P}(W)\}
+\gamma_{P,\mathcal{H}_0}(w)\{a-\pi_P(w)\}\{y-\mu_{P,\mathcal{H}_0}(a,w)\}.$
When \(\pi_P(W)\{1-\pi_P(W)\}>0\) almost surely, \(\gamma_{P,\mathcal{H}_0}\) is the overlap-weighted \(L^2(P)\) projection of \(\{\pi_P(1-\pi_P)\}^{-1}\) onto the working model \(\mathcal{T}_0\). Simpler working models \(\mathcal{T}_0\) typically yield lower efficiency bounds because the inverse propensity weights are replaced by their projection \(\gamma_{P,\mathcal{H}_0}\). At one extreme, if \(\mathcal{T}_0=L^2(P_{0,W})\), then \(\Psi_0=\Psi\) and \(\gamma_{P,\mathcal{H}_0}=\{\pi_P(1-\pi_P)\}^{-1}\), so we recover the nonparametric EIF for the ATE. At the other extreme, if \(\mathcal{T}_0\) is the class of constant functions, then \(\gamma_{P,\mathcal{H}_0}\) reduces to the marginalized weight, $\gamma_{P,\mathcal{H}_0}=\{E_0[\pi_P(W)\{1-\pi_P(W)\}]\}^{-1},$ recovering the EIF for the overlap-weighted ATE \citep{imbensOverlapEstimand2006, li2019overlapWeights}.
\end{example}

\begin{example}[continued]
For the oracle model \(\mathcal{H}_0\) using \(\mu_0(A,W)\) as a dimension reduction of \((A,W)\), \(\alpha_{P,\mathcal{H}_0}(a,w) = E_0[\alpha_{P,\mathcal{H}}(A,W) \mid \mu_0(A,W) = \mu_0(a,w)]\). For the counterfactual mean \(m(w,\mu) := \mu(1,w)\), one can use \((A,\mu_0(1,W))\) as a bivariate reduction of \((A,W)\). The \(P_0\)-EIF of \(\Psi_0\) is then given pointwise by \(\mu_0(1,w) - \psi_0 +\frac{1(a=1)}{\pi_{0,\mu_0}(w)}\{y-\mu_0(1,w)\}\), where \(\pi_{0,\mu_0}(w)=P_0\{A=1 \mid \mu_0(1,W)=\mu_0(1,w)\}\) is the dimension-reduced propensity score. This is precisely the influence function of the nonparametric superefficient estimator proposed by \cite{ATEsupereff}, which uses \(\mu_n\) as a data-adaptive dimension reduction. Our results therefore show that this estimator is nonparametrically efficient for the oracle parameter \(\Psi_0\).
\end{example}

\begin{example}[continued]
Recall that \(\mathcal{H}_0\) is the PLR model induced by the CATE model \(\mathcal{T}_0 = \{f \circ \tau_0 : f\}\), which uses \(\tau_0(W)\) as a dimension reduction. For pathwise differentiability of \(\Psi_0\), it suffices to assume that the reciprocal of \(w \mapsto E_0[(A-\pi_0(W))^2 \mid \tau_0(W)=\tau_0(w)]\) has finite second moment, a mild overlap condition within strata of the CATE. In this case, the \(P_0\)-EIF of \(\Psi_0\) for the ATE parameter is pointwise
\[
\tau_0(w)- E_0[\tau_0(W)]+\gamma_{0,\mathcal{H}_0}(w)\{a-\pi_0(w)\}\{y-\mu_0(a,w)\}, \quad
\gamma_{0,\mathcal{H}_0}(w)=\{E_0[(A-\pi_0(W))^2 \mid \tau_0(W)=\tau_0(w)]\}^{-1}.
\]
Thus, overlap affects the efficiency bound only through this aggregated inverse-weighting factor. When the true CATE \(\tau_0\) is almost surely constant, we recover the overlap-weighted ATE and its EIF.
\end{example}

\section{Model approximation error}
\label{section::modelapprox}

\subsection{Decomposition and overview}

Recall the setup of Section~\ref{section::AMLEGeneral1}. ADML seeks to obtain efficient inference for the oracle parameter \(\Psi_0\) by targeting the data-adaptive parameter \(\Psi_n\), induced by a working model \(\mathcal{M}_n\) learned from the data to approximate the oracle model \(\mathcal{M}_0\). By definition, an ADML estimator is debiased for the data-adaptive target \(\Psi_n(P_0)\). Our analysis begins with the decomposition
\[
\widehat{\psi}_n - \Psi_0(P_0)
=
\underbrace{\widehat{\psi}_n - \Psi_n(P_0)}_{\text{(I)}}
+
\underbrace{\Psi_n(P_0) - \Psi_0(P_0)}_{\text{(II)}}.
\]
Term (I) is the estimation error for the data-adaptive target \(\Psi_n\), whereas term (II) is the model approximation error, which quantifies how well the working model \(\mathcal{M}_n\) approximates the oracle model \(\mathcal{M}_0\). Term (I) concerns inference for a data-adaptive target parameter, as studied by \cite{van2013AdaptTarget, dataAdaptTargetParam}. Under appropriate conditions, given in Section~\ref{section::dataAdaptParam}, term (I) is \(\sqrt{n}\{\widehat{\psi}_n - \Psi_n(P_0)\}\) is asymptotically mean-zero normal and determines the asymptotic distribution of \(\widehat{\psi}_n\).

The key step in establishing the validity of ADML estimators is to show that the model approximation error is negligible, that is, \(\Psi_n(P_0) - \Psi_0(P_0) = o_p(n^{-1/2})\), so that \(\sqrt{n}\{\widehat{\psi}_n - \Psi_0(P_0)\} = \sqrt{n}\{\widehat{\psi}_n - \Psi_n(P_0)\} + o_p(1)\), and hence the ADML estimator is \(n^{1/2}\)-consistent and asymptotically normal. At first glance, one might expect this error to depend at first order on how well the loss-based projection \(\Pi_n P_0 \in \mathcal{M}_n\) approximates the true distribution \(P_0 \in \mathcal{M}_0\), for example in Hellinger distance, which could preclude negligibility. The next subsection shows, however, that the model approximation error is second order. Its rate is instead governed by both how well \(\Pi_n P_0\) approximates \(P_0\) and how well the working model \(\mathcal{M}_n\) locally approximates the oracle model \(\mathcal{M}_0\) around \(P_0\), in a sense made precise below. We then derive a simpler expression for this error in the case of linear functionals of the outcome regression.

\subsection{Second-order expansion of the model approximation error: general case}

To state our main result, we embed the working model \(\mathcal{M}_n\) and the oracle model \(\mathcal{M}_0\) in a common union model \(\mathcal{M}_{n,0}\). For example, in Section~\ref{section::ADMLlinear}, we may take \(\mathcal{M}_{n,0} = \mathcal{M}_{\mathcal{H}_{n,0}}\), where \(\mathcal{H}_{n,0} := \mathcal{H}_n \oplus \mathcal{H}_0\) is the orthogonal direct sum of the regression models \(\mathcal{H}_n\) and \(\mathcal{H}_0\). This construction lets us define an extended projection parameter \(\Psi_{n,0}\) and its EIF, which together characterize the model approximation error. Specifically, let \(\mathcal{M}_{n,0}\) be a smooth union model satisfying
\begin{enumerate}
    \item[(i)] \(\mathcal{M}_n \cup \mathcal{M}_0 \subseteq \mathcal{M}_{n,0}\);
    \item[(ii)] \(T_{\mathcal{M}_{n,0}}(P) = T_{\mathcal{M}_n}(P) \oplus T_{\mathcal{M}_0}(P)\) for all \(P \in \mathcal{M}_{n,0}\).
\end{enumerate}
Define the extended projection parameter \(\Psi_{n,0} := \Psi \circ \Pi_{n,0}\) induced by \(\mathcal{M}_{n,0}\), where \(\Pi_{n,0}(P) \in \arg\min_{Q \in \mathcal{M}_{n,0}} P\ell(\cdot,Q)\), and let \(D_{n,0,P_{n,0}}\) denote its EIF at the loss-based projection \(P_{n,0} := \Pi_n P_0\). For comparison, let \(D_{n,P_{n,0}}\) denote the \(P_{n,0}\)-EIF of \(\Psi_n\) induced by \(\mathcal{M}_n\).

Our main result requires specifying an approximation, derived from the working model \(\mathcal{M}_n\), to the EIF \(D_{n,0,P_{n,0}}\) of the extended projection parameter \(\Psi_{n,0}\). The relevant approximation space is the \textit{loss-based tangent space}
\[
\mathcal{S}_{\mathcal{M}_n}(P_{n,0}) := \overline{\mathrm{span}}\{\dot{\ell}_{P_{n,0}}(s) : s \in T_{\mathcal{M}_n}(P_{n,0})\},
\]
defined as the closure of the linear span of the \(P_{n,0}\)-weak G\^{a}teaux derivatives of \(\ell\). By Theorem~\ref{theorem::EIFmain}, the EIF \(D_{n,0,P_{n,0}}\) of \(\Psi_{n,0}\) lies in the loss-based tangent space \(\mathcal{S}_{\mathcal{M}_{n,0}}(P_{n,0})\) of \(\mathcal{M}_{n,0}\). We therefore compare \(D_{n,0,P_{n,0}}\) with an element \(\bar D_{n,0,P_{n,0}} \in \mathcal{S}_{\mathcal{M}_n}(P_{n,0})\), viewing \(\mathcal{S}_{\mathcal{M}_n}(P_{n,0}) \subseteq \mathcal{S}_{\mathcal{M}_{n,0}}(P_{n,0})\) as the approximating subspace. For example, \(\bar D_{n,0,P_{n,0}}\) may be chosen either as \(D_{n,P_{n,0}}\), the \(P_{n,0}\)-EIF of \(\Psi_n\), or as the \(L^2(P_{n,0})\)-projection of \(D_{n,0,P_{n,0}}\) onto \(\mathcal{S}_{\mathcal{M}_n}(P_{n,0})\); for the log-likelihood loss, these two choices coincide. The former choice can also be viewed as a projection with respect to the Hessian inner product introduced above Theorem~\ref{theorem::EIFmain}; see Appendix~\ref{appendix::modelapproxproj} for details. The following theorem presents our main result, along with a brief sketch of its proof.



\begin{enumerate}[label=A\arabic*), ref=A\arabic*, resume = cond]
\item \textit{Pathwise differentiability:} $\Psi_n$, $\Psi_0$, and $\Psi_{n,0}$ are pathwise differentiable at $P_{n,0}$ and $P_0$. \label{cond::oracleProjInModelFinal}
\end{enumerate}

\begin{theorem}[Representation for model approximation error]
  
    Suppose that Conditions \ref{cond::lossProjIdent}-\ref{cond::lossProjDeriv} hold and that \ref{cond::oracleProjInModelFinal} holds.  Let $\bar D_{n,0,P_{n,0}}$ be any element of $ \mathcal{S}_{\mathcal{M}_n}(P_{n,0})$, viewed as an  approximation of $D_{n,0,P_{n,0}}$. Then, the model approximation error can be decomposed as 
    \begin{align*}
        \Psi_n(P_0) - \Psi_0(P_0) = B_{n,0} + R_{n,0} 
    \end{align*} with $B_{n,0}:=(P_{n,0} - P_0)\left( D_{n,0,P_{n,0}} - \bar{D}_{n,0,P_{n,0}} \right)$ and 
   $R_{n,0} :=  \Psi_{n,0}(P_{n,0}) - \Psi_{n,0}(P_0) +  P_0 D_{n,0,P_{n,0}}$.
      \label{theorem::exactBiasOracle}
\end{theorem}
\begin{proof}[Proof sketch]
The proof leverages orthogonality properties of loss-based projections and the EIF characterization of Theorem \ref{theorem::EIFmain}. A pathwise Taylor expansion of $\Psi_{n,0}$ around $P_{n,0}$ yields $\Psi_n(P_0)-\Psi_0(P_0) = \Psi_{n,0}(P_{n,0})-\Psi_{n,0}(P_0) = (P_{n,0}-P_0)D_{n,0,P_{n,0}} + R_{n,0}. $ Because \(\bar{D}_{n,0,P_{n,0}} \in \mathcal{S}_{\mathcal{M}_n}(P_{n,0})\) by Theorem \ref{theorem::EIFmain}, the
first-order optimality conditions for the projection \(P_{n,0}\) imply that
\(P_0\bar{D}_{n,0,P_{n,0}}=0\), while the mean-zero property of scores in the tangent-space
implies that \(P_{n,0}\bar{D}_{n,0,P_{n,0}}=0\). Therefore, $(P_{n,0}-P_0)\bar{D}_{n,0,P_{n,0}}=0,$
and hence $(P_{n,0}-P_0)D_{n,0,P_{n,0}}
=
(P_{n,0}-P_0)\{D_{n,0,P_{n,0}}-\bar{D}_{n,0,P_{n,0}}\}
=
B_{n,0}.$
Substituting this identity into the expansion completes the proof.
\end{proof}
The key term in the bias expansion of Theorem~\ref{theorem::exactBiasOracle} is \(B_{n,0}\), which, by the Cauchy--Schwarz inequality, is bounded as
\begin{equation}
\label{eqn::approxerrorcauchy}
|B_{n,0}|
\lesssim
H(P_{n,0}, P_0)\,
\left\|\bar{D}_{n,0,P_{n,0}} - D_{n,0,P_{n,0}}\right\|_{P_{n,0}+P_0},
\end{equation}
where $H(P_{n,0}, P_0)
:=
\{\int (\sqrt{dP_{n,0}} - \sqrt{dP_0})^2 \}^{1/2} $ denotes the Hellinger distance. Thus, \(B_{n,0}\) is second order: the first factor measures how well the projection \(P_{n,0}\) approximates \(P_0\), while the second measures how well the loss-based tangent space \(\mathcal{S}_{\mathcal{M}_n}(P_{n,0})\) approximates the influence function \(D_{n,0,P_{n,0}}\), which lies in the larger space \(\mathcal{S}_{\mathcal{M}_{n,0}}(P_{n,0})\). Informally, the second factor is small if the working model \(\mathcal{M}_n\) provides a sufficiently accurate local approximation to \(\mathcal{M}_0\), and hence to the union model \(\mathcal{M}_{n,0}\), around \(P_{n,0}\). For variable and basis function selection, error bounds for the terms in \eqref{eqn::approxerrorcauchy} can be obtained under smoothness assumptions on \(P_0\) and \(D_{n, 0,P_{n,0}}\) using approximation theory for sieves \citep{shen1997methods,spnpsieve, BelloniSieveApprox} and for lasso-selected models under approximate sparsity \citep{belloni2013least, chernoapproxSparse2019}. Finally, the term \(R_{n,0}\) is a Taylor remainder and is typically second order in the error with which \(P_{n,0}\) approximates \(P_0\).

Theorem~\ref{theorem::exactBiasOracle} is deliberately stated in general form. Its interpretation is simpler, however, in the important special case where the model-selection procedure satisfies \(\mathcal{M}_n \subseteq \mathcal{M}_0\) with probability tending to one. Intuitively, such a procedure learns increasingly rich models as more data become available, but eventually remains within a fixed oracle model \(\mathcal{M}_0\). For example, a variable-selection procedure may include a covariate only once its signal is sufficiently large relative to the noise; the oracle model \(\mathcal{M}_0\) then corresponds to the collection of variables eventually included by \(\mathcal{M}_n\) in the infinite-data limit. In that case, \(\mathcal{M}_{n,0} = \mathcal{M}_0\) and $\mathcal{S}_{\mathcal{M}_{n,0}}(P_{n,0}) = \mathcal{S}_{\mathcal{M}_{0}}(P_{n,0})$, and hence \(\Psi_{n,0} = \Psi_0\) and \(D_{n,0,P_{n,0}} = D_{0,P_{n,0}}\). In this setting, the rate of the model approximation error is governed by how quickly the pair \((\mathcal{M}_n, T_{\mathcal{M}_n}(P_{n,0}))\) approaches \((\mathcal{M}_0, T_{\mathcal{M}_0}(P_{n,0}))\) from below, as well as by the smoothness of \(P_0\) and \(D_{0,P_0}\), closely paralleling approximation conditions for nested sieves \citep{shen1997methods,spnpsieve,SieveQiu}.

The inclusion condition \(P(\mathcal{M}_n \subseteq \mathcal{M}_0) \to 1\) holds under suitable conditions for a variety of model-selection procedures, including cross-validation over a finite collection or sieve of models \citep{shao1993linear, van2003unified, SieveQiu}, sparsity-based procedures such as the Lasso \citep{Tibshirani94regressionshrinkage}, adaptive Lasso \citep{AdaptLassoOracle}, and SCAD \citep{oracleSelectSCAD}, as well as various semiparametric model-selection procedures \citep{bunea2004consistent, li2008variable, claeskens2007asymptotic, ravikumar2009sparse, huang2010variable, su2014variable, xu2016faithful, amato2022wavelet}. Many existing results instead establish the stronger property of consistent model selection, such as exact support recovery, under which \(P(\mathcal{M}_n = \mathcal{M}_0) \to 1\) for an ``optimal'' oracle model \(\mathcal{M}_0\), for example, one containing only signal variables \citep{bauer1988model, potscher1991effects, buhlmann1999efficient, zhao2006model, wainwright2009sharp, belloni2013least}. By contrast, Theorem~\ref{theorem::exactBiasOracle} applies more broadly: it continues to hold when the selected models asymptotically contain some irrelevant variables or recover \(\mathcal{M}_0\) only approximately. Consequently, the model approximation error may vanish under weaker conditions than those required for exact support recovery, providing a route to valid post-selection inference for smooth functionals under weaker assumptions than those commonly imposed in the literature.

\subsection{Specialization to linear functionals of the outcome regression}

\label{adml::modelselectlinearfunc}

In this section, we apply Theorem~\ref{theorem::exactBiasOracle} to the special case of linear functionals of the outcome regression under lower-level conditions. To state the result, for a given model \(\mathcal{H}\), let \(\alpha_{0,\mathcal{H}}\) denote the Riesz representer in \(L^2(P_0)\) of the linear functional \(\mu \mapsto E_0[m(W,\mu)]\). Define the union regression model as the orthogonal direct sum \(\mathcal{H}_{n,0} := \mathcal{H}_n \oplus \mathcal{H}_0\), and let \(\operatorname{Proj}_n : \mathcal{H}_{n,0} \to \mathcal{H}_n\) denote the orthogonal projection onto \(\mathcal{H}_n\) in \(L^2(P_0)\).

\begin{enumerate}[label=(B\arabic*), ref=B\arabic*,resume=condlinear]
\item  \textit{Linear functional is continuous under union model:} \label{cond::boundedlinearfun} $\sup_{\mu \in \mathcal{H}_{n,0} \backslash \{0\} } \frac{E_0[m(W,\mu)]}{\|\mu\|} < \infty$.
\end{enumerate}

\begin{theorem}[Model approximation error for linear functionals]
\label{theorem::oraclebiaslinear}
    Under \ref{cond::boundedlinearfun} and $\mu_0 \in \mathcal{H}_0$, we have 
    $$\Psi_n(P_0) - \Psi(P_0) = - \langle \operatorname{Proj}_n(\alpha_{0,\mathcal{H}_{n,0}}) - \alpha_{0, \mathcal{H}_{n,0}},\, \operatorname{Proj}_n(\mu_{0}) - \mu_0  \rangle_{P_0}.$$
    Hence, by Cauchy-Schwarz, $|\Psi_n(P_0) - \Psi(P_0)| \leq \|\operatorname{Proj}_n (\alpha_{0,\mathcal{H}_{n,0}})  - \alpha_{0,\mathcal{H}_{n,0}}\|_{P_0} \|  \operatorname{Proj}_n(\mu_0) - \mu_0 \|_{P_0}$.
\end{theorem}
Theorem~\ref{theorem::oraclebiaslinear} follows from Theorem~\ref{theorem::exactBiasOracle} by taking \(\overline{D}_{n,0,P_0}\) to be the EIF of \(\Psi_n\), namely \(D_{n,P_{n,0}}\), though it may also be proved directly using orthogonality properties of projections. It generalizes existing formulas for omitted-variable bias \citep{chernozhukov2021long} and sieve approximation error \citep{shen1997methods, sieveOneStepPlugin, spnpsieve}, which correspond to the nonadaptive case in which \(\mathcal{H}_0 = \mathcal{H}_{n,0}= \mathcal{H}\).

For the approximation error \(\Psi_n(P_0)-\Psi_0(P_0)\) to vanish, the working model \(\mathcal H_n\) must approximate the oracle model \(\mathcal H_0\) sufficiently well. In particular, both \(\alpha_{0,\mathcal H_{n,0}}\) and \(\mu_0\) must be well approximated in \(L^2(P_0)\) by their projections onto \(\mathcal H_n\). For sieve-type selection procedures satisfying \(P(\mathcal H_n \subseteq \mathcal H_0)\to 1\), we have \(\alpha_{0,\mathcal H_{n,0}}=\alpha_{0,\mathcal H_0}\) with probability tending to one, so it suffices that \(\mathcal H_n\) grow sufficiently quickly toward a dense limit in \(\mathcal H_0\). More generally, when this inclusion condition fails, convergence requires that basis functions in \(\mathcal H_0\) orthogonal to \(\mathcal H_n\) contribute asymptotically negligibly to the basis expansion of the union-model representer \(\alpha_{0,\mathcal H_{n,0}}\).


The next example illustrates how the model approximation error vanishes in an approximately sparse setting in which basis functions are selected by marginal correlation screening under an orthonormal design. The same general mechanism extends to other screening procedures, such as the lasso under suitable design conditions; see, for example, \cite{donoho2005stable}, \cite{AdaptLassoOracle}, \cite{huang2008adaptive}, and \cite{wainwright2019high}.

\begin{example}[Marginal correlation screening with a truncated dictionary]
\label{example::screening}
Let \(\Phi=\{\varphi_j:j\in\mathbb N\}\) be an orthonormal basis for \(L^2(P_0)\), and write \(\mu_0=\sum_{j=1}^\infty \beta_j\varphi_j\) and \(\alpha_{0,\mathcal H}=\sum_{j=1}^\infty \gamma_j\varphi_j\). Define the active set \(S_0:=\{j\in\mathbb N:\beta_j\neq 0\}\) and the oracle model \(\mathcal H_0:=\mathrm{span}\{\varphi_j:j\in S_0\}\). Then \(\alpha_{0,\mathcal H_0}=\sum_{j\in S_0}\gamma_j\varphi_j\), so \(\mathcal H_0\) is the span of the support of \(\mu_0\), and \(\mu_0\) and \(\alpha_{0,\mathcal H_0}\) share the same support by construction. Thus, the main requirements are that \(\mathcal H_n \subseteq \mathcal H_0\) with probability tending to one and that both \(\mu_0\) and \(\alpha_{0,\mathcal H_0}\) are well approximated by elements of \(\mathcal H_n\).

Now let \(\Phi_{k(n)}:=\{\varphi_1,\ldots,\varphi_{k(n)}\}\) be a candidate dictionary of size \(k(n)\), where \(\log k(n)=O(\log n)\), and write \(S_{0,k}:=S_0\cap\{1,\ldots,k(n)\}\). Since \(\beta_j=\langle \varphi_j,\mu_0\rangle_{P_0}\), a natural selection rule is marginal correlation screening:
\[
\widehat S_n
:=
\left\{
j \le k(n):
\left|
\frac{1}{n}\sum_{i=1}^n \varphi_j(A_i,W_i)Y_i
\right| > (\log n)c_n
\right\},
\qquad
c_n \asymp \sqrt{\frac{\log k(n)}{n}},
\]
and we set \(\mathcal H_n:=\mathrm{span}\{\varphi_j:j\in \widehat S_n\}\). Under mild conditions, \(\widehat S_n \subseteq S_{0,k}\) with probability tending to one, and hence \(\mathcal H_n \subseteq \mathcal H_0\) asymptotically; see Appendix~\ref{appendix:example:proof}.

The approximation error for each of \(\mu_0\) and \(\alpha_{0,\mathcal H_0}\) has two components: a \emph{screening} term, due to relevant basis functions in \(\Phi_{k(n)}\) that are missed by the selection rule, and a \emph{truncation} term, due to basis functions outside \(\Phi_{k(n)}\). To control the screening term, assume \(\sum_{j=1}^\infty |\beta_j|^{q_\mu}<\infty\) for some \(q_\mu\in(0,2)\), and \(|\gamma_j|\lesssim |\beta_j|^a\) for all \(j\in S_0\) and some \(a>q_\mu/2\). Appendix~\ref{appendix:example:proof} shows that, with probability tending to one,
\[
\|\operatorname{Proj}_n(\mu_0)-\mu_0\|_{P_0}
\,
\|\operatorname{Proj}_n(\alpha_{0,\mathcal H_0})-\alpha_{0,\mathcal H_0}\|_{P_0}
\lesssim
n^{-(1+a-q_\mu)/2},
\]
up to logarithmic factors and second-order truncation terms.
\qed
\end{example}

To illustrate the generality of Theorem~\ref{theorem::oraclebiaslinear} (and Theorem~\ref{theorem::exactBiasOracle}) in allowing \(\mathcal{M}_n \not\subseteq \mathcal{M}_0\), consider the setting of Example~\ref{example::map}, in which \(\mathcal{H}_n\) is induced by a data-adaptive feature representation and \(\mathcal{H}_0\) by its oracle limit. Specifically, let \(\mathcal{H}_n := \mathcal{H}_{\phi_n}\) and \(\mathcal{H}_0 := \mathcal{H}_{\phi_0}\), where, for any feature representation \(\phi:\mathcal{W}\times\mathcal{A}\to\mathbb{R}^m\), we define \(\mathcal{H}_{\phi} := \{f\circ\phi : f:\mathbb{R}^m\to\mathbb{R}\}\). The combined model is then \(\mathcal{H}_{n,0} = \mathcal{H}_{(\phi_n,\phi_0)}\), where \((\phi_n,\phi_0)\) denotes the stacked feature representation. For any \(\phi\), let \(\alpha_{0,\phi}\) and \(\mu_{0,\phi}\) denote the \(L^2(P_0)\)-projections of \(\alpha_{0,\mathcal H}\) and \(\mu_0\) onto \(\mathcal{H}_{\phi}\), respectively. Then Theorem~\ref{theorem::oraclebiaslinear} implies that the approximation error is bounded by $\|\alpha_{0,\phi_n}-\alpha_{0,(\phi_n,\phi_0)}\|_{P_0}\,\|\mu_{0,\phi_n}-\mu_0\|_{P_0}.$
This bias vanishes if either \(\mu_0\) is well approximated by a function of \(\phi_n\), or the Riesz representer associated with \(\phi_n\) converges to that associated with the stacked representation \((\phi_n,\phi_0)\). Heuristically, the latter requires that conditioning on \((\phi_n,\phi_0)\) provides asymptotically no more information than conditioning on \(\phi_n\) alone.

The following lemma characterizes the feature-estimation bias for $\psi_0$ in terms of the quality of the estimated representation \(\phi_n\). We measure the quality of the learned representation by its mean integrated squared error,
\[
\|\phi_n-\phi_0\|_{P_0,2}
:=
\left\{\int \|\phi_n(a,w)-\phi_0(a,w)\|_{\mathbb{R}^d}^2\,P_0(dw,da)\right\}^{1/2}.
\]
Our result assumes Lipschitz continuity of certain bivariate conditional expectation functions.

\begin{enumerate}[label=(B\arabic*), ref=B\arabic*, resume=condlinear]
    \item \textit{Lipschitz continuity:} The bivariate mappings
$(t_1,t_2)\mapsto E_0[Y\mid \phi_n(A,W)=t_1,\ \phi_0(A,W)=t_2,\ \mathcal{D}_n]$ and $ (t_1,t_2)\mapsto E_0[\alpha_{0,\mathcal H}(A,W)\mid \phi_n(A,W)=t_1,\ \phi_0(A,W)=t_2,\ \mathcal{D}_n]$    are almost surely \(L\)-Lipschitz continuous for a fixed $L \in (0,\infty)$, where \(\mathcal H\) is any linear space satisfying \(\mathcal{H}_{n,0}\subseteq\mathcal H\). \label{cond::lipschitzFeature}
\end{enumerate}
When \(\alpha_{0,\mathcal H}(A,W)\) is a conditional density ratio for \(A\) given \(W\), such as inverse propensity weights, and \(\phi_n(A,W) = (A, \widetilde{\phi}_n(W))\), the quantity $E_0[\alpha_{0,\mathcal H}(A,W)\mid \phi_n(A,W),\ \phi_0(A,W),\ \mathcal{D}_n]$ typically reduces to the corresponding conditional density ratio given \((\widetilde{\phi}_n(W),\widetilde{\phi}_0(W))\) \citep{ATEsupereff}.

\begin{lemma}[Estimation error for feature representations]
\label{lemma::lipschitzdependent}
Under Conditions~\ref{cond::boundedlinearfun} and \ref{cond::lipschitzFeature},
\[
\|\operatorname{Proj}_n(\alpha_{0,\mathcal{H}_{n,0}})-\alpha_{0,\mathcal{H}_{n,0}}\|_{P_0}
+
\|\operatorname{Proj}_n(\mu_0)-\mu_0\|_{P_0}
\lesssim
L \|\phi_n-\phi_0\|_{P_0,2}.
\]
Hence, $|\Psi_n(P_0)-\Psi(P_0)|
\lesssim
L^2 \|\phi_n-\phi_0\|_{P_0,2}^2$.
\end{lemma}

Lemma~\ref{lemma::lipschitzdependent} shows that the model approximation error vanishes when the estimated feature representation \(\phi_n\) converges to the oracle representation \(\phi_0\) in mean integrated squared error. The key requirement is \ref{cond::lipschitzFeature}, which transfers this convergence to the nuisance functions appearing in Theorem~\ref{theorem::oraclebiaslinear}. Related bounds are implicit in prior work on debiased machine learning with estimated features, where either \(\|\alpha_{0,\phi_n}-\alpha_{0,(\phi_n,\phi_0)}\|_{P_0}\|\mu_{0,\phi_n}-\mu_0\|_{P_0}\) is directly assumed negligible (Theorem 1 of \cite{ATEsupereff}; Assumption 3 of \cite{dukes2024doubly}), or conditions similar to \ref{cond::lipschitzFeature} are imposed (Condition (v)  of \cite{ATEsupereff}; Lemma I.1 of \cite{wang2023super}; and Section 3.1 of \cite{bonvini2024doubly}).

 \section{Large-sample theory: general case}
\label{section::theory}

\subsection{Asymptotic linearity and efficiency for oracle parameter}
\label{section::dataAdaptParam}

\label{section::oracleParamInference}

With Theorem~\ref{theorem::exactBiasOracle} in hand, we are now in a position to show that the ADML estimator is regular, asymptotically linear, and nonparametrically efficient for the oracle parameter $\Psi_0$ at $P_0$. Our main result assumes the following high-level conditions.

 \begin{enumerate}[label=(C\arabic*), ref=C\arabic*,series=condB]

\item \textit{First order estimator expansion:}  $\widehat{\psi}_n = \Psi_n(P_0) + (P_n-P_0) D_{n,P_0} + o_p(n^{-1/2})$ with $D_{n,P_0}$ the EIF of $\Psi_n:\mathcal{M}_{\mathrm{np}}\rightarrow\mathbb{R}$;\label{cond::debiased}
\item \textit{Local consistency of model:}  $ \norm{D_{n,P_0}- D_{0,P_{0}}}_{P_0}  = o_p(1)$;\label{cond::consDn}
\item \textit{Negligible empirical process remainder:}  $(P_n - P_0) \left(D_{n,P_0} - D_{0,P_0}\right) =  o_p(n^{-1/2})$. 
\label{cond::Donsker}
\item \textit{Negligible model approximation error:}  $\Psi_n(P_0) - \Psi_0(P_0)  =o_p(n^{-1/2})$.\label{cond::doubleRemrootnFinal}
 
\end{enumerate}

Condition~\ref{cond::debiased} is the defining property of the ADML estimator. Such an expansion can typically be verified for debiased machine learning estimators of the data-adaptive working parameter \(\Psi_n\), as shown by \cite{dataAdaptTargetParam}. In particular, the one-step debiased estimator of \(\Psi_n\), given by \(\Psi(\widehat{P}_n) + P_n D_{n,\widehat{P}_n}\) for an estimator \(\widehat{P}_n \in \mathcal{M}_n\) of \(P_0\), satisfies this expansion provided that the second-order and empirical-process remainders are negligible; namely, \(\Psi_n(\widehat{P}_n) - \Psi_n(P_0) + P_0 D_{n,\widehat{P}_n} = o_p(n^{-1/2})\) and \((P_n - P_0)(D_{n,\widehat{P}_n} - D_{n,P_0}) = o_p(n^{-1/2})\) \citep{bickel1993efficient, dataAdaptTargetParam}. Condition~\ref{cond::consDn} requires the \(P_0\)-EIF of \(\Psi_n\) to converge to the \(P_0\)-EIF of \(\Psi_0\), and may be interpreted as a local consistency requirement on the working model \(\mathcal{M}_n\) relative to \(\mathcal{M}_0\). In Appendix~\ref{appendix::stable}, we show that this condition holds when \(T_{\mathcal{M}_n}(P_{n,0})\) approximates the union tangent space \(T_{\mathcal{M}_n}(P_{n,0}) \oplus T_{\mathcal{M}_0}(P_{n,0})\) in the sense that \(\inf_{s \in T_{\mathcal{M}_n}(P_{n,0})} P_0 \ddot{\ell}_{P_{n,0}}(s_{n,0,P_{n,0}} - s,\; s_{n,0,P_{n,0}} - s) = o_p(1)\), together with other mild conditions. Condition~\ref{cond::Donsker} follows from \ref{cond::consDn} if \(D_{n,P_0}\) belongs to a \(P_0\)-Donsker class, or alternatively if suitable sample-splitting techniques are used \citep{CVTMLEworking2010, DoubleML}. Finally, \ref{cond::doubleRemrootnFinal} requires the model approximation error \(\Psi_n(P_0) - \Psi_0(P_0)\) to be asymptotically negligible. By Theorem~\ref{theorem::exactBiasOracle}, this holds if \(B_{n,0} = o_p(n^{-1/2})\) and \(R_{n,0} = o_p(n^{-1/2})\).

\begin{theorem}[Regularity and efficiency for oracle parameter]
\label{theorem::limitDataAdaptOracle}
\label{theorem::oracleEff}

Suppose Conditions~\ref{cond::lossProjIdent}--\ref{cond::oracleProjInModelFinal} and \ref{cond::debiased}--\ref{cond::Donsker} hold.
\begin{enumerate}
    \item[(i)] The ADML estimator \(\widehat{\psi}_n\) is \(P_0\)-asymptotically linear for \(\Psi_n(P_0)\), satisfying
    \[
    \widehat{\psi}_n = \Psi_n(P_0) + (P_n-P_0)D_{0,P_0} + o_p(n^{-1/2}).
    \]

 \item[(ii)] If, in addition, \ref{cond::doubleRemrootnFinal} holds, then \(\widehat{\psi}_n\) is \(P_0\)-asymptotically linear for \(\Psi_0\), with influence function equal to the EIF of \(\Psi_0:\mathcal{M}_{\mathrm{np}}\to\mathbb{R}\), and satisfies
\[
\widehat{\psi}_n = \Psi_0(P_0) + P_n D_{0,P_0} + o_p(n^{-1/2}).
\]
Consequently, \(\sqrt{n}\{\widehat{\psi}_n-\Psi_0(P_0)\} \rightsquigarrow N\!\left(0,\,P_0D_{0,P_0}^2\right)\), and \(\widehat{\psi}_n\) is \(P_0\)-regular and efficient for \(\Psi_0\).
\end{enumerate}
\end{theorem}

Theorem~\ref{theorem::oracleEff} shows that the ADML estimator is \(P_0\)-regular and nonparametrically efficient for \(\Psi_0\). This regularity implies that inference based on the asymptotic normality of the estimator is valid for $\Psi_0$ both pointwise and locally uniformly over local perturbations of \(P_0\) in the nonparametric model. Thus, at least in the local asymptotic sense, there is no loss from learning \(\mathcal{M}_0\) empirically rather than knowing \(\mathcal{M}_0\) or \(\Psi_0\) in advance. We next show how these oracle properties extend to the original target parameter \(\Psi\).


\begin{remark}
Asymptotic normality, though not necessarily asymptotic linearity or efficiency, can hold without \ref{cond::consDn}. This is particularly relevant for unstable model-selection procedures, for which \ref{cond::consDn} may fail because \(\mathcal{M}_n\) does not converge, for example when it alternates between competing models. One remedy is sample splitting, in which the working model is learned on an independent sample \citep{van2013AdaptTarget, rinaldo2019bootstrapping}. Although this may reduce efficiency, it allows inference for the data-adaptive target parameter \(\Psi_n(P_0)\) under conditions similar to those used in standard nonadaptive DML. Another approach is to assume that \(\mathcal{M}_n\) is asymptotically deterministic, while still varying with \(n\). In that case, asymptotic normality follows from a suitable triangular-array central limit theorem applied to \(\sqrt{n/\sigma_n^2}\,(P_n-P_0)D_{n,P_0}\), where \(\sigma_n^2\) is an appropriate scaling constant. For example, \cite{danielleWittenLassoWorks} established asymptotic determinism for Lasso-based working models under relatively mild conditions.
\end{remark}

\subsection{Superefficiency and regularity properties for the original target parameter}
\label{section::oracleParamInference::original}

We now show that, when the efficiency bound for \(\Psi_0\) is smaller than that for \(\Psi\), the ADML estimator is \(P_0\)-superefficient for \(\Psi\) in the nonparametric model. This follows directly from Theorem~\ref{theorem::oracleEff}, since \(\Psi(P)=\Psi_0(P)\) for all \(P\in\mathcal{M}_0\), and in particular \(\Psi(P_0)=\Psi_0(P_0)\) when \(P_0 \in \mathcal{M}_0\).

\begin{theorem}[Superefficiency and regularity for the original parameter]
\label{theorem::oracleRegularity}
Suppose that \(P_0 \in \mathcal{M}_0\) and that the conditions of Theorem~\ref{theorem::oracleEff} hold. Then the ADML estimator \(\widehat{\psi}_n\) has the following properties:
\begin{enumerate}
    \item[(i)] it satisfies the asymptotically linear expansion
    \[
    \widehat{\psi}_n = \psi_0 + (P_n - P_0) D_{0,P_0} + o_p(n^{-1/2})
    \]
    at \(P_0\), where \(D_{0,P_0}\) is the \(P_0\)-EIF of \(\Psi_0\);

    \item[(ii)] it is \(P_0\)-regular for \(\Psi\) over all local perturbations \(P_{0,hn^{-1/2}}\) in the oracle submodel \(\mathcal{M}_0\);

    \item[(iii)] if, in addition, \(\ell\) is the negative log-likelihood loss, then \(\widehat{\psi}_n\) is asymptotically \(P_0\)-efficient for \(\Psi\) relative to the oracle submodel \(\mathcal{M}_0\).
\end{enumerate}
Consequently, $\sqrt{n}\,(\widehat{\psi}_n - \psi_0) \rightsquigarrow N\!\left(0,\,P_0D_{0,P_0}^2\right),$
including under sampling from local perturbations \(P_{0,hn^{-1/2}}\) of \(P_0\) that remain in \(\mathcal{M}_0\).
\end{theorem}

By Theorem \ref{theorem::oracleRegularity}, adaptive Wald-type inference for \(\Psi(P_0)\) can be conducted using a consistent estimator of the limiting variance \(\sigma_0^2\). Under \ref{cond::consDn}, \(\sigma_0^2\) can be consistently estimated by the empirical variance estimator \(\sigma_n^2 := n^{-1}\sum_{i=1}^{n} D_{n,\widehat P_n}(O_i)^2\), based on the influence function for \(\Psi_n(P_0)\), provided that \(\widehat P_n\) is a suitably regular estimator of \(P_0\).  For the M-estimator in Example~\ref{example::mestimator}, the limiting variance \(\sigma_0^2\) can be estimated by the model-robust sandwich variance estimator for the finite-dimensional model \(\mathcal{M}_n\).

The ADML estimator \(\widehat{\psi}_n\) is \(P_0\)-superefficient for \(\Psi\) when the efficiency bound for \(\Psi_0\), and hence its limiting variance, is smaller than the efficiency bound for \(\Psi\) in the nonparametric model; superefficiency is maximized under the log-likelihood loss. This typically occurs when the oracle tangent space \(T_{\mathcal{M}_0}(P_0)\) is strictly smaller than \(T_{\mathcal{M}_{\mathrm{np}}}(P_0)\), since the EIF for \(\Psi_0\) may then differ from that for \(\Psi\). Any such estimator is necessarily irregular for \(\Psi\) at \(P_0\) relative to \(\mathcal{M}_{\mathrm{np}}\). Nevertheless, Theorem~\ref{theorem::oracleRegularity} shows that inference for \(\Psi\) remains pointwise valid and locally uniformly valid over \(\mathcal{M}_0\), while Theorem~\ref{theorem::limitDataAdaptOracle} shows that inference for the oracle parameter \(\Psi_0\) is locally uniformly valid over the nonparametric model \(\mathcal{M}_{\mathrm{np}}\). Thus, ADML estimators lie on a continuum between regularity and superefficiency, determined by the complexity of the oracle model. Sacrificing some regularity may be worthwhile to gain efficiency, especially when regular nonparametric estimators of \(\Psi\) are unavailable, such as when the ATE is not nonparametrically identifiable or overlap is insufficient, as discussed by \cite{CTMLE} , \citet{ATEsupereff}, \citet{moosavi2023costs}, and \citet{dukes2024doubly}.

To understand the impact of irregularity on inference for \(\Psi\), the following theorem characterizes the limiting bias of the ADML estimator under local perturbations \(P_{0,hn^{-1/2}}\) of \(P_0\) within the prespecified model \(\mathcal{M}\). Instead of examining the least-favorable asymptotic mean squared error as \(h \to \infty\), which recovers the local asymptotic minimax bounds of \cite{hajek1972local} but is necessarily infinite for nonregular estimators, we focus on the behavior at a fixed perturbation size \(h\), which is more informative for superefficient estimators.

\begin{theorem}[Limiting distribution under local perturbations]
    Suppose that the conditions of Theorem \ref{theorem::oracleEff} hold, $P_0 \in \mathcal{M}_0$, and that $\Psi$ is pathwise differentiable at $P_0$ relative to the prespecified statistical model $\mathcal{M}$ with EIF $D_{\mathcal{M}, P_0} \in T_{\mathcal M} (P_0)$. Then, under sampling from any local perturbation $P_{0,hn^{-1/2}} \in  \mathcal{M}$ of $P_0$ with $h \in \mathbb{R}$ and score $s \in T_{\mathcal M} (P_0)$, the ADML estimator $\widehat{\psi}_n$ satisfies that
    $$\sqrt{n}\,\{\widehat{\psi}_n - \Psi(P_{0,hn^{-1/2}})\}\xrightarrow[ ]{\;\;d\;\;}  N(b_0(h;s), \sigma_0^2)\ ,$$ 
    where $b_{0}(h;s) := h\langle s, D_{0,P_0} -  D_{\mathcal{M}, P_0}\rangle_{P_0} $ and $\sigma_0^2 := \mathrm{var}_0\{D_{0,P_0}(O)\}$.  
    \label{theorem::irreg}
    
\end{theorem}

Theorem~\ref{theorem::oracleRegularity} characterizes the local asymptotic bias of the ADML estimator outside the oracle model \(\mathcal{M}_0\). To interpret \(h\) as a local distance, suppose \(P_{0,hn^{-1/2}}\) has unit-norm score, that is, \(\|s\|_{P_0}=1\). Then \(h\) is asymptotically equal to the $\sqrt{n}$-scaled Hellinger distance between \(P_{0,hn^{-1/2}}\) and \(P_0\), since $n^{1/2}\|\sqrt{p_{0,hn^{-1/2}}}-\sqrt{p_0}\|_{\mu}
=
h+o(1)$
as \(n\to\infty\), where \(p_0 := dP_0/d\mu\) denotes the \(\mu\)-density of \(P_0\). Within the oracle model, the bias vanishes: \(b_0(h;s)=0\) for every score \(s \in T_{\mathcal{M}_0}(P_0)\), corresponding to local perturbations of \(P_0\) that remain in \(\mathcal{M}_0\) to first order. More generally, by the Cauchy--Schwarz inequality, among local perturbations with \(\|s\|_{P_0}=1\), the asymptotic bias \(b_0(h;s)\) is maximized when \(s\) is proportional to \(D_{0,P_0}-D_{\mathcal{M},P_0}\); the resulting maximal absolute bias is \(h\|D_{0,P_0}-D_{\mathcal{M},P_0}\|_{P_0}\).

The theorem reveals a local asymptotic bias--variance tradeoff for inference on \(\Psi\), governed by the complexity of \(\mathcal{M}_0\): the bias under local perturbations is determined by the degree of irregularity, while the variance is determined by the degree of superefficiency. In particular, ADML estimators can outperform regular estimators in mean squared error under local perturbations of \(P_0\) that remain near the oracle submodel \(\mathcal{M}_0\), that is, for small \(|h|\). For large \(|h|\), however, they become suboptimal under perturbations outside \(\mathcal{M}_0\), and their mean squared error diverges as \(|h| \to \infty\). To make this concrete, suppose \(\ell\) is the log-likelihood loss, so that \(S_{\mathcal{M}_0}(P_0) \subseteq T_{\mathcal{M}_0}(P_0)\). Let \(\sigma^2(\mathcal{M}) := \mathrm{var}_0\{D_{\mathcal{M},P_0}(O)\}\) denote the efficiency bound at \(P_0\) relative to \(\mathcal{M}\). Then \(\|D_{0,P_0} - D_{\mathcal{M},P_0}\|_{P_0}^2\) equals the absolute efficiency gain \(\Delta_0^2 := \sigma^2(\mathcal{M}) - \sigma_0^2\) achieved by working under \(\mathcal{M}_0\) rather than \(\mathcal{M}\). In this case, Theorem~\ref{theorem::irreg} shows that the asymptotic mean squared error of the ADML estimator under a least-favorable local perturbation in \(\mathcal{M}\) with unit-norm score is \(h^2\Delta_0^2 + \sigma_0^2\). This least-favorable mean squared error is strictly smaller than the local asymptotic minimax bound over \(\mathcal{M}\) when \(|h|<1\), and equals that bound when \(|h|=1\), since the latter is precisely \(\sigma^2(\mathcal{M}) = \Delta_0^2 + \sigma_0^2\) \citep{hajek1972local}. Thus, for perturbations near \(\mathcal{M}_0\), in the sense that \(|h|\leq 1\), the ADML estimator has mean squared error no larger than that of a prespecified efficient estimator for \(\mathcal{M}\), despite potentially having larger bias. A related observation was made by \cite{lumley2017robustness} for estimation of the average treatment effect under nearly true models.

An important implication is that a prespecified estimator for \(\Psi\) that is \(P_0\)-efficient relative to a known model \(\mathcal{M}' \subseteq \mathcal{M}_0\) generally has larger least-favorable local asymptotic bias than the ADML estimator based on \(\mathcal{M}_0 \subseteq \mathcal{M}\). Thus, if the working model \(\mathcal{M}_n\) is learned under the constraint \(\mathcal{M}' \subseteq \mathcal{M}_n\), then the resulting ADML estimator is asymptotically no more biased, under sampling from any distribution in \(\mathcal{M}\), than the prespecified estimator based on \(\mathcal{M}'\). In particular, under negative log-likelihood loss with \(P_0 \in \mathcal{M}'\), the least-favorable local asymptotic bias over \(\mathcal{M}\) of the prespecified estimator is \(\|D_{\mathcal{M},P_0} - D_{\mathcal{M}',P_0}\|_{P_0}\), which is no smaller than the corresponding least-favorable bias of the ADML estimator, \(\|D_{\mathcal{M},P_0} - D_{\mathcal{M}_0,P_0}\|_{P_0}\), whenever \(\mathcal{M}' \subseteq \mathcal{M}_0\). For example, the partially linear ADML estimators in Examples~\ref{example::CATE} and~\ref{example::catecal} always include the homogeneous partially linear model, which assumes a constant CATE, as a special case. Consequently, these ADML estimators yield locally valid inference over a broader class of distributions than the prespecified estimator based on the homogeneous partially linear model.
 
\section{Adaptive debiased machine learning for linear functionals}
\label{adml::linearfunc}
\subsection{Proposed estimator}
In this section, we introduce automatic ADML estimators for linear functionals of the outcome regression and revisit the examples from Section~\ref{section::ADMLlinear}. We then develop the asymptotic theory for this class by building on the results of Section~\ref{section::oracleParamInference}.

Recall that \(\mathcal{H}_n\) is a working model used to approximate an oracle model \(\mathcal{H}_0\). Let \(\mu_n \in \mathcal{H}_n\) be an estimator of \(\mu_0\), and let \(\alpha_n \in \mathcal{H}_n\) be an estimator of the Riesz representer \(\alpha_{0,\mathcal{H}_n}\) of the linear functional over \(\mathcal{H}_n\), as defined above Theorem~\ref{theorem::boundedlinearEIF}. The ADML estimator of \(\psi_0 = E_0[m(W,\mu_0)]\) is the automatic DML estimator \citep{chernozhukov2022automatic} for \(\Psi_n(P_0)\):
\begin{equation}
    \widehat{\psi}_n := \frac{1}{n}\sum_{i=1}^n m(W_i, \mu_n) + \frac{1}{n}\sum_{i=1}^n \alpha_n(A_i, W_i)\{Y_i - \mu_n(A_i, W_i)\}. \label{eqn::linearfuncestimator}
\end{equation}
The second term is a one-step bias correction of the plug-in estimator \(\frac{1}{n} \sum_{i=1}^n m(W_i, \mu_n)\) using the EIF of \(\Psi_n\), so that the first-order expansion in \ref{cond::debiased} holds under weak conditions \citep{bickel1993efficient}. Following the autoDML framework of \cite{chernozhukov2022automatic}, one can estimate \(\mu_n\) by least-squares regression, that is, by minimizing \(\sum_{i=1}^n \{Y_i-\theta(A_i,W_i)\}^2\) over \(\theta \in \mathcal{H}_n\), and estimate \(\alpha_n\) analogously by Riesz regression, that is, by minimizing \(\sum_{i=1}^n \bigl[\alpha(A_i,W_i)^2 - 2\,m(W_i,\alpha)\bigr]\) over \(\alpha \in \mathcal{H}_n\).

Below, we revisit the examples in Section~\ref{section::ADMLlinear} and show that the estimators are of the above form.

\setcounter{example}{1}
\begin{example}[continued]
Recall the partially linear ADML estimator \(\widehat{\psi}_n = \frac{1}{n} \sum_{i=1}^n \tau_n(W_i)\) for the ATE, where \(\tau_n\) is the R-learner with working CATE model \(\mathcal{T}_n\) \citep{nie2021quasi}. This estimator is a special case of \eqref{eqn::linearfuncestimator}, with \(\mu_n(a,w) := m_n(w) + (a-\pi_n(w))\tau_n(w)\) a Robinson-parameterized regression estimator \citep{robinson1988root}, \(\alpha_n(a,w) := \gamma_{0,\mathcal{T}_n}(w)\{a-\pi_n(w)\}\), and \(m:(w,\mu) \mapsto \mu(1,w)-\mu(0,w)\). For these choices, we first observe that \(\widehat{\psi}_n = \frac{1}{n} \sum_{i=1}^n m(W_i,\mu_n) = \frac{1}{n} \sum_{i=1}^n \{\mu_n(1,W_i)-\mu_n(0,W_i)\}\), so \(\widehat{\psi}_n\) is the plug-in estimator. Moreover, the correction term in \eqref{eqn::linearfuncestimator} vanishes exactly by the first-order optimality conditions for the empirical risk minimizer \(\tau_n\): for each \(\gamma \in \mathcal{T}_n\), \(\frac{1}{n} \sum_{i=1}^n \gamma(W_i)\{A_i-\pi_n(W_i)\}\{Y_i-\mu_n(A_i,W_i)\}=0\). In particular, setting \(\gamma=\gamma_{0,\mathcal{T}_n}\) yields \(\frac{1}{n} \sum_{i=1}^n \tau_n(W_i) = \frac{1}{n} \sum_{i=1}^n m(W_i,\mu_n) + \frac{1}{n} \sum_{i=1}^n \alpha_n(A_i,W_i)\{Y_i-\mu_n(A_i,W_i)\}\).
\end{example}

\begin{example}[continued]
Recall that the calibrated plug-in regression estimator \(\widehat{\psi}_n := \frac{1}{n} \sum_{i=1}^n m(W_i,\mu_n^*)\) is an ADML estimator with \(\mathcal{H}_n := \{f \circ \mu_n^* : f\}\), where \(\mu_n^* = f_n \circ \mu_n\) denotes the isotonic-calibrated regression estimator. We show that \(\widehat{\psi}_n\) admits the representation \eqref{eqn::linearfuncestimator} with $\mu_n := \mu_n^*$ and \(\alpha_n = \alpha_{0,\mathcal{H}_n}\), the corresponding true representer. The key observation is that the KKT conditions for the isotonic regression solution imply the orthogonality relations \citep{van2024automatic}
\begin{equation*}
\frac{1}{n} \sum_{i=1}^n f(\mu_n^*(A_i,W_i))\{Y_i - \mu_n^*(A_i,W_i)\} = 0, \quad \forall f:\mathbb{R}\to\mathbb{R}.
\end{equation*}
Hence, by the definition of \(\mathcal{H}_n\), this orthogonality condition implies that the debiasing term vanishes: $\frac{1}{n}\sum_{i=1}^n \alpha_{0,\mathcal{H}_n}(A_i,W_i)\{Y_i - \mu_n^*(A_i,W_i)\} = 0.$
Thus, the plug-in estimator \(\widehat{\psi}_n\) satisfies \eqref{eqn::linearfuncestimator} for our choice of nuisance estimators. In the special case of the counterfactual mean \(E_0[\mu(1,W)]\), the superefficient TMLE of \citet{ATEsupereff} instead enforces $\frac{1}{n}\sum_{i=1}^n A_i \widehat{\alpha}_n(W_i)\{Y_i - \mu_n^*(A_i,W_i)\} = 0 $ for an estimator \(\widehat{\alpha}_n(W)\) of $1/\pi_{0,\mu_0}(W)$. Their approach is asymptotically equivalent to a modified version of \(\widehat{\psi}_n\) that calibrates only on the subset of observations with \(A_i=1\). A key advantage of our approach is that it is automatic: isotonic calibration simultaneously debiases all linear functionals without requiring explicit adjustment based on the corresponding Riesz representers.
 
\end{example}

\begin{example}[continued]
Recall the calibrated plug-in R-learner \(\widehat{\psi}_n := \frac{1}{n}\sum_{i=1}^n \tau_n^*(W_i)\), where \(\tau_n^* = f_n \circ \tau_n\) is obtained via isotonic calibration with the R-learner loss \citep{van2023causal, nie2021quasi}. Recall also that \(\mathcal{H}_n = \mathcal{H}_{\tau_n^*}\) and \(\mathcal{H}_0 = \mathcal{H}_{\tau_0}\), where \(\mathcal{H}_\tau := \{(a,w) \mapsto m(w) + a\, f(\tau(w)) : m \in L^2(P_{0,W}),\ f:\mathbb{R} \to \mathbb{R}\}\). We show that \(\widehat{\psi}_n\) admits the representation \eqref{eqn::linearfuncestimator} with \(\mu_n(a,w) = m_n(w) + (a - \pi_n(w))\tau_n^*(w)\), \(\alpha_n(a,w) := \{a - \pi_n(w)\}/E_0[(A - \pi_0(W))^2 \mid \tau_n^*(W) = \tau_n^*(w)]\), and \(m(o,\mu) := \mu(1,w) - \mu(0,w)\). As in Example~\ref{example::map}, the calibrated R-learner \(\tau_n^*\) satisfies the orthogonality conditions
\[
\frac{1}{n} \sum_{i=1}^n f(\tau_n^*(W_i))\{A_i - \pi_n(W_i)\}\{Y_i - \mu_n(A_i,W_i)\} = 0, \quad \forall f:\mathbb{R} \to \mathbb{R}.
\]
Consequently, taking \(f \circ \tau_n^*\) to be \(\gamma_{0,\mathcal{T}_n} := \{E_0[(A - \pi_0(W))^2 \mid \tau_n^*(W) = \tau_n^*(\cdot)]\}^{-1}\), the debiasing term vanishes for our choice of \(\alpha_n\), and the plug-in estimator \(\widehat{\psi}_n = \frac{1}{n}\sum_{i=1}^n \tau_n^*(W_i)\) satisfies \eqref{eqn::linearfuncestimator} with the claimed choices.
\end{example}

\begin{example}[Augmented minimax linear estimation]
The augmented minimax linear estimator of \cite{augmentedMinimax} can be viewed as an ADML estimator corresponding to the affine spaces $\mathcal{H}_n := \{\mu_n + f: f \in \mathcal{F}_n\}$ and $\mathcal{H}_0 := \{\mu_0 + f: f \in \mathcal{F}_0\}$, where $\mathcal{F}_n \subseteq \mathcal{F}_0$ is a working linear model for the regression error $\mu_n - \mu_0 \in \mathcal{F}_0$. In this case, the Riesz representers $\alpha_{0,\mathcal{H}_n}$ and $\alpha_{0,\mathcal{H}_0}$ may be defined relative to the tangent spaces $\mathcal{F}_n$ and $\mathcal{F}_0$. One instance corresponds to \eqref{eqn::linearfuncestimator}, with $\alpha_n := \argmin_{\alpha \in \mathcal{F}_n} \sum_{i=1}^n \{\alpha(A_i, W_i)^2 - 2 m(W_i, \alpha)\}$ given by a sieve estimator.
\end{example}

\subsection{Large-sample theory: asymptotic linearity and superefficiency}

We now establish the asymptotic linearity and superefficiency of ADML estimators for linear functionals of the outcome regression. The main result relies on the following conditions and may be viewed as a special case of Theorems~\ref{theorem::oracleEff} and \ref{theorem::oracleRegularity}. In what follows, let \(\operatorname{Proj}_n\) and \(\operatorname{Proj}_0\) denote the \(L^2(P_0)\)-orthogonal projections from \(\mathcal{H}\) onto \(\mathcal{H}_n\) and \(\mathcal{H}_0\), respectively.

\begin{enumerate}[label=(B\arabic*), ref=B\arabic*,resume=condlinear]
     \item \textit{Boundedness:} $\alpha_n$, $\mu_n$, $\mu_0$, $\alpha_{0,\mathcal{H}_{0}}$ are $P_0$-uniformly bounded almost surely. \label{cond::bounded}
    \item \textit{Empirical process control:} \label{cond::linearsamplesplit} Either (i) $\alpha_n$ and $\mu_n$ are trained on a dataset independent from  $\{O_i\}_{i=1}^n$; or (ii) $ (P_n - P_0) (\widehat{D}_n - D_{0, P_0}) = o_p(n^{-1/2})$, where $\widehat{D}_n(o) := \alpha_n(a,w)\{y - \mu_n(a,w)\} + m(w, \mu_n) -  P_n m(\cdot, \mu_n)$ 
    \item \textit{Rates for working nuisances:}  \label{cond::linearnuisancerate}
    \begin{enumerate}
        \item[(i)] $\|\alpha_n - \alpha_{0, \mathcal{H}_n}\|_{P_0} + \|\mu_n - \mu_{0, \mathcal{H}_n}\|_{P_0} = o_p(1)$;
        \item[(ii)] $\|\alpha_n - \alpha_{0, \mathcal{H}_n}\|_{P_0}\|\mu_n - \mu_{0, \mathcal{H}_n}\|_{P_0} = o_p(n^{-1/2})$.
    \end{enumerate} 
    \item \textit{Negligible model approximation error:} \label{cond::oraclebiaslinear}
    \begin{enumerate}
        \item[(i)] $\|\operatorname{Proj}_n (\alpha_{0,\mathcal{H}_{n,0}})  - \alpha_{0,\mathcal{H}_{n,0}}\|_{P_0} + \|\operatorname{Proj}_0 (\alpha_{0,\mathcal{H}_{n,0}})  - \alpha_{0,\mathcal{H}_{n,0}}\|_{P_0} +\|  \operatorname{Proj}_n(\mu_0) - \mu_0 \|_{P_0}  = o_p(1)$;
        \item[(ii)] $\|\operatorname{Proj}_n (\alpha_{0,\mathcal{H}_{n,0}})  - \alpha_{0,\mathcal{H}_{n,0}}\|_{P_0} \|  \operatorname{Proj}_n(\mu_0) - \mu_0 \|_{P_0} = o_p(n^{-1/2})$.  
    \end{enumerate} 
\end{enumerate}
Conditions~\ref{cond::bounded}--\ref{cond::linearnuisancerate} are standard in semiparametric statistics for debiased estimation of the working parameter \(\Psi_n(P_0)\); see, for example, \cite{bickel1993efficient, robins1995analysis, vanderlaanunified, vanderLaanRose2011, DoubleML, chernozhukov2022automatic}. In particular, \ref{cond::linearsamplesplit} ensures that \ref{cond::Donsker} holds by requiring either (i) sample splitting, so that the nuisance estimators are computed using data independent of that used to form the one-step estimator, or (ii) that the relevant empirical process remainder is asymptotically negligible. To use the full sample, one can instead apply cross-fitting, in which nuisance and parameter estimates are computed repeatedly across multiple splits and then averaged \citep{van2011cross, DoubleML}. Condition~\ref{cond::linearnuisancerate} requires that the nuisance estimators \(\mu_n\) and \(\alpha_n\) converge to their respective working-model targets in \(\mathcal{H}_n\), thereby ensuring that the ADML estimator \(\widehat{\psi}_n\) is debiased for the working parameter \(\Psi_n\) and that \ref{cond::debiased} holds. By contrast, \ref{cond::oraclebiaslinear} is specific to adaptive DML. It ensures that the requirement $\|D_{n,P_0} - D_{0,P_0}\|_{P_0} = o_p(1)$ in \ref{cond::consDn} holds and that the model approximation error $\Psi_n(P_0) - \Psi_0(P_0)$ of Theorem~\ref{theorem::oraclebiaslinear} is negligible, thereby verifying \ref{cond::doubleRemrootnFinal}.

\begin{theorem}[Superefficiency for linear functionals]
\label{theorem::ALlinear}
 Assume \ref{cond::boundedlinearfun}-\ref{cond::oraclebiaslinear}.  Then, the ADML estimator $\widehat{\psi}_n$ is $P_0$-asymptotically linear, regular, and efficient for $\Psi_0$, with
\begin{equation*}
\widehat{\psi}_n - \Psi_0(P_0) = \frac{1}{n}\sum_{i=1}^n \left\{ \alpha_{0, \mathcal{H}_0}(A_i, W_i) \{Y_i - \mu_0(A_i, W_i)\} + m(W_i, \mu_0) - \Psi_0(P_0) \right\} + o_p(n^{-1/2}).
\end{equation*}
If, in addition, \(\mu_0 \in \mathcal{H}_0\), then the ADML estimator \(\widehat{\psi}_n\) is asymptotically linear for \(\Psi(P_0)\) with the same influence function. It is regular for $\Psi$ over the oracle submodel \(\mathcal{M}_0\), and if $E_0[\{Y - \mu_0(A,W)\}^2 \mid A, W]$ is almost surely constant, then it is semiparametrically efficient with respect to \(\mathcal{M}_0\).
\end{theorem}

Theorem \ref{theorem::ALlinear} implies that $\sqrt{n}\,( \widehat{\psi}_n - \psi_0)/\sigma_0 \rightarrow_d N(0, 1)$, where $\sigma^2_0$ equals the efficiency bound $\operatorname{var}_{0}\left\{D_{0,P_0}(O)\right\}$. As long as the nuisance estimation and model approximation errors decay sufficiently fast, this theorem establishes that the ADML estimator $\widehat{\psi}_n$ is $P_0$-superefficient for the parameter $\Psi(P_0)$, with a limiting variance that adapts to the complexity of the outcome regression $\mu_0$ through $\mathcal{H}_0$. The precise behavior of the estimator under sampling from local perturbations of $P_0$ follows from the general results in Section~\ref{section::oracleParamInference::original}. The ADML estimator exhibits double robustness \citep{bang2005doubly} in two senses: with respect to estimation error for the working nuisances under \(\mathcal{H}_n\), and with respect to model approximation error for the true nuisances under \(\mathcal{H}_0\). Hence, it allows slow estimation or approximation error for one nuisance function, provided the corresponding error for the other nuisance vanishes sufficiently quickly.

We now discuss the implications of Theorem~\ref{theorem::ALlinear} for our examples. Verification of the theorem's conditions in these settings is deferred to Appendix~\ref{appendix::conditionsexamples}. For Examples~\ref{example::map} and \ref{example::catecal}, the main condition to verify is Condition~\ref{cond::oraclebiaslinear}, which we control using Lemma~\ref{lemma::lipschitzdependent}.

\setcounter{example}{1}

\begin{example}[continued]

Specialized theory for this estimator under low-level conditions is given in Theorem~\ref{example::theorem::RlearnerLimitDistORACLE} of Appendix~\ref{section::DataAdaptexampleATEPartially}. To illustrate the advantages of ADML, we now consider the prespecified semiparametric estimator of the ATE under the homogeneous partially linear model \citep{robinson1988root, imbensOverlapEstimand2006, li2019overlapWeights, d2021overlap}, corresponding to the intercept model \(\mathcal{T}_{\mathrm{const}} := \{w \mapsto c : c \in \mathbb{R}\}\) for \(\tau_0\). This estimator is regular only under local perturbations that preserve a constant CATE; outside the homogeneous partially linear model, it is generally irregular and asymptotically biased. In contrast, the partially linear ADML estimator is regular over the submodel of distributions \(P\) such that \(\tau_P \in \mathcal{T}_0\), and is therefore regular over a richer class of distributions, provided \(\mathcal{T}_n\) and \(\mathcal{T}_0\) contain the constant functions \(\mathcal{T}_{\mathrm{const}}\). More generally, one can use a post hoc test for treatment effect homogeneity to choose between ADML, equivalently \(\mathcal{T}_n\), and the prespecified estimator, equivalently \(\mathcal{T}_{\mathrm{const}}\), \citep{crump2008nonparametric, dukes2024nonparametric}. Under suitable conditions, this hybrid ADML procedure is asymptotically equivalent to the prespecified estimator when treatment effects are homogeneous, while allowing strict improvement when treatment effects are heterogeneous. \qed
\end{example}

\begin{example}[continued]
Theorem~\ref{theorem::ALlinear} establishes nonparametric superefficiency of the calibrated plug-in regression estimator \(\widehat{\psi}_n = \frac{1}{n}\sum_{i=1}^n m(W_i, \mu_n^*)\), whose limiting variance attains the efficiency bound for \(\Psi\) under the model that uses \(\mu_0(A,W)\) as an oracle dimension reduction of \((A,W)\). It further shows that this estimator is regular for \(\Psi\) over the submodel of all \(P \in \mathcal{M}\) with \(P \in \mathcal{H}_0\), that is, distributions whose outcome regression is a transformation of \(\mu_0\). In the special case of a counterfactual mean, this recovers the superefficiency result of \citet{ATEsupereff} for a TMLE-type variant of the estimator. More generally, our theory also characterizes the estimator's local asymptotic behavior and regularity properties, which were not studied by \citet{ATEsupereff}. \qed
\end{example}

\begin{example}[continued]
Theorem~\ref{theorem::ALlinear} establishes nonparametric superefficiency of the calibrated plug-in R-learner \(\widehat{\psi}_n = \frac{1}{n}\sum_{i=1}^n \tau_n^*(W_i)\). As discussed in Section~\ref{section::EIFlinear}, overlap affects the variance only through the aggregated inverse-weighting factor \(\{E_0[(A-\pi_0(W))^2 \mid \tau_0(W)=\tau_0(w)]\}^{-1}\), and hence its asymptotic variance is less sensitive to limited treatment overlap than that of regular nonparametric estimators. The estimator is regular over all distributions for which the CATE is a transformation of \(\tau_0\).

An important special case is when the treatment effect is truly homogeneous under \(P_0\), that is, when the CATE \(\tau_0(W)\) is almost surely constant. In this setting, \(\mathcal{T}_0 \subseteq \{f \circ \tau_0 : f\}\) reduces to the class of constant functions \(\mathcal{T}_{\mathrm{const}}\), \(\mathcal{H}_0\) coincides with the homogeneous PLR model \citep{robinson1988root}, and \(\Psi_0\) reduces to the overlap-weighted ATE \citep{imbensOverlapEstimand2006, li2019overlapWeights}. As a result, the ADML estimator is locally asymptotically equivalent to the prespecified overlap-weighted effect estimators of \cite{robinson1988root,imbensOverlapEstimand2006, li2019overlapWeights} derived under this model. When treatment effect homogeneity fails, however, the ADML estimator continues to yield valid nonparametric inference while adapting to the degree of heterogeneity and retaining superefficiency. Thus, under the stated conditions, the prespecified estimator based on the homogeneous effect model is dominated by the ADML estimator in the following sense: when the homogeneous model is correct, the two estimators are asymptotically equivalent, whereas when it is misspecified, the prespecified estimator is generally biased while the ADML estimator remains nonparametrically valid.\qed
\end{example}

\begin{example}[continued]
Augmented minimax linear estimation improves on standard DML estimators by having error driven by the projected residual error $\|(\mu_n - \mu_0) - \Pi_{\mathcal{F}_n}(\mu_n - \mu_0)\|_{P_0}$. In particular, a variant of Theorem~\ref{theorem::oraclebiaslinear} for affine spaces shows that
\[
|\Psi_n(P_0) - \Psi_0(P_0)| \leq
\|\alpha_{0,\mathcal{F}_0} - \Pi_{\mathcal{F}_n}\alpha_{0,\mathcal{F}_0}\|_{P_0}
\,
\|(\mu_n - \mu_0) - \Pi_{\mathcal{F}_n}(\mu_n - \mu_0)\|_{P_0},
\]
where $\alpha_{0,\mathcal{F}_0}$ is the Riesz representer over $\mathcal{F}_0$ and $\Pi_{\mathcal{F}_n} : \mathcal{F} \to \mathcal{F}_n$ is the orthogonal projection operator.
\end{example}

\section{Numerical experiments}

\subsection{Data-generating distributions and nuisance estimation}

We conducted a simulation study to evaluate the performance of the ADML estimators for the ATE based on the plug-in regression estimator defined in \eqref{eqn::pluginadml} and the PLR estimator in Example~\ref{example::partially::intro}. Both ADML estimators use the relaxed highly adaptive Lasso (HAL) estimator \citep{vanderlaanGenerlaTMLEFIRST,HAL2016,bibautHAL} for estimation of the outcome regression and the CATE. HAL is based on a sectional variation norm penalty, which extends first-order total variation denoising to nonparametric settings \citep{geerLocalAdapt,fang2021multivariate,marswithLasso}. Using a tensor-product basis of piecewise linear hinge functions of the form \(x \mapsto (x-u)1(x \ge u)\) with knot point \(u \in \mathbb{R}\), it performs variable selection and adapts to sparse functions \citep{marswithLasso}. We implemented HAL using the R package \texttt{hal9001} \citep{hal2}, with the sectional variation norm tuning parameter selected by cross-validation. Code for both ADML estimators is available in the R package \texttt{causalHAL}. As nonadaptive benchmarks, we included a semiparametric ATE estimator based on the partially linear intercept model \citep{robinson1988root,imbensOverlapEstimand2006} and the nonparametric efficient augmented inverse probability-weighted (AIPW) estimator \citep{robinsCausal,robins1995analysis}.

For the simulation studies, we considered sample sizes $n \in \{500, 1000, 2000, 3000, 4000, 5000\}$ and independent covariates $W_1, W_2, W_3, W_4$ each drawn from the uniform distribution on $(-1,+1)$. Given $W=w:=(w_1,w_2,w_3,w_4)$, the treatment assignment $A$ was generated from a Bernoulli distribution with conditional mean $\pi_0(w)$ defined by $\text{logit}\{\pi_0(w)\}=\gamma \sum_{j=1}^4 \{w_j + \sin(4w_j)\}$, where $\gamma \in \{0.5, 1, 2\}$ controls the degree of treatment overlap. Given $(A,W)=(a,w)$, the outcome variable was generated from a normal distribution with mean $\mu_0(0,w)+a\tau_0(w)$ and variance $\sigma^2=0.5$, where $\mu_0(0,w)$ is the control conditional mean and $\tau_0(w) = 1 + w_1 + |w_2| + \cos(4w_3) + w_4$ is the CATE. We note that $\tau_0$ is approximately sparse under the HAL basis, implying potential superefficiency of the HAL-ADML estimators. Two choices of the control conditional mean were considered: the piecewise linear form  $\mu_0(0,w) =  w_1 + |w_2|  + w_3 + |w_4|$ and the nonlinear form $\mu_0(0,w) = \cos(4w_2) + \sum_{j=1}^4\sin(4w_j)$.  

To ensure comparability, we employed identical nuisance estimators for $\pi_0$, $\mu_0$ and $m_0$ across all four estimators. The outcome regression $\mu_0$ was estimated using the relaxed HAL least-squares estimator, with separate additive models and regularization parameters for $\mu_0(0,\cdot)$ and $\tau_0$. The number of prespecified basis functions included in the Lasso regression for $\mu_0$ were, respectively, $k=80, 400, 608, 608, 800,800$ for sample sizes $n=500, 1000, 2000, 3000, 4000, 5000$. To estimate the propensity score $\pi_0$, we used least-squares regression with 10-fold cross-validation employed to select among three candidate algorithms: generalized additive models implemented in R by the \texttt{mgcv} package \citep{hastie1987generalized, wood2001mgcv}, multivariate adaptive regression splines implemented by the \texttt{earth} package \citep{friedman1991multivariate, milborrow2019package}, and random forests implemented by the \texttt{ranger} package \citep{breiman2001random,wright2015ranger}. To ensure that the estimated propensity scores are bounded away from $0$ and $1$, we truncated estimates to fall within the range $(c_n,1-c_n)$, where $c_n$ is a data-adaptive cutoff selected by minimizing a loss function for the inverse propensity score \citep{chernozhukov2022automatic}. Finally, we estimated $m_0$ using the plug-in estimator $\pi_n \mu_n(1,\cdot) + (1-\pi_n) \mu_n(0,\cdot)$, where $\mu_n$ and $\pi_n$ are the estimators of $\mu_0$ and $\pi_0$ described above. 

\subsection{Experimental findings: demonstrating superefficiency}

To summarize overlap across settings, we report the overlap constant $c_0 := \inf_w \min\{\pi_0(w),\,1-\pi_0(w)\},$ which varies with \(\gamma\). For each estimator, we report Monte Carlo estimates of bias, variance, mean squared error, and confidence interval coverage. Figure~\ref{fig:simsEasyMain} presents results for the setting with a linear control conditional mean; results for the remaining settings are given in Appendix~\ref{appendix::figures} and are qualitatively similar.

 \begin{figure}[htb]
     \centering
      \begin{subfigure}[b]{0.48\linewidth} 
     \includegraphics[width=0.5\linewidth]{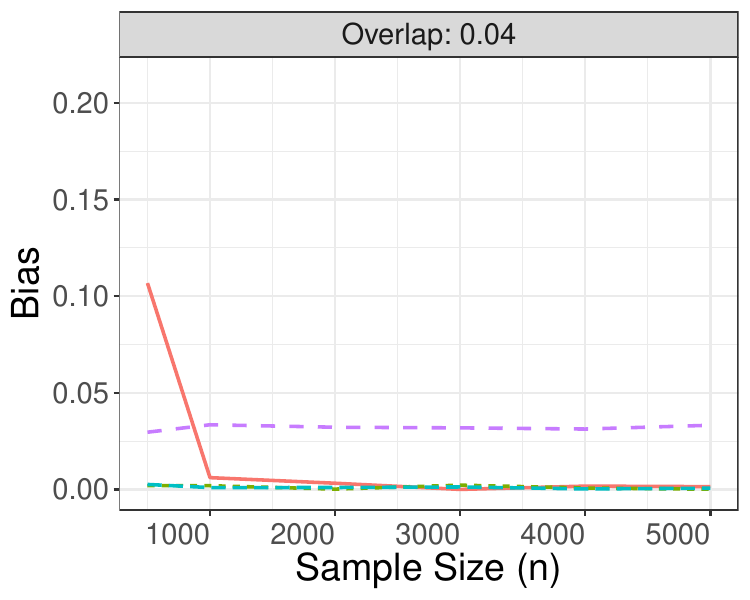}\includegraphics[width=0.5\linewidth]{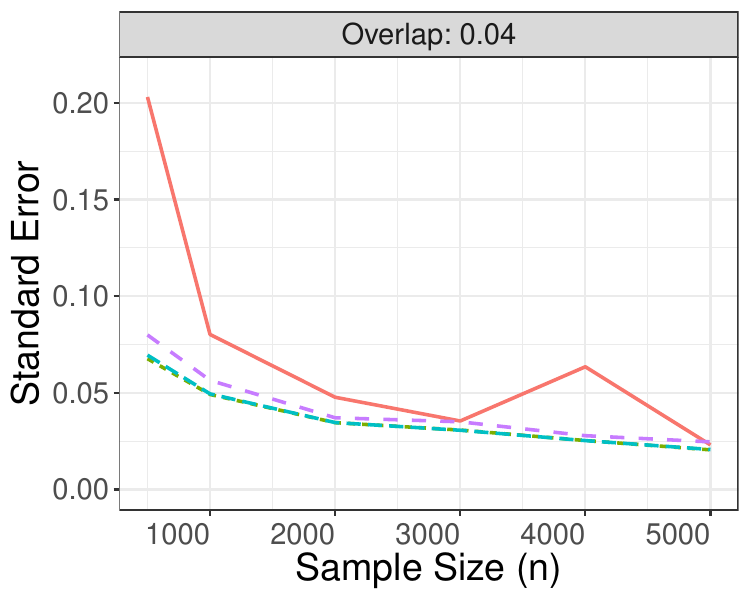}
     \includegraphics[width=0.5\linewidth]{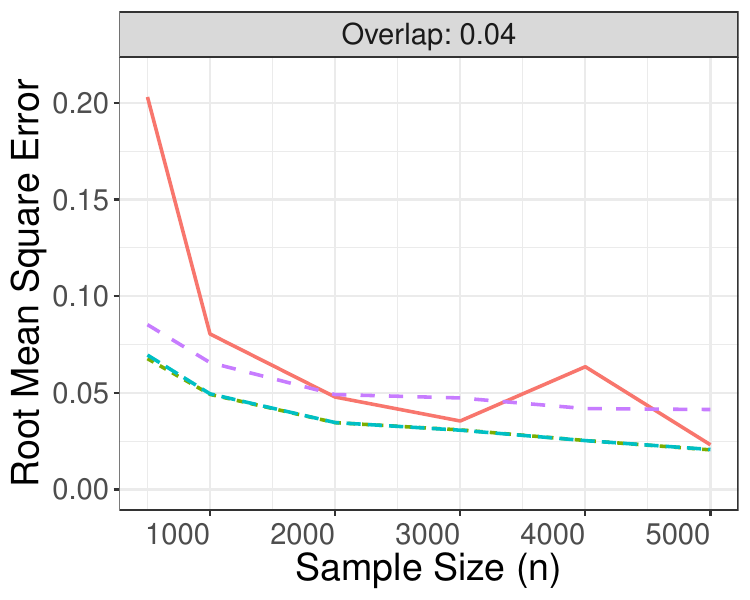}\includegraphics[width=0.5\linewidth]{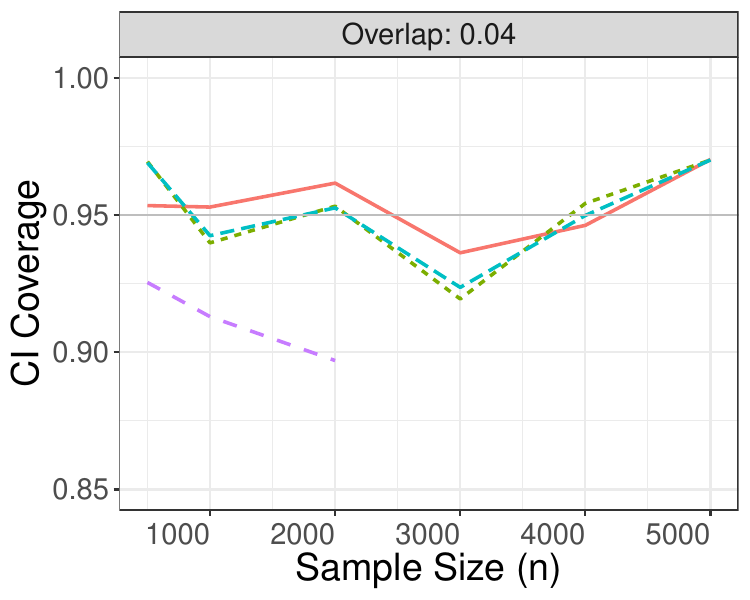}
     \subcaption{Linearity and moderate overlap  ($c_0 \approx 0.04$)}
     \end{subfigure} \hfill \begin{subfigure}[b]{0.48\linewidth} 
     \includegraphics[width=0.5\linewidth]{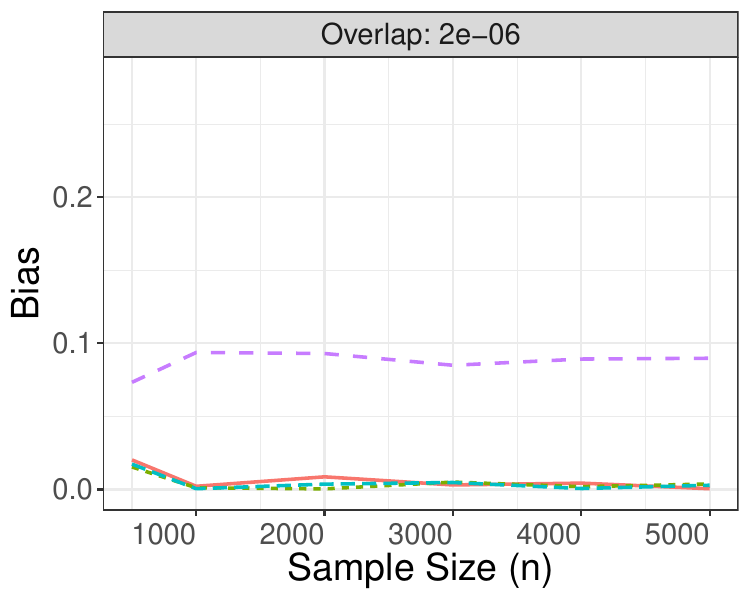}\includegraphics[width=0.5\linewidth]{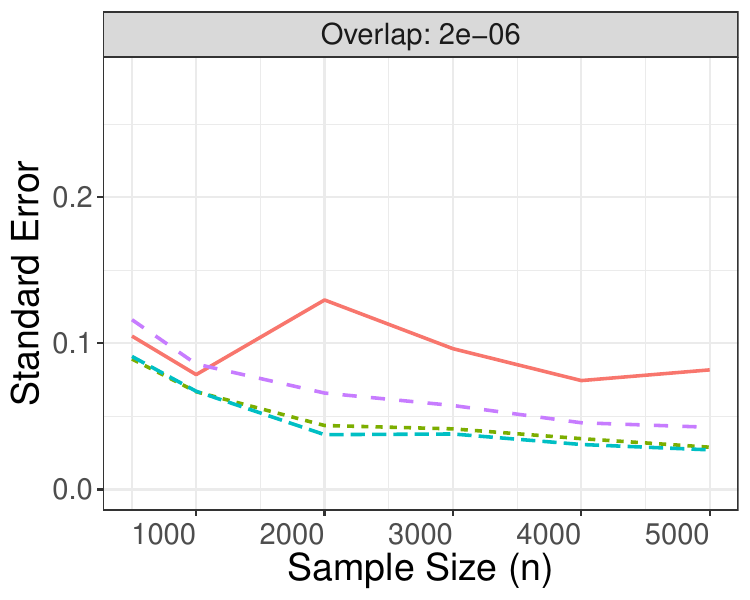}
     \includegraphics[width=0.5\linewidth]{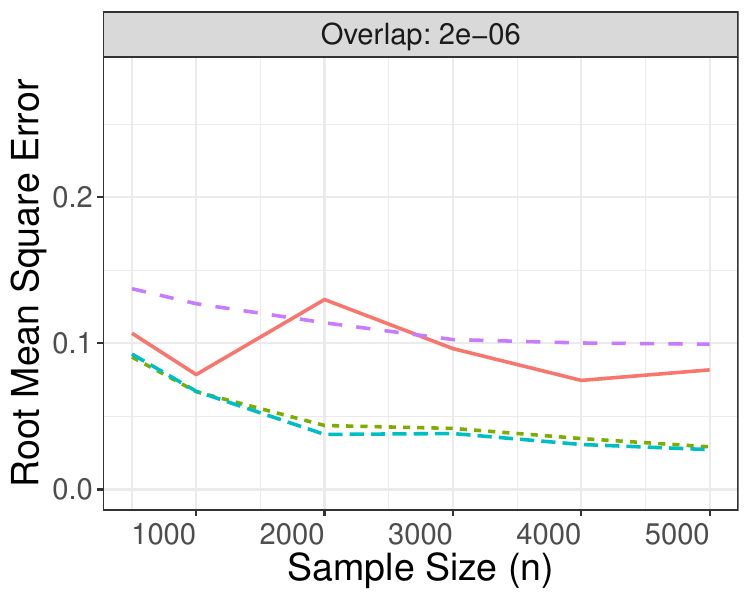}\includegraphics[width=0.5\linewidth]{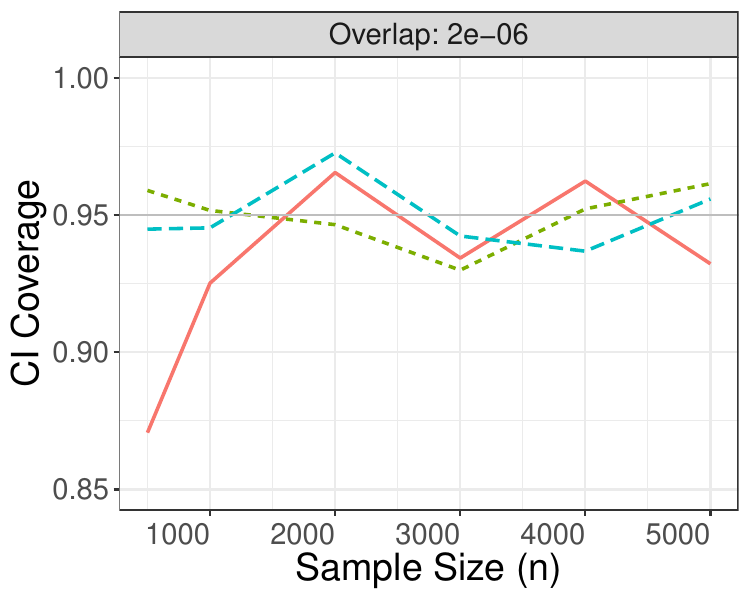}
      \subcaption{Linearity and limited overlap ($c_0 \approx 10^{-6}$)}
     \end{subfigure}
     
      \includegraphics[width=0.8\linewidth]{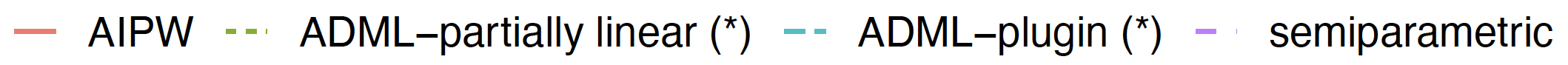}
     \caption{Comparison of empirical bias, standard error and root mean squared error of estimator, and coverage of nominal 95\% confidence interval for partially linear and plug-in HAL-ADML estimators, prespecified semiparametric estimator (assuming constant CATE), and nonparametric AIPW estimator, under sampling from a fixed distribution satisfying linearity and with varying degrees of treatment overlap. Coverage probabilities for intervals based on the prespecified semiparametric estimator were consistently poor, exceeding the y-axis range.}
     \label{fig:simsEasyMain}
 \end{figure}

 \begin{figure}[htb]
     \centering
     \begin{subfigure}[b]{0.48\linewidth} 
     \includegraphics[width=0.5\linewidth]{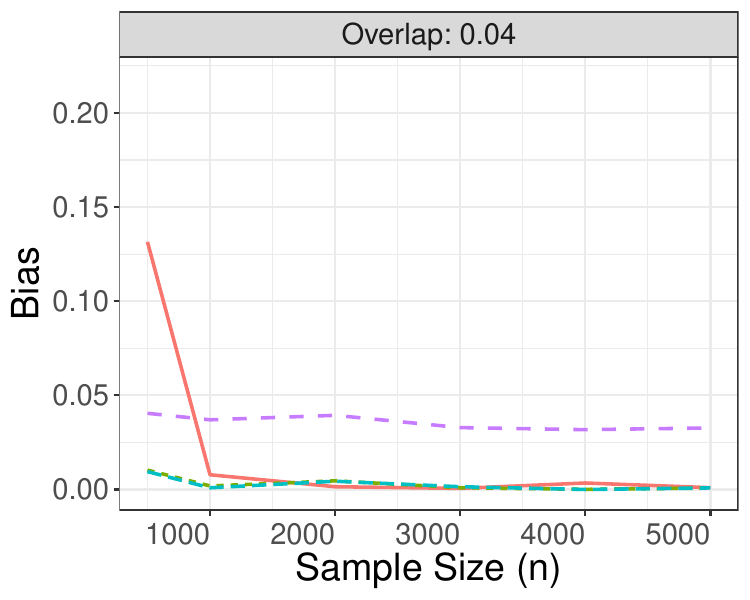}\includegraphics[width=0.5\linewidth]{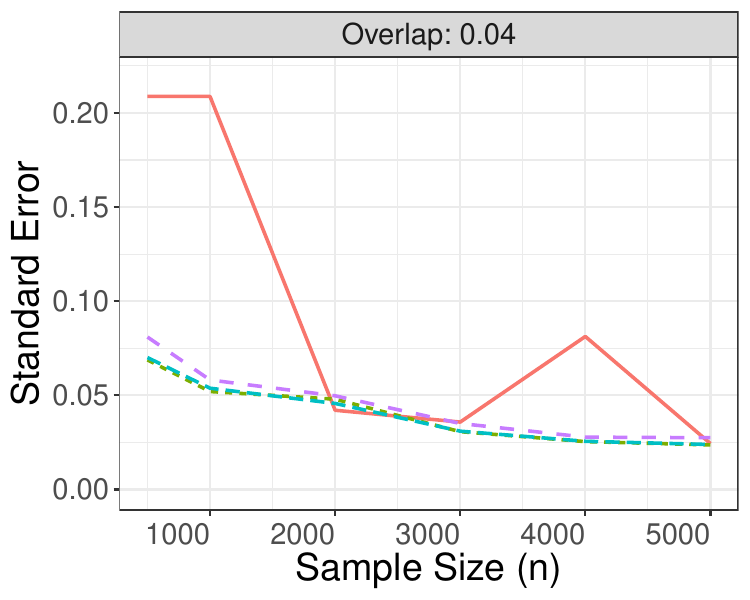}
     \includegraphics[width=0.5\linewidth]{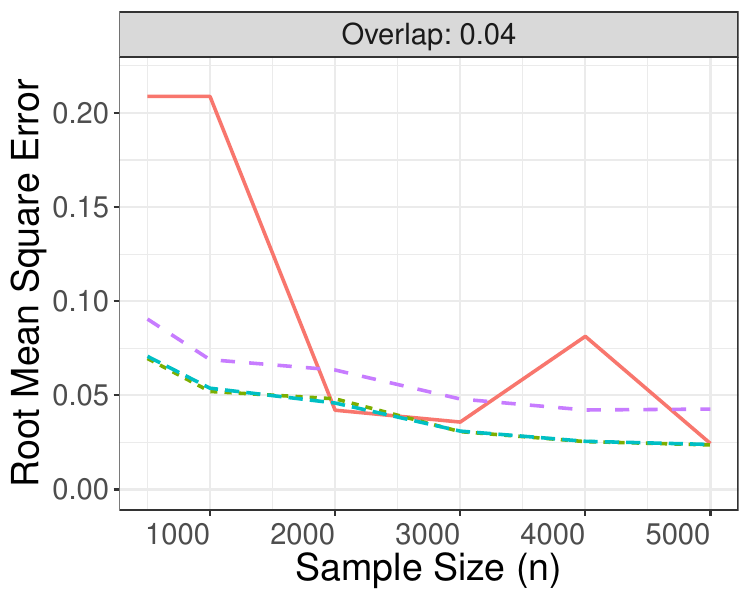}\includegraphics[width=0.5\linewidth]{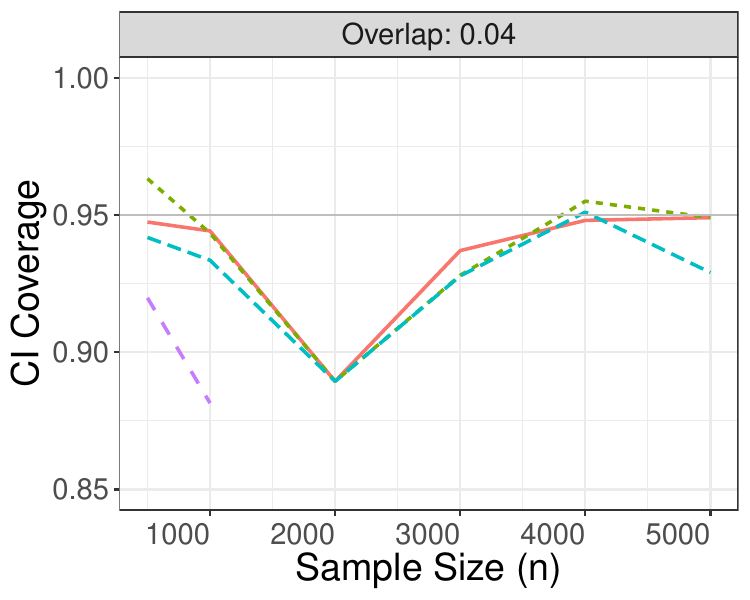}
     \subcaption{Nonlinearity and moderate overlap  ($c_0 \approx 0.04$)}
     \end{subfigure} \hfill \begin{subfigure}[b]{0.48\linewidth} 
     \includegraphics[width=0.5\linewidth]{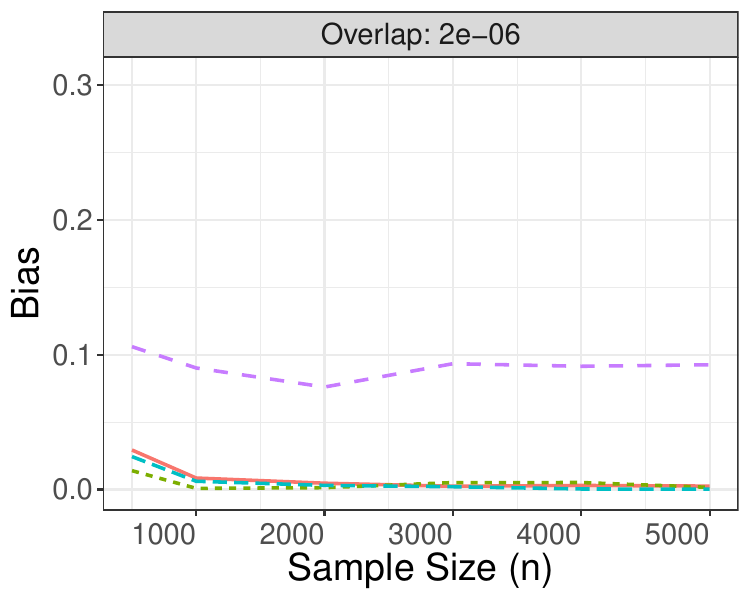}\includegraphics[width=0.5\linewidth]{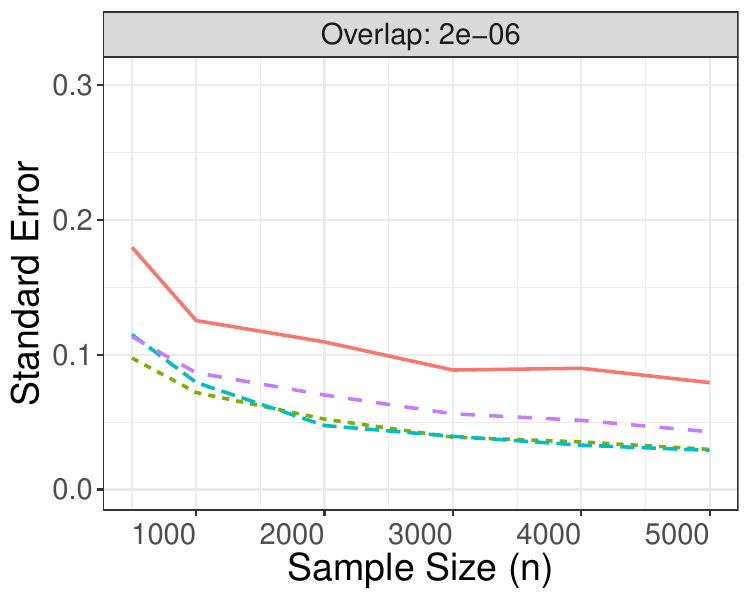}
     \includegraphics[width=0.5\linewidth]{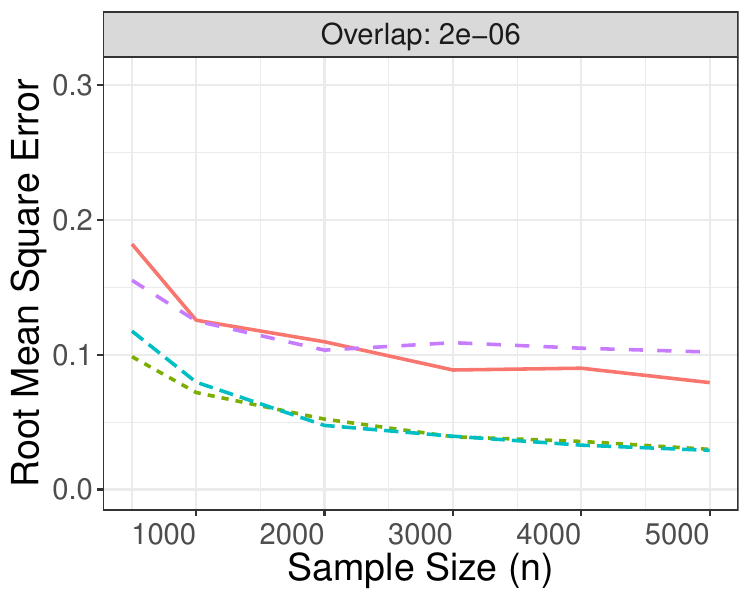}\includegraphics[width=0.5\linewidth]{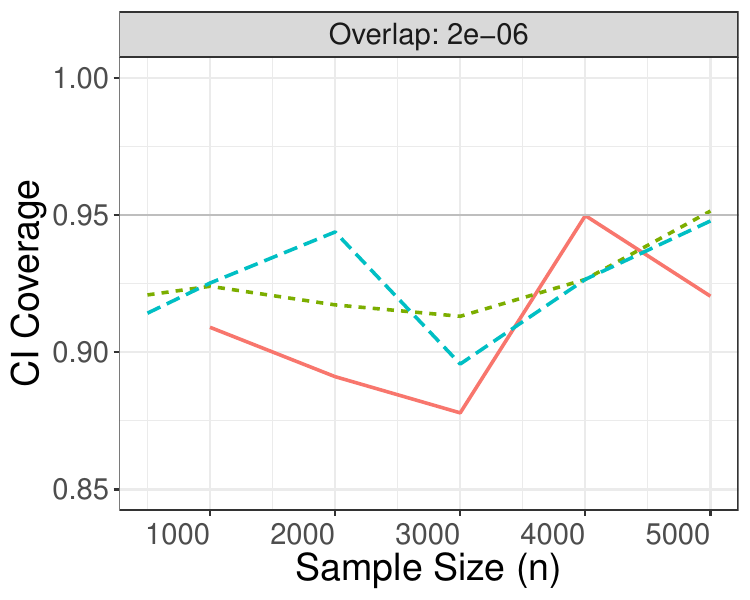}
      \subcaption{Nonlinearity and limited overlap ($c_0 \approx 10^{-6}$)}
     \end{subfigure}
      \includegraphics[width=0.8\linewidth]{plots/Legend.png}
     \caption{Comparison of empirical bias, standard error and root mean squared error of estimator, and coverage of nominal 95\% confidence interval for partially linear and plug-in HAL-ADML estimators, prespecified semiparametric estimator (assuming constant CATE), and nonparametric AIPW estimator, under sampling from a fixed distribution \emph{not} satisfying linearity and with varying degrees of treatment overlap. Coverage probabilities for intervals based on the prespecified semiparametric estimator were consistently poor, exceeding the y-axis range.}
     \label{fig:simsEasyMain}
 \end{figure}

The simulations are broadly consistent with the theory. In all settings considered, both ADML estimators improve upon the prespecified semiparametric estimator and the nonparametric AIPW estimator in terms of bias, variance, mean squared error, and coverage. In the sparse linear setting, the plug-in ADML estimator is typically more efficient than the partially linear ADML estimator, whereas in the nonlinear setting their variances are similar. These findings reflect the theoretical tradeoff: the plug-in estimator can be more (super)efficient in favorable settings, but it is regular over a smaller submodel than the partially linear estimator, which may help explain its greater finite-sample bias.

\subsection{Experimental findings: demonstrating irregularity}

In this experiment, we study the performance of the estimators under a least-favorable local perturbation \(P_{0,h n^{-1/2}}\) of \(P_0\) for the ATE in the nonparametric model. To construct this perturbation, we modify the outcome regression through both the control mean \(\mu_0(0,\cdot)\) and the CATE \(\tau_0\). Specifically, for a given sample size \(n\), we define $\mu_{n,0}(0,w) := \mu_0(0,w) - n^{-1/2}\{1-\pi_0(w)\}^{-1}$
and $\tau_{n,0}(w) := 1 + n^{-1/2}\bigl[\pi_0(w)\{1-\pi_0(w)\}\bigr]^{-1},$
while leaving all other components of the data-generating distribution unchanged.

Under this perturbation, the oracle submodel \(\mathcal{M}_0\) for the unperturbed distribution \(P_0\) is the partially linear intercept model, and the prespecified semiparametric estimator based on the intercept CATE model remains correctly specified to first order. Thus, Corollary~\ref{cor::CATEinf} implies that the partially linear ADML estimator should be asymptotically equivalent to the prespecified semiparametric estimator under sampling from \(P_{0,h n^{-1/2}}\). Results for the linear setting under moderate and limited overlap are shown in Figure~\ref{fig:simsEasyMain}; results for the remaining settings are qualitatively similar and are reported in Appendix~\ref{appendix::figures}.

 \begin{figure}[htb]
     \centering
   \begin{subfigure}[b]{0.48\linewidth} 
     \includegraphics[width=0.5\linewidth]{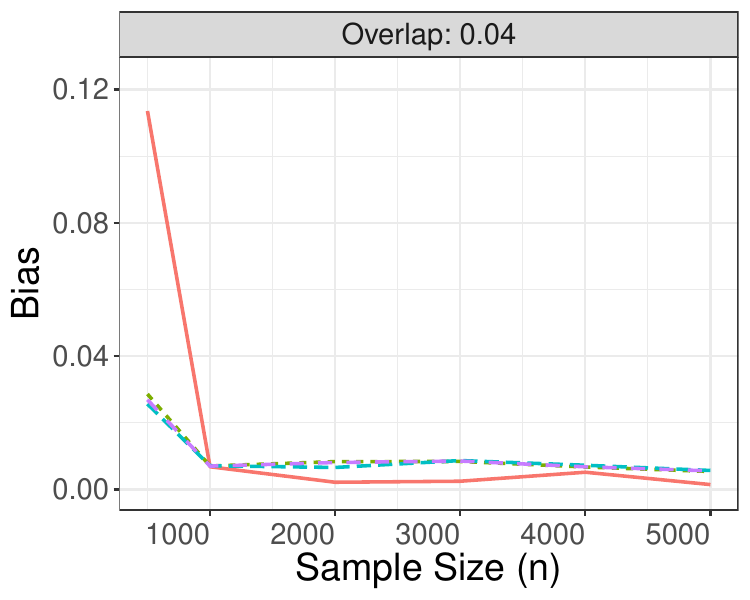}\includegraphics[width=0.5\linewidth]{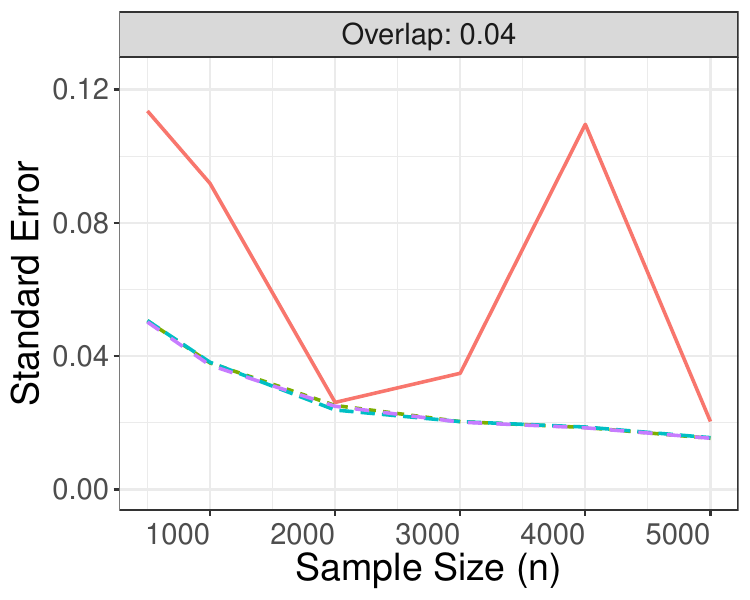}
     \includegraphics[width=0.5\linewidth]{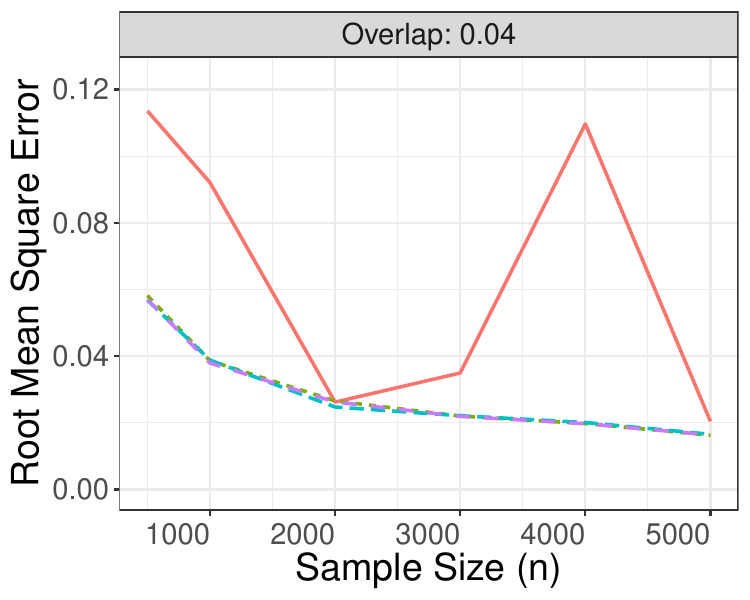}\includegraphics[width=0.5\linewidth]{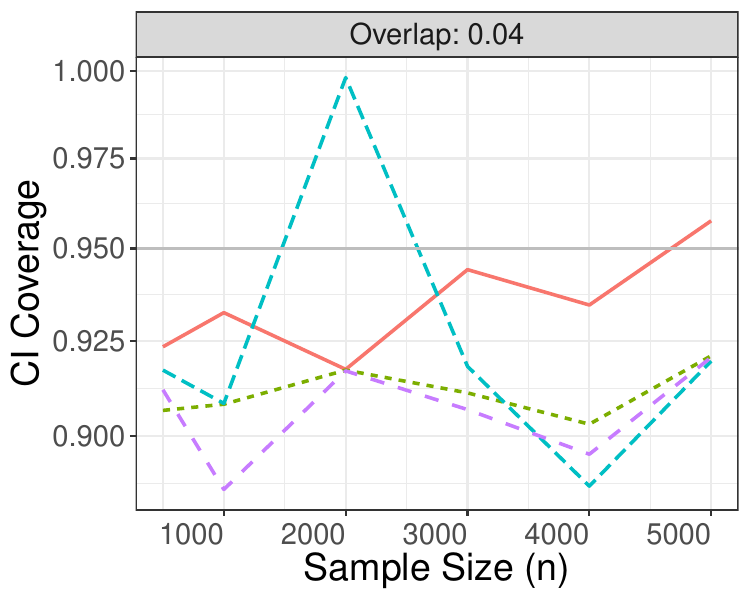}
     \subcaption{Linearity and moderate overlap  ($c_0 \approx 0.04$)}
     \end{subfigure} \hfill \begin{subfigure}[b]{0.48\linewidth} 
     \includegraphics[width=0.5\linewidth]{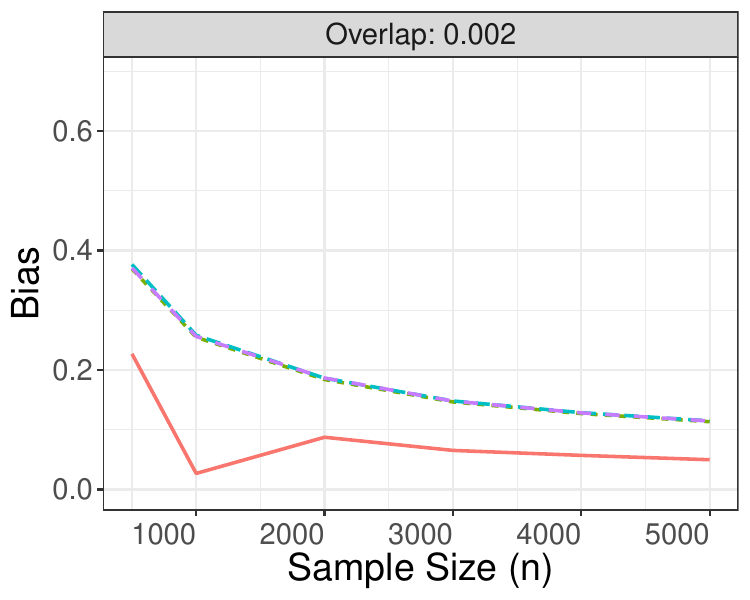}\includegraphics[width=0.5\linewidth]{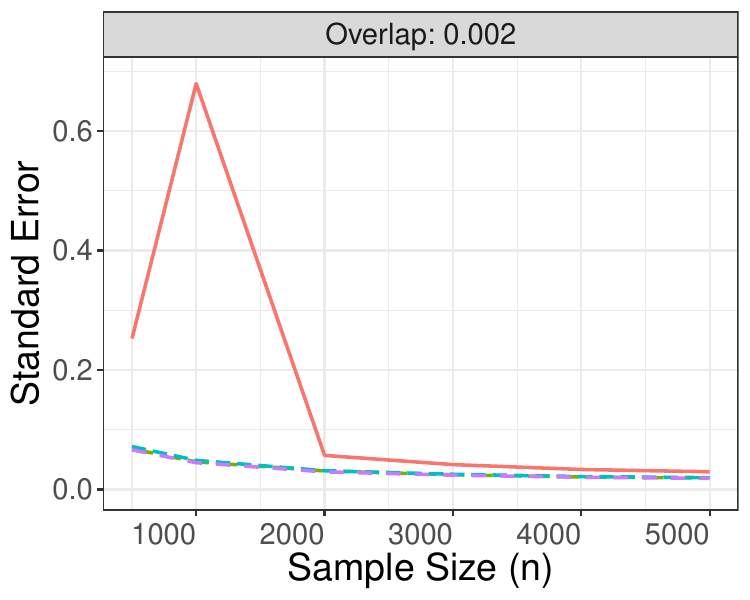}
     \includegraphics[width=0.5\linewidth]{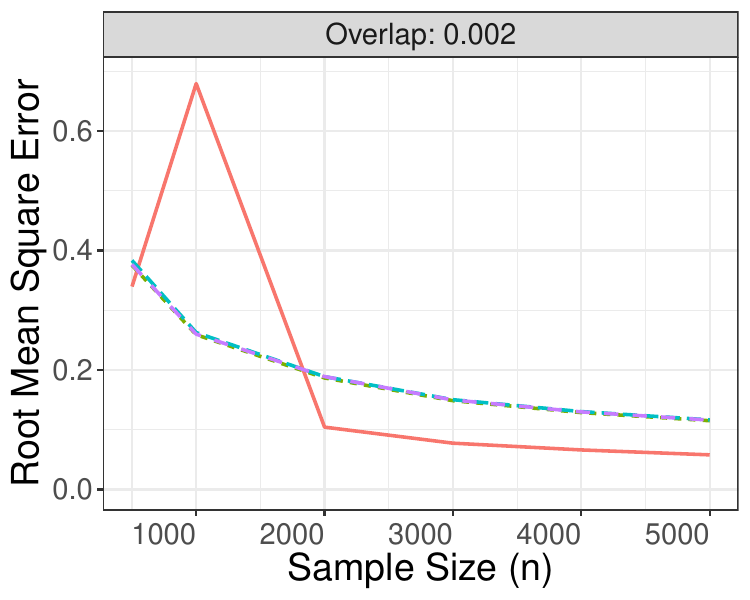}\includegraphics[width=0.5\linewidth]{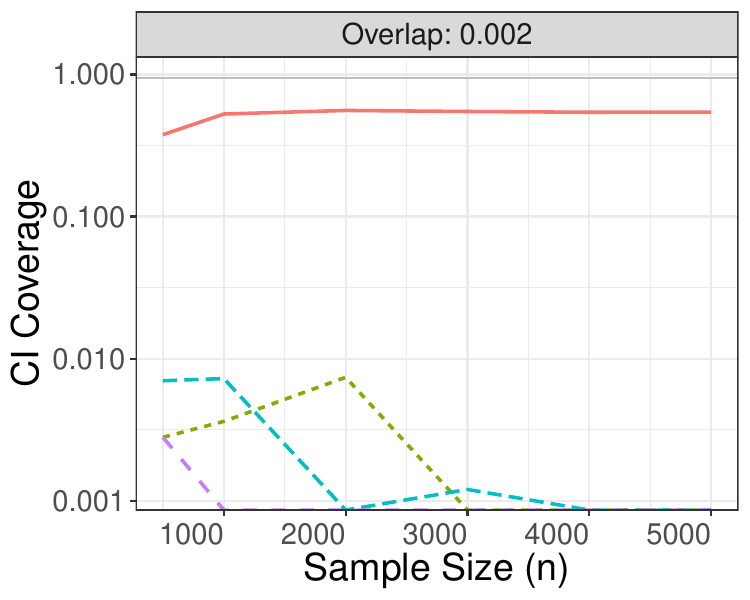}
      \subcaption{Linearity and limited overlap  ($c_0 \approx 0.002$)}
     \end{subfigure}  
     
      \includegraphics[width=0.8\linewidth]{plots/Legend.png}
     \caption{Comparison of empirical bias, standard error and root mean squared error of estimator, and coverage of nominal 95\% confidence interval for partially linear and plug-in HAL-ADML estimators, prespecified semiparametric estimator (assuming constant CATE), and nonparametric AIPW estimator, under sampling from a least-favorable local perturbation of a distribution satisfying linearity in outcome regression and with varying degrees of treatment overlap. Coverage probabilities for intervals based on the prespecified semiparametric estimator were consistently poor, exceeding the y-axis range.}
     \label{fig:simsEasyMain}
 \end{figure}
 
We find empirical support for the theoretical predictions under least-favorable local perturbations. In particular, the prespecified semiparametric estimator and both ADML estimators exhibit irregularity relative to the nonparametric model, resulting in nonvanishing asymptotic bias, while all estimators remain \(\sqrt{n}\)-consistent. Consistent with Corollary~\ref{cor::CATEinf}, the prespecified semiparametric estimator based on the intercept CATE model appears asymptotically equivalent to the partially linear ADML estimator, suggesting that there is no asymptotic cost to learning the oracle submodel from the data rather than knowing it in advance, even under the least-favorable local alternative. By contrast, confidence intervals based on the AIPW estimator attain approximately 95\% coverage under the least-favorable perturbation because of its regularity, but only at the cost of substantially higher variance. Moreover, in the moderate-overlap setting, the AIPW estimator has worse mean squared error and only marginally better coverage. When the overlap constant is \(c_0 \approx 10^{-6}\), the AIPW estimator is both biased and highly variable, likely reflecting the weak identifiability of the ATE in the nonparametric model. Finally, in the linear setting, the plug-in ADML estimator exhibits greater asymptotic bias than the partially linear ADML estimator, as expected given that it is regular over a smaller oracle submodel.

\section{Conclusion and practical considerations}

Adaptive debiased machine learning provides a general framework for constructing asymptotically linear and superefficient estimators of pathwise differentiable parameters using data-driven model selection. In this work, we showed that ADML estimators are regular and yield locally uniformly valid inference for an oracle projection parameter that coincides with the target parameter on the oracle submodel. Thus, data-driven model selection is asymptotically free in a local sense over the nonparametric model: an ADML procedure that adaptively learns a model from the data incurs no asymptotic loss relative to an oracle procedure that knows the model in advance. Although we focused on the \emph{iid} setting, the arguments extend naturally to dependent data under suitable empirical process conditions and central limit theorems; see, for example, \cite{van2014causal}.

Although data-driven model selection does not affect the first-order asymptotics of the ADML estimator for the oracle target, it may still inflate finite-sample variability and distort confidence interval coverage. One way to reduce this effect is sample splitting: estimate the working model on one half of the sample and construct the ADML estimator on the other, then reverse the roles of the two halves and average the resulting estimators. Our theory applies directly to this split-averaged estimator and yields asymptotic linearity and efficiency for the oracle target. This idea extends naturally to multi-fold sample splitting via cross-fitting \citep{vanderLaanRose2011, DoubleML}. Alternatively, one may account for the additional finite-sample variability induced by model selection using bootstrap \citep{efron1994introduction, bootstrapHAL, rinaldo2019bootstrapping} or subsampling \citep{guo2023rank} methods for variance estimation. Conditional selective inference methods may also improve finite-sample performance \citep{leeb2009conditional, condSelectInf, taylor2015statistical, zhao2017selective}, though they typically rely on strong assumptions, such as Gaussian errors or design, and their extension to general pathwise differentiable functionals remains unclear.

A stronger guarantee than our local oracle-equivalence result would be uniform asymptotic equivalence over all distributions that yield the same learned oracle model. For the Lasso, such a result may be obtainable from existing uniform convergence results over approximately sparse models \citep{belloni2013honest, belloni2014inference, belloni2013least, chernozhukov2018auto}. More generally, however, Wald-type inference for the original parameter \(\Psi\) based on ADML is not expected to be uniformly valid over the full nonparametric model, since adaptive estimators are inherently irregular \citep{LeebModelSelect2005, oracleselectionHodgebad}. Consequently, reliable coverage may require larger sample sizes than for prespecified DML procedures, so that the working model is learned stably and the additional finite-sample variability induced by model selection is negligible.

These considerations suggest that, in practice, conservative forms of ADML may be preferable, avoiding overly small working models and thereby preserving regularity over a broader class of distributions. In particular, we have shown that ADML can be used to construct estimators that, in a local asymptotic sense, match or improve upon prespecified parametric or semiparametric estimators while remaining valid under broader forms of model misspecification. One way to do so is to anchor ADML at a prespecified model and then learn a data-adaptive extension, as in our semiparametric ADML estimators anchored at the homogeneous partially linear model. More generally, irregularity is not unique to ADML. Prespecified parametric and semiparametric estimators are also irregular, in that they are asymptotically biased under local perturbations that leave the assumed model. Nonparametric regular estimators may fail to exist when the parameter is not nonparametrically pathwise differentiable. Moreover, even when nonparametric estimators exist, they may be irregular when identification is weak, for example because of insufficient overlap \citep{mehrabi2024off}. Finally, although local asymptotic behavior is useful for comparing estimators and characterizing efficiency, finite-sample performance is ultimately what matters \citep{chernozhukov2023simple}. From this perspective, irregular estimators may outperform regular ones when the sample size is large enough for the asymptotic approximation to be informative for the sampling distribution of interest.

Finally, since the posting of this work on arXiv, ADML has been applied in several settings. ADML with working models learned via the Highly Adaptive Lasso is studied in \cite{li2025regularized}. The use of ADML to learn working models for confounding-robust fusion of randomized trial data with potentially biased observational data is studied in \cite{van2024adaptive}. The extension of the calibrated plug-in regression estimator of Example~\ref{example::map} to sequential decision-making in Markov decision processes is studied in \cite{van2025automatic}.

\vspace{0.3in}
\singlespacing

\section*{Acknowledgments}
Research reported in this publication was supported in part by grants DP2-LM013340 (AL), R01-HL137808 (MC) and R01-AI074345 (MvdL) from the National Institutes of Health.

\vspace{.3in}
 
 \bibliography{ref}

\appendix

\begingroup
\setcounter{tocdepth}{2} 
\renewcommand{\contentsname}{Appendix contents}
\tableofcontents
\endgroup

\section{Notation table}
\label{app:notation}

\singlespacing

\begin{table}[H]
\centering
\caption{Summary of notation I}
\label{tab:notation1}
\begin{tabular}{p{0.22\textwidth}p{0.72\textwidth}}
\toprule
Notation & Meaning \\
\midrule
\(O_1, \ldots, O_n\) & Independent and identically distributed observations drawn from the data-generating distribution \(P_0\). \\

\(P_0\) & Data-generating distribution. \\

\(\mathcal{M}_{\mathrm{np}}\) & Nonparametric model containing \(P_0\). \\

\(\mathcal{M}\) & Prespecified statistical model, with \(P_0 \in \mathcal{M} \subseteq \mathcal{M}_{\mathrm{np}}\). \\

\(\Psi\) & Target parameter mapping from a statistical model to \(\mathbb{R}\). \\

\(\psi_0\) & Target value at \(P_0\), defined by \(\psi_0 := \Psi(P_0)\). \\

\(\mathcal{M}_n\) & Data-adaptive working model selected from the data. \\

\(\mathcal{M}_0\) & Oracle submodel, that is, an unknown fixed submodel containing \(P_0\) that is approximated by \(\mathcal{M}_n\). \\

\(\Pi_n\) & Loss-based projection operator onto \(\mathcal{M}_n\). \\

\(\Pi_0\) & Loss-based projection operator onto \(\mathcal{M}_0\). \\

\(\Psi_n\) & Working projection parameter, defined by \(\Psi_n := \Psi \circ \Pi_n\). \\

\(\Psi_0\) & Oracle projection parameter, defined by \(\Psi_0 := \Psi \circ \Pi_0\). \\

\(\widehat{\psi}_n\) & ADML estimator of the target parameter. \\
\bottomrule
\end{tabular}
\end{table}

\begin{table}[H]
\centering
\caption{Summary of notation II}
\label{tab:notation2}
\begin{tabular}{p{0.22\textwidth}p{0.72\textwidth}}
\toprule
Notation & Meaning \\
\midrule
\(D_{n,P}\) & Efficient influence function of \(\Psi_n\) at \(P\). \\

\(D_{0,P}\) & Efficient influence function of \(\Psi_0\) at \(P\). \\

\(\mathcal{M}_{n,0}\) & Union model used to compare the working and oracle models. \\

\(\Pi_{n,0}\) & Loss-based projection operator onto \(\mathcal{M}_{n,0}\). \\

\(\Psi_{n,0}\) & Extended projection parameter induced by \(\mathcal{M}_{n,0}\), defined by \(\Psi_{n,0} := \Psi \circ \Pi_{n,0}\). \\

\(P_{n,0}\) & Projection of \(P_0\) onto \(\mathcal{M}_n\), defined by \(P_{n,0} := \Pi_n P_0\). \\

\(D_{n,0,P_{n,0}}\) & Efficient influence function of \(\Psi_{n,0}\) evaluated at \(P_{n,0}\). \\

\(\bar D_{n,0,P_{n,0}}\) & Approximation to \(D_{n,0,P_{n,0}}\) lying in the loss-based tangent space of the working model. \\

\(T_{\mathcal{M}}(P)\) & Tangent space of the model \(\mathcal{M}\) at \(P\). \\

\(\ell(\cdot, Q)\) & Loss function used to define projections onto the working and oracle models. \\

\(P_n\) & Empirical distribution of the observed data. \\

\(\sigma_n^2\) & Influence-function-based variance estimate, typically \(n^{-1} \sum_{i=1}^n D_{n,\widehat P_n}(O_i)^2\). \\
\bottomrule
\end{tabular}
\end{table}

\begin{table}[H]
\centering
\caption{Summary of notation III}
\label{tab:notation3}
\begin{tabular}{p{0.22\textwidth}p{0.72\textwidth}}
\toprule
Notation & Meaning \\
\midrule
\(W\) & Covariates in the regression and causal inference examples. \\

\(A\) & Binary treatment variable. \\

\(Y\) & Outcome variable. \\

\(\mu_P(a,w)\) & Outcome regression under \(P\), defined by \(\mu_P(a,w) := E_P(Y \mid A = a, W = w)\). \\

\(\mu_0\) & True outcome regression, defined by \(\mu_0 := \mu_{P_0}\). \\

\(m(W,\mu)\) & Known mapping, linear in \(\mu\), used to define linear functionals of the regression. \\

\(\tau_P(w)\) & Conditional average treatment effect under \(P\), defined by \(\tau_P(w) := \mu_P(1,w) - \mu_P(0,w)\). \\

\(\tau_0\) & True conditional average treatment effect. \\

\(\pi_P(w)\) & Propensity score under \(P\), defined by \(\pi_P(w) := P(A = 1 \mid W = w)\). \\

\(\pi_0\) & True propensity score. \\

\(m_P(w)\) & Marginal regression, defined by \(m_P(w) := E_P(Y \mid W = w)\). \\

\(\mathcal{H}\) & Ambient regression function space. \\

\(\mathcal{H}_n\) & Data-adaptive working regression model. \\

\(\mathcal{H}_0\) & Oracle regression model. \\

\(\alpha_{P,\mathcal{H}}\) & Riesz representer of the linear functional \(\mu \mapsto E_P\{m(W,\mu)\}\) over \(\mathcal{H}\). \\
\bottomrule
\end{tabular}
\end{table}

\doublespacing

\section{Expanded related work}
 
  \label{section::AMLEGeneral3Lit}

\textit{Post-hoc model selection and inference.} The effect of data-adaptive model selection on inference has been studied extensively; see, for example, \cite{bauer1988model, potscher1991effects, buhlmann1999efficient, hjort2003frequentist, bunea2004consistent, LeebModelSelect2005, claeskens2007asymptotic}. A central lesson from this literature is that adaptivity often comes at the cost of regularity. In particular, superefficient procedures can achieve lower variance at selected distributions, but this typically requires sacrificing uniform validity over local perturbations. This tradeoff may nevertheless be worthwhile in settings where regular nonparametric inference is unavailable or too unstable to be useful, such as ATE estimation under limited or no overlap \citep{moosavi2023costs}.

Some of the existing literature studies estimators based on consistent model selection procedures that aim to select a correctly specified model with probability tending to one, a property often called the oracle property \citep{buhlmann1999efficient, oracleSelectSCAD, LeebModelSelect2005, AdaptLassoOracle, AdaptLassoOracleStationary}. However, this approach has been criticized because it can perform poorly when the selected model is incorrect or only approximately correct, and because very large sample sizes may be needed before correct model selection becomes likely \citep{LeebModelSelect2005}. ADML relaxes this requirement: rather than requiring consistent selection of the true model, it only requires the selected working model to approximate a fixed oracle submodel at the relevant sample size. Similar relaxations appear in post-Lasso inference under approximate sparsity \citep{belloni2012sparse, belloni2013honest, belloni2014inference}. A second criticism of superefficient estimators is that inference for the original target parameter may fail to be locally uniformly valid over the full nonparametric model \citep{LeebModelSelect2005, adaptLassoBadPerf, wu2019hodges}. ADML shares this limitation for the original target \(\Psi\). However, in Section~\ref{section::oracleParamInference} we show that ADML does yield locally uniformly valid inference for the oracle projection parameter \(\Psi_0\), and, for the original parameter \(\Psi\), it remains locally uniformly valid along perturbations within the oracle submodel \(\mathcal{M}_0\).

\textit{Selective inference.} Selective inference concerns inference conducted after examining the data, especially in high-dimensional settings with data-driven model selection \citep{modelselectionbias, zhang2014confidence, lee2016exact, zhao2017selective, danielleWittenLassoWorks, PostselectionInference}. Most prior work has focused on regression coefficients indexed by a selected model. That problem is inherently difficult because such targets are typically not pathwise differentiable, leading to irregular behavior and nonstandard rates of convergence \citep{potscherAdaptLasso2009distribution, adaptLassoBadPerf, cai2017confidence, yang2021rates}. By contrast, our work shows that post-selection inference for pathwise differentiable functionals in general statistical models is considerably more tractable: the error induced by model selection is second order and is therefore asymptotically negligible under suitable conditions. Exploiting this smoothness, we derive \(\sqrt{n}\)-consistent, asymptotically linear estimators under high-level conditions on black-box model-selection procedures. In particular, we establish the local asymptotic validity of model-based inference procedures that ignore variation from the model-selection step.

 \textit{Data-adaptive target parameters.} Our work advances the literature on inference for data-adaptive target parameters by providing asymptotically normal estimators for a broad class of parameters defined through projections onto data-dependent working models \citep{van2013AdaptTarget, dataAdaptTargetParam, aronow2016data, rinaldo2019bootstrapping}. In some settings, inference for a data-adaptive parameter also yields valid inference for a fixed population parameter because the bias of the data-dependent estimand relative to the fixed target is negligible, as in estimation of the causal effect of the optimal dynamic treatment rule \citep{van2015OptRule} and certain measures of variable importance \citep{williamson2021nonparametric}. In Section \ref{section::modelapprox}, we extend this principle by characterizing the model approximation error and giving conditions under which inference for a data-adaptive projected target remains valid for the corresponding fixed oracle target.

 \textit{Calibration in causal inference.} Recent work has used calibration to improve the performance of DML estimators \citep{gutman2022propensity, van2023causal, deshpande2023calibrated, ballinari2024improving, van2024automatic, van2024stabilized, rabenseifner2025calibration}. For DML estimation of continuous linear functionals of a regression, \citet{van2024automatic} proposed \textit{calibrated DML} and showed that calibrating the regression and Riesz representer yields estimators that remain asymptotically normal even when one nuisance is estimated inconsistently or at an arbitrarily slow rate, provided the other converges at the \(n^{-1/4}\) rate. In this setting, calibration delivers not only doubly robust point estimation but also doubly robust inference, including for confidence intervals. Building on this literature, we show in Example~\ref{example::map} that calibrating a single nuisance in a plug-in estimator yields superefficiency for such functionals. We further show that this idea extends beyond regression, as illustrated by CATE calibration in Example~\ref{example::catecal}.

\section{Additional details for ADML}

 \subsection{Details on applying  Theorem~\ref{theorem::exactBiasOracle}}
 \label{appendix::modelapproxproj}

 Theorem~\ref{theorem::exactBiasOracle} is best applied with \(\bar D_{n,0,P_{n,0}}\) chosen either as \(D_{n,P_{n,0}}\), the \(P_{n,0}\)-EIF of \(\Psi_n\), or as the \(L^2(P_{n,0})\)-projection of \(D_{n,0,P_{n,0}}\) onto the loss-based tangent space \(\mathcal{S}_{\mathcal{M}_n}(P_{n,0})\); for the log-likelihood loss, these two choices coincide. The former choice can often be re-expressed as an appropriate projection of \(D_{n,0,P_{n,0}}\) onto \(\mathcal{S}_{\mathcal{M}_n}(P_{n,0})\). In particular, the EIF--Riesz characterization in Theorem~\ref{theorem::EIFmain} shows that \(D_{n,0,P_{n,0}} = -\dot{\ell}_{P_{n,0}}(s_{n,0,P_{n,0}})\) for a union-model Riesz representer \(s_{n,0,P_{n,0}} \in T_{\mathcal{M}_{n,0}}(P_{n,0})\), defined as in \eqref{eqn::autodml}. Similarly, \(D_{n,P_{n,0}} = -\dot{\ell}_{P_{n,0}}(s_{n,P_{n,0}})\) for \(s_{n,P_{n,0}} \in T_{\mathcal{M}_n}(P_{n,0})\). Since \(T_{\mathcal{M}_n}(P_{n,0}) \subseteq T_{\mathcal{M}_{n,0}}(P_{n,0})\), projection properties of Riesz representers imply that \(s_{n,P_{n,0}}\) is the projection of \(s_{n,0,P_{n,0}}\) onto \(T_{\mathcal{M}_n}(P_{n,0})\) with respect to the Hessian inner product \(P_0 \ddot{\ell}_{P_{n,0}}(\cdot,\cdot)\) \citep{bickel1993efficient}. In particular, under mild conditions, Lemma~\ref{lemma::stableLogLik} in Appendix~\ref{appendix::stable} shows that \(\|D_{n,0,P_{n,0}} - D_{n,P_{n,0}}\|_{P_0} = o_p(1)\) whenever the tangent-space approximation error vanishes:
\[
\inf_{s \in T_{\mathcal{M}_n}(P_{n,0})} P_0 \ddot{\ell}_{P_{n,0}}(s_{n,0,P_{n,0}} - s,\; s_{n,0,P_{n,0}} - s) = o_p(1).
\]

\subsection{Sufficient conditions for Condition~\ref{cond::consDn}}

  \label{appendix::stable}

We conclude this subsection with the following lemma, which gives sufficient conditions for Condition~\ref{cond::consDn}. Recall that $D_{n,0,P_{n,0}} = -\dot{\ell}_{P_{n,0}}(s_{n,0,P_{n,0}})$ is the EIF of the projection parameter $\Psi_{n,0}$ at $P_{n,0}$ for the union model $\mathcal{M}_{n,0}$, where $s_{n,0,P_{n,0}}$ is an element of the union tangent space $T_{\mathcal{M}_{n,0}}(P_{n,0})$. The lemma is stated under the following condition:

\begin{enumerate}[label=\textbf{\ref{cond::consDn}*}), ref={\ref{cond::consDn}*}]
\item \label{cond::suffcondAL} Each of the following hold:

\begin{enumerate}[label={ \alph*.}, ref=\ref{cond::suffcondAL}\alph*,leftmargin=15pt]
\item \textit{Lipschitz loss derivative:} \label{cond::lipschitzloss} For $P' \in \{P_{n,0}, P_0\}$, there exists $C < \infty$ such that 
$$
\|\dot{\ell}_{P'}(s') - \dot{\ell}_{P'}(s)\|_{P_0} 
   \leq C \Bigl\{ P_0\,\ddot{\ell}_{P'}(s' - s,\, s' - s) \Bigr\}^{1/2}
   \quad \text{for all } s, s' \in T_{\mathcal{M}}(P').
$$
\item \textit{Consistency of $P_{n,0}$ for $P_0$: } $\norm{D_{n, P_{n,0}} - D_{n, P_0}}_{P_0} +  \norm{D_{n, 0, P_{n,0}} - D_{n, 0, P_0}}_{P_0}= o_p(1)$. \label{cond::weakconstencyOracle}
\item \textit{Oracle and working tangent spaces approximate \(\overline{T}_{\mathcal{M}_{n,0}}(P_0)\):} \label{cond::consistentTangentSpace}
For both \((P',\mathcal{M}')=(P_0,\mathcal{M}_0)\) and \((P',\mathcal{M}')=(P_{n,0},\mathcal{M}_n)\), as $n \rightarrow \infty$,
\[
\min_{s \in \overline{T}_{\mathcal{M}'}(P')}
P_0\ddot{\ell}_{P'}(s-s_{n,0,P'},\,s-s_{n,0,P'}) = o_p(1).
\]
 
\end{enumerate}
 
\end{enumerate}

\begin{lemma}[Sufficient conditions for B2]
Under Condition~\ref{cond::suffcondAL}, both \(\|D_{n,0,P_{n,0}} - D_{n,P_{n,0}}\|_{P_0} = o_p(1)\) and \(\|D_{n,0,P_{n,0}} - D_{0,P_{n,0}}\|_{P_0} = o_p(1)\). Furthermore, Condition~\ref{cond::suffcondAL} implies Condition~\ref{cond::consDn}.
\label{lemma::stableLogLik}
\end{lemma}

  The proof is given in Appendix \ref{appendix::proofstable}. The key requirement is Condition~\ref{cond::consistentTangentSpace}, which ensures that the relevant tangent spaces are rich enough to approximate the score directions \(s_{n,0,P_0}\) and \(s_{n,0,P_{n,0}}\) with vanishing error. Specifically, it requires that the oracle tangent space \(\overline{T}_{\mathcal{M}_0}(P_0)\) approximate \(s_{n,0,P_0} \in \overline{T}_{\mathcal{M}_{n,0}}(P_0)\), and that the working tangent space \(\overline{T}_{\mathcal{M}_n}(P_{n,0})\) approximate \(s_{n,0,P_{n,0}} \in \overline{T}_{\mathcal{M}_{n,0}}(P_{n,0})\), both with vanishing error. Condition~\ref{cond::lipschitzloss} is a mild regularity assumption requiring the loss derivative to be Lipschitz with respect to the inner product induced by the second derivative of the loss. Condition~\ref{cond::weakconstencyOracle} imposes a mild consistency requirement on the working-model projection \(P_{n,0}\) as an estimator of \(P_0\).

\section{Verifying the conditions of Theorem \ref{theorem::ALlinear} in our examples}
\label{appendix::conditionsexamples}

\setcounter{example}{1}

\begin{example}[continued]

Theory for this estimator is given in Appendix~\ref{app:ADMLpartiallylinear}, where Theorem~\ref{example::theorem::RlearnerLimitDistORACLE} specializes Theorem~\ref{theorem::ALlinear} under lower-level conditions on nuisance estimation rates and on the approximation of \(\mathcal{T}_0\) by \(\mathcal{T}_n\).
\end{example}

\begin{example}[continued]

Since \(\widehat{\psi}_n\) is a special case of a calibrated DML estimator, Theorem 3 of \cite{van2024automatic}, applied with \(\alpha_{n,j}^* = \overline{\alpha}_0 := 0\), gives lower-level sufficient conditions for Theorem~\ref{theorem::ALlinear}. Related sufficient conditions also appear in \cite{ATEsupereff}. Here we explain how these conditions imply those of Theorem~\ref{theorem::ALlinear}. Conditions~\ref{cond::bounded}--\ref{cond::linearnuisancerate} imply \(\|\widehat{D}_n - D_{0,P_0}\| = o_p(1)\). Hence, the empirical process condition \((P_n - P_0)(\widehat{D}_n - D_{0,P_0}) = o_p(n^{-1/2})\) in Condition~\ref{cond::linearsamplesplit} typically follows when the initial estimator \(\mu_n\) is sample split (or cross-fit), using the fact that isotonic functions form a Donsker class (Theorem 1 of \citealp{rabenseifner2025calibration}; Theorem 3 of \cite{van2024automatic}). Condition~\ref{cond::linearnuisancerate} holds automatically since \(\alpha_n = \alpha_{0,\mathcal{H}_n}\), and therefore \(\|\alpha_n - \alpha_{0,\mathcal{H}_n}\|_{P_0} = 0\). To verify Condition~\ref{cond::oraclebiaslinear}, we apply Lemma~\ref{lemma::lipschitzdependent2} in the Appendix (see also the proof of Lemma~\ref{lemma::lipschitzdependent}), which shows that the approximation errors \(\|\operatorname{Proj}_n(\alpha_{0,\mathcal{H}_{n,0}})-\alpha_{0,\mathcal{H}_{n,0}}\|_{P_0}\), \(\|\operatorname{Proj}_0(\alpha_{0,\mathcal{H}_{n,0}})-\alpha_{0,\mathcal{H}_{n,0}}\|_{P_0}\), and \(\|\operatorname{Proj}_n(\mu_0)-\mu_0\|_{P_0}\) are each bounded above by \(\|\mu_n^*-\mu_0\|_{P_0}\), provided that \((t_1,t_2)\mapsto E_0[Y \mid \mu_n^*(A,W)=t_1,\mu_0(A,W)=t_2,\mathcal{D}_n]\) and \((t_1,t_2)\mapsto E_0[\alpha_{0,\mathrm{np}}(A,W)\mid \mu_n^*(A,W)=t_1,\mu_0(A,W)=t_2,\mathcal{D}_n]\) are almost surely \(L\)-Lipschitz, where \(\alpha_{0,\mathrm{np}}\) denotes the nonparametric Riesz representer. It follows from this result and Lemma~\ref{lemma::lipschitzdependent} that Condition~\ref{cond::oraclebiaslinear} holds whenever \(\|\mu_n^*-\mu_0\|_{P_0}=o_p(n^{-1/4})\). 
\end{example}

\begin{example}[continued]
Theorem~\ref{example::theorem::RlearnerLimitDistORACLE} in Appendix~\ref{section::DataAdaptexampleATEPartially} provides a specialization of Theorem~\ref{theorem::ALlinear} for ADML estimators in the partially linear regression model that is easier to apply. Here, we discuss how to verify the conditions of of Theorem~\ref{theorem::ALlinear} directly. As in the previous example, the second part of Condition~\ref{cond::linearsamplesplit} is typically satisfied when the initial CATE estimator \(\tau_n\) is constructed using sample splitting and calibration is performed by isotonic regression \citep{van2024automatic,klaassen2025calibration}. Under Condition~\ref{cond::bounded}, Condition~\ref{cond::linearnuisancerate} holds if the R-learner nuisance estimators satisfy \(\|m_n-m_0\|\|\pi_n-\pi_0\| = o_p(n^{-1/2})\) and \(\|\pi_n-\pi_0\| = o_p(n^{-1/4})\). Under unconfoundedness, Lemma~\ref{lemma::lipschitzdependent2} in the Appendix (see also Lemma~\ref{lemma::lipschitzdependent}) implies that Condition~\ref{cond::oraclebiaslinear} holds if \(\|\tau_n^*-\tau_0\|_{P_0} = o_p(n^{-1/4})\) and the mappings \((t_1,t_2) \mapsto E_0[Y(1)-Y(0)\mid \tau_n^*(W)=t_1,\tau_0(W)=t_2,\mathcal{D}_n]\) and \((t_1,t_2) \mapsto E_0[\gamma_{0,\mathcal{T}}(A,W)\mid \tau_n^*(W)=t_1,\tau_0(W)=t_2,\mathcal{D}_n]\) are almost surely \(L\)-Lipschitz, where \(\mathcal{T}\) is any linear space containing \(\mathcal{T}_n\) and \(\mathcal{T}_0\). When the CATE is constant, the second part of Condition~\ref{cond::oraclebiaslinear} holds automatically, since \(\mu_0 \in \mathcal{H}_n\) and \(\|\operatorname{Proj}_n(\mu_0)-\mu_0\|_{P_0}=0\).
\end{example}

 \section{Application of theory to partially linear model selection}
\label{appendix::PLM}

\subsection{Setup and estimator}

\label{app:ADMLpartiallylinear}
In the following subsections, we apply our theory to the partially linear ADML estimator proposed in Example~\ref{example::CATE}, which we recall as follows:

 \begin{example*}[ADML for partially linear working models]
    \label{example::partially::intro}
   Suppose $\mathcal{H}_n$ and $\mathcal{H}_0$ are partially linear regression models corresponding to the CATE models $ \mathcal{T}_n$ and $\mathcal{T}_0$, respectively \citep{robinson1988root}. The partially linear regression model enables direct modelling of the conditional average treatment effect. Given user-supplied estimators $m_n$ and $\pi_n$ of $m_0$ and $\pi_0$, a semiparametric ADML estimator for the ATE is given by $\widehat\psi_n := \frac{1}{n}\sum_{i=1}^n \tau_n(W_i)$, where
$$\tau_n := \argmin_{\tau\in {\mathcal{T}}_n} \sum_{i=1}^n \left[Y_i - m_n(W_i) - \left\{A_i - \pi_n(W_i)\right\}\tau(W_i)\right]^2.$$
This partially linear ADML estimator encompasses various data-adaptive CATE estimators, including the post-Lasso R-learner \citep{belloni2014inference, zhao2017selective, nie2021quasi}.
\end{example*}

\subsection{Oracle parameter and EIF}
\label{section::oracleparam::example}

Recall from Section~\ref{section::EIFlinear} that the oracle statistical model \(\mathcal{M}_0\) is defined by \(P \in \mathcal{M}_0\) if and only if \(\mu_P\) belongs to the partially linear regression model
\[
\mathcal{H}_0 := \left\{(a,w) \mapsto \mu(w,0) + a\tau(w) : \mu \in L^2(P_{0,A,W}),\ \tau \in \mathcal{T}_0 \right\},
\]
where \(\mathcal{T}_0\) is an unknown but learnable linear CATE model for \(\tau_0\), and \(P_{0,A,W}\) denotes the distribution of \((A,W)\) under \(P_0\). Using Robinson's transformation \citep{robinson1988root}, the outcome regression can be written \(P_0\)-almost surely as
\[
\mu_0:(a,w)\mapsto m_0(w) + \{a - \pi_0(w)\}\tau_0(w),
\]
where \(m_0(w):=E_0(Y \mid W=w)\). The oracle parameter can therefore be expressed as \(P \mapsto \Psi_0(P) := E_P\{\Pi_0\tau_P(W)\}\), where
\[
\Pi_0\tau_P := \argmin_{\tau \in \mathcal{T}_0} E_P \left[Y - m_P(W) - \{A - \pi_P(W)\}\tau(W)\right]^2.
\]
Moreover, \(\Pi_0\tau_P\) can be shown to equal the overlap-weighted projection of the CATE \citep{imbensOverlapEstimand2006,li2019overlapWeights,d2021overlap,morzywolek2023general}.

The oracle parameter \(\Psi_0\) corresponds to the composite least-squares loss defined pointwise by
\[
\ell(o, Q) := \{y-\mu_Q(a,w)\}^2 - \log\left\{\frac{dQ_W}{d\mu}(w)\right\},
\]
where \(Q_W\) denotes the distribution of \(W\) under \(Q\). The negative log-likelihood term ensures that the induced projection \(\Pi_0 P\) preserves the covariate distribution of \(P\). Although the minimizer of the risk \(Q \mapsto P\ell(\cdot,Q)\) over a submodel of \(\mathcal{M}_{\mathrm{np}}\) is typically not unique, the loss function \(\ell\) satisfies the conditions of Theorem~\ref{theorem::lossbasedEIF}. In particular, when \(\Psi_0\) is pathwise differentiable, its efficient influence function lies in the loss-based tangent space and is characterized by the following theorem. For its statement, we introduce the following condition:
\begin{enumerate}[label=(E\arabic*), ref=E\arabic*,series=condE]
 \item $\gamma_{P, \mathcal{T}_0}(W) := \argmin_{\gamma \in \mathcal{T}_0} E_P\left[\pi_P(W)\{1-\pi_P(W)\}\gamma(W)^2 -  2 \gamma(W) \right]$ exists. \label{cond::CATE::boundedFunc1}
\end{enumerate}

\begin{theorem}[Efficient influence function under partially linear model]
    Under Condition \ref{cond::CATE::boundedFunc1}, the oracle parameter $\Psi_0:\mathcal{M}_{\mathrm{np}}\rightarrow\mathbb{R}$ is pathwise differentiable at $P$ with efficient influence function
    $$D_{0,P}(o) = \Pi_0 \tau_P(w) - E_P\left\{\Pi_0\tau_P(W)\right\} +  \gamma_{P, \mathcal{T}_0}(w)\left\{a - \pi_P(w)\right\}\left\{y - \Pi_0 \mu_P(a,w) \right\},$$ 
    which is an element of $\mathcal{S}_{\mathcal{M}_0}(P_0) = L^2_0(P_{0,W}) \oplus \left\{o \mapsto h(a,w)\{y - \Pi_0 \mu_P(a,w)\} : h \in L^2(P_{0,A,W})\right\}$.\label{example::theorem::RlearnerEIF}
\end{theorem}
Condition \ref{cond::CATE::boundedFunc1} holds if and only if the linear functional $\mu \mapsto E_P\left\{\mu(1,W) - \mu(0,W)\right\}$ is bounded on $\mathcal{H}_0$. When $\pi_P(W)\{1-\pi_P(W)\} > 0$ almost surely and its reciprocal has finite variance, $\gamma_{P, \mathcal{T}_0}$ equals the overlap-weighted $L^2(P)$-projection of $\{\pi_P(1-\pi_P)\}^{-1}$ onto the linear working model $\mathcal{T}_0$. If $\mathcal{T}_0 := L^2(P_{0,W})$, then $\Psi_0 = \Psi$ and $\gamma_{P, \mathcal{T}_0} = \{\pi_P(1-\pi_P)\}^{-1}$ so that Theorem \ref{example::theorem::RlearnerEIF} recovers the nonparametric efficient influence function of the ATE.

\subsection{Large-sample theory}
\label{section::DataAdaptexampleATEPartially}

The corresponding data-adaptive working parameter \(\Psi_n:\mathcal{M}\to\mathbb{R}\) is defined pointwise by
\[
\Psi_n(P) := E_P\{\Pi_n\tau_P(W)\},
\]
where
\[
\Pi_n\tau_P := \argmin_{\tau \in \mathcal{T}_n} E_P \left[Y - m_P(W) - \{A - \pi_P(W)\}\tau(W)\right]^2.
\]
Hence, the partially linear ADML estimator \(\widehat{\psi}_n\) is simply a plug-in estimator of \(\Psi_n(P_0)\). The first-order equations characterizing the empirical risk minimizer \(\tau_n\) imply that \(\frac{1}{n}\sum_{i=1}^{n} D_{n,\widehat{P}_n}(O_i) = 0\), so that \(\widehat{\psi}_n\) is in fact an ATMLE and satisfies Condition~\ref{cond::debiased} under mild conditions.

We now state our main result. Let \(\mathcal{T}_{n,0} = \mathcal{T}_n \oplus \mathcal{T}_0\) denote the orthogonal sum of \(\mathcal{T}_n\) and \(\mathcal{T}_0\), and let \( \gamma_{0,\mathcal{T}_{n,0}}\) be the overlap-weighted Riesz representer under \(P_0\) for the model \(\mathcal{T}_{n,0}\). Let \(\Pi_{0}\) denote the overlap-weighted projection onto \(\mathcal{T}_0\). We denote the overlap-weighted \(L^2(P_0)\)-norm of a function \(f \in L^2(P_{0,W})\) by \(\|f\|_{w_0P_0} := \|w_0^{1/2}f\|_{P_0}\), where \(w_0 := \pi_0(1-\pi_0)\), and introduce the following conditions:

\begin{enumerate}[label=E\arabic*), ref={E\arabic*}, resume=condE]
    \item \textit{Continuity of linear functional}: the representer $\gamma_{P, \mathcal{T}_{n,0}}(W)$ exists with finite $P_0$-variance. \label{cond::CATE::boundedFunc}
    
    \item \textit{Donsker condition:} \(\tau_n\), \(\pi_n\), \(m_n\), and \(\Pi_n \gamma_{0,\mathcal{T}_0}\) are uniformly bounded and fall in a fixed \(P_0\)-Donsker class with probability tending to one;
     \label{cond::CATE::DonskerMLE}
     \item \textit{Consistency of nuisance estimators:} \(\norm{\pi_n - \pi_0}_{P_0} + \norm{\tau_n - \Pi_n \tau_0}_{P_0} + \|m_n - m_0\|_{P_0} = o_p(1)\);
     \label{cond::CATE::consistentNuis}
     \item \textit{Sufficient nuisance rates:} \(\norm{\pi_n - \pi_0}_{P_0} = o_p(n^{-1/4})\) and \(\norm{\pi_n - \pi_0}_{P_0}\norm{m_n - m_0}_{P_0} = o_p(n^{-1/2})\).
     \label{cond::CATE::DRterm}
     \item \textit{Consistency of working model:} \(\norm{\gamma_{0,\mathcal{T}_0} -  \gamma_{0,\mathcal{T}_n}}_{w_0P_0} + \norm{\Pi_n \tau_0 - \tau_0}_{P_0} = o_p(1)\);
     \label{cond::CATE::consistentWorking}
     \item \textit{Negligible model approximation error term:} \(\norm{\gamma_{0,\mathcal{T}_{n,0}} - \Pi_n \gamma_{0,\mathcal{T}_{n,0}}}_{w_0P_0}\norm{\tau_0 - \Pi_n \tau_0}_{w_0P_0} = o_p(n^{-1/2})\).
     \label{cond::CATE::critBias}
\end{enumerate}

Conditions~\ref{cond::CATE::boundedFunc} ensure that \(\Psi_n\) and \(\Psi_0\) are pathwise differentiable. Condition~\ref{cond::CATE::DonskerMLE} restricts the complexity of the nuisance estimators \(\pi_n\) and \(m_n\), but it can be relaxed to accommodate generic machine learning methods via cross-fitting \citep{vanderLaanRose2011, DoubleML}. The requirement that \(\tau_n\) lie in a Donsker class is satisfied by various estimators, including the highly adaptive Lasso and Lasso-regularized regression over reproducing kernel Hilbert spaces. However, without strong sparsity conditions, this requirement may fail in high-dimensional settings \citep{DoubleML,chernoapproxSparse2019}.
Conditions~\ref{cond::CATE::consistentNuis} and \ref{cond::CATE::consistentWorking} typically require that \(\mathcal{T}_n\) be finite-dimensional and impose mild consistency conditions on the nuisance estimators and projections. Condition~\ref{cond::CATE::DRterm} is a standard nuisance-rate condition for partially linear regression and is trivially satisfied when the propensity score \(\pi_0\) is known and \(\pi_n = \pi_0\). Finally, Condition~\ref{cond::CATE::critBias} ensures that the model approximation error is negligible, so that Condition~\ref{cond::oraclebiaslinear} holds.

\begin{theorem}[Superefficiency of partially linear ADML]
\label{example::theorem::RlearnerLimitDistORACLE}
\label{example::theorem::RlearnerLimitDist}
Under Conditions~\ref{cond::CATE::boundedFunc}--\ref{cond::CATE::critBias}, the partially linear ADML estimator \(\widehat{\psi}_n\)
\begin{enumerate}
    \item is a \(P_0\)-asymptotically linear estimator of \(\Psi_n(P_0)\) with influence function equal to the \(P_0\)-efficient influence function \(D_{0,P_0}\) of \(\Psi_0\) relative to \(\mathcal{M}_{\mathrm{np}}\);
    \item is \(P_0\)-asymptotically linear, regular, and efficient for the oracle parameter \(\Psi_0:\mathcal{M}_{\mathrm{np}}\to\mathbb{R}\), with
    \[
    \widehat{\psi}_n - \Psi_0(P_0) = (P_n - P_0)D_{0,P_0} + o_p(n^{-1/2});
    \]
    \item is, in addition, \(P_0\)-efficient for \(\Psi:\mathcal{M}_0\to\mathbb{R}\) relative to the oracle submodel \(\mathcal{M}_0\) if the conditional variance of \(Y\) given \((A,W)\) is almost surely constant.
\end{enumerate}
\end{theorem}

Theorem~\ref{example::theorem::RlearnerLimitDistORACLE} implies that $\sqrt{n}\,(\widehat{\psi}_n-\psi_0)/\sigma_0 \rightarrow_d N(0,1),$
where \(\sigma_0^2 = \mathrm{var}_{0}\{D_{0,P_0}(O)\}\) is the efficiency bound. Thus, \(\widehat{\psi}_n\) is \(P_0\)-superefficient for the ATE parameter \(\Psi\), with limiting variance adaptive to the complexity of the CATE \(\tau_0\).

 The following corollary shows that the ADML estimator is regular along local perturbations of \(P_0\) whose corresponding CATE lies in the oracle submodel \(\mathcal{T}_0\). If the learned oracle submodel \(\mathcal{M}_0\) is only approximately correct at a given sample size \(n\), then the ADML estimator may exhibit asymptotic bias under perturbations outside \(\mathcal{M}_0\). Nevertheless, even under least-favorable local perturbations outside the oracle submodel, the ADML estimator continues to yield valid inference for the corresponding oracle projection-based ATE estimand.

\begin{corollary}[Limiting behavior under local perturbations]
     \label{cor::CATEinf}  The ADML estimator is \(P_0\)-regular for the ATE parameter \(\Psi\) with respect to local perturbations of \(P_0\) in the oracle submodel \(\mathcal{M}_0\). Moreover, under sampling from any local perturbation \(P_{0,hn^{-1/2}} \in \mathcal{M}_{\mathrm{np}}\) outside \(\mathcal{M}_0\), it holds that
   \[
   \sqrt{n}\,\{\widehat{\psi}_n-\Psi_0(P_{0,hn^{-1/2}})\}/\sigma_0 \rightarrow_d N(0,1).
   \]
 
\end{corollary}

To highlight the advantages of ADML estimators, consider the semiparametric ATE estimator based on the partially linear intercept model \citep{robinson1988root,imbensOverlapEstimand2006, li2019overlapWeights, d2021overlap}, corresponding to \(\mathcal{T}:=\{w \mapsto c : c \in \mathbb{R}\}\). This estimator is irregular and asymptotically biased under local perturbations outside the semiparametric model. By contrast, if \(\mathcal{T}_0\) contains the intercept model, then the partially linear ADML estimator remains regular and asymptotically unbiased over a broader class of local perturbations. In addition, by Corollary~\ref{cor::CATEinf} and Theorem~\ref{theorem::irreg}, the ADML estimator is typically less biased than the semiparametric estimator under local perturbations outside the oracle submodel. If \(\mathcal{T}_0\) is itself the intercept model, then the ADML estimator and the semiparametric estimator are asymptotically equivalent under sampling from \(P_0\) and its local perturbations in \(\mathcal{M}_{\mathrm{np}}\).

\section{Supplementary experimental results}

\label{appendix::figures}
\subsection*{Sampling from true distribution}

 \begin{figure}[H]
     \centering
     \begin{subfigure}[b]{0.48\linewidth} 
         \includegraphics[width=0.5\linewidth]{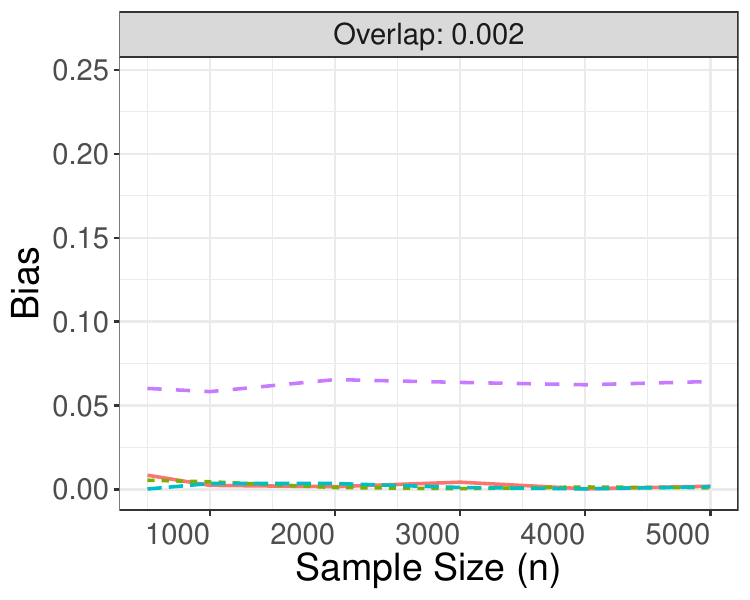}\includegraphics[width=0.5\linewidth]{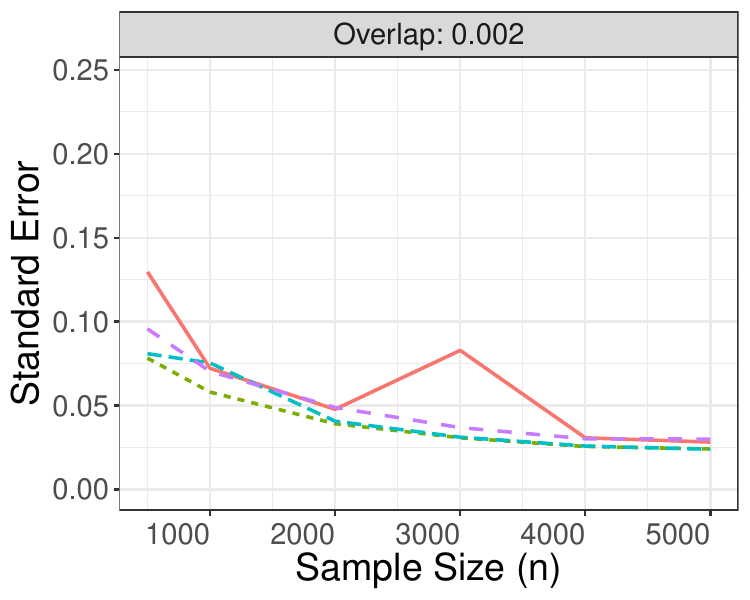}
     \includegraphics[width=0.5\linewidth]{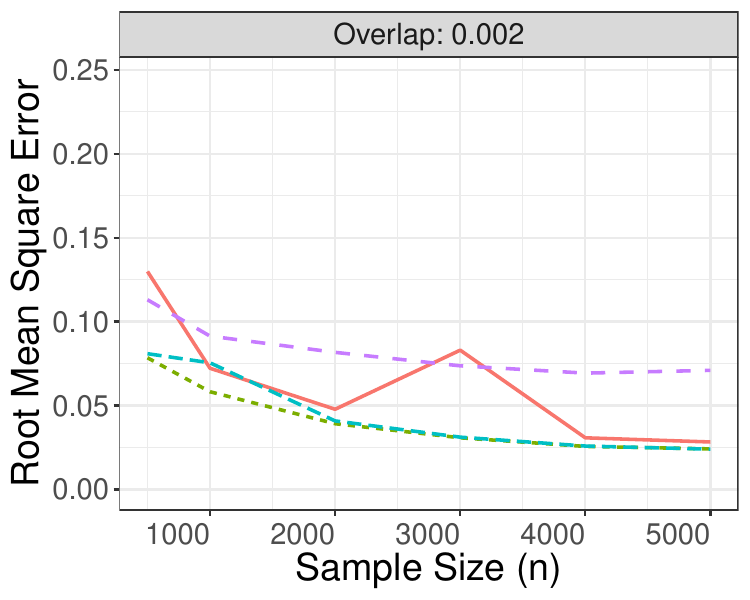}\includegraphics[width=0.5\linewidth]{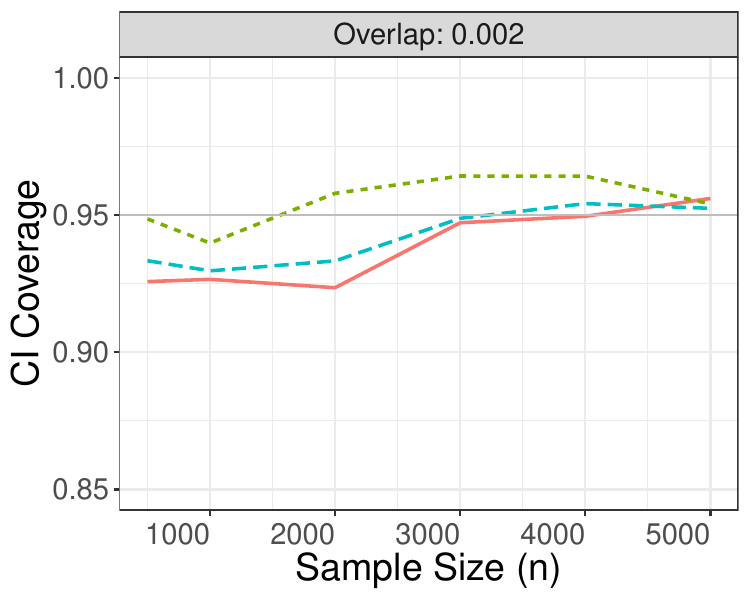}
     \subcaption{Linear setting}
     \end{subfigure} \begin{subfigure}[b]{0.48\linewidth} 
         \includegraphics[width=0.5\linewidth]{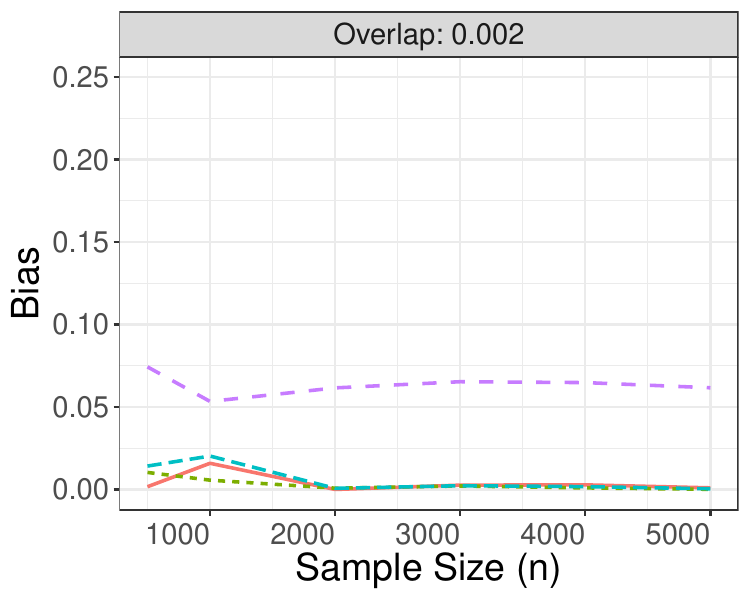}\includegraphics[width=0.5\linewidth]{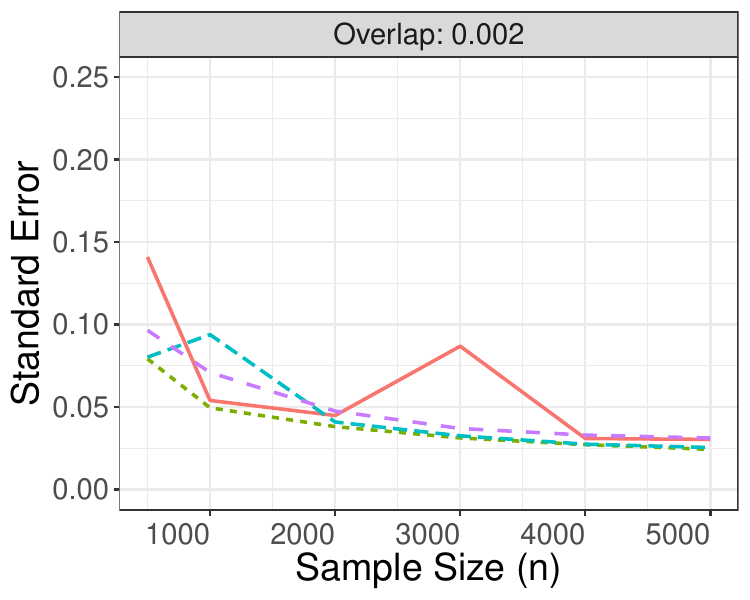}
     \includegraphics[width=0.5\linewidth]{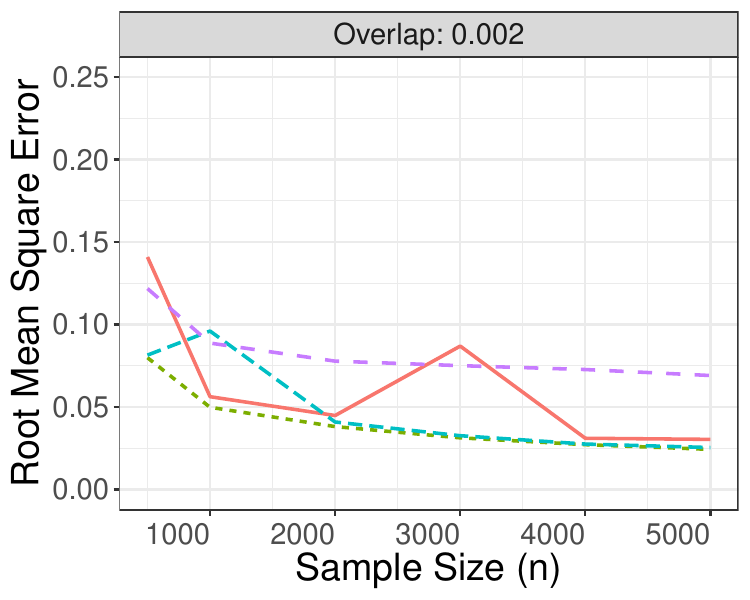}\includegraphics[width=0.5\linewidth]{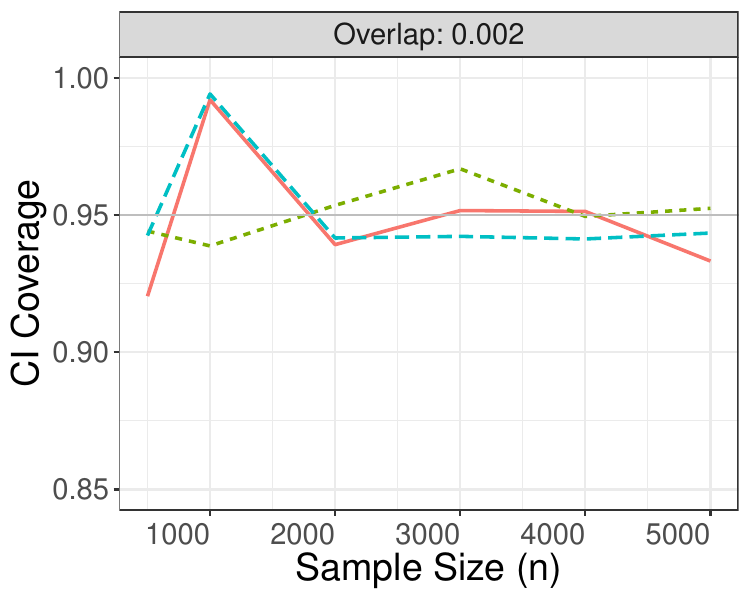}
     \subcaption{Nonlinear setting}
     \end{subfigure}
            \includegraphics[width=0.7\linewidth]{plots/Legend.png}

     \caption{Comparison of empirical bias, standard error and root mean squared error of estimator, and coverage of nominal 95\% confidence interval for partially linear and plug-in HAL-ADML estimators, prespecified semiparametric estimator (assuming constant CATE), and nonparametric AIPW estimator, under sampling from a fixed distribution satisfying linearity and with varying degrees of treatment overlap. Coverage probabilities for intervals based on the prespecified semiparametric estimator were consistently poor, exceeding the y-axis range.}
     \label{fig:my_label}
 \end{figure}

\subsection*{Sampling from least-favorable local alternative}

  \begin{figure}[H]
     \centering
   \begin{subfigure}[b]{0.48\linewidth} 
     \includegraphics[width=0.5\linewidth]{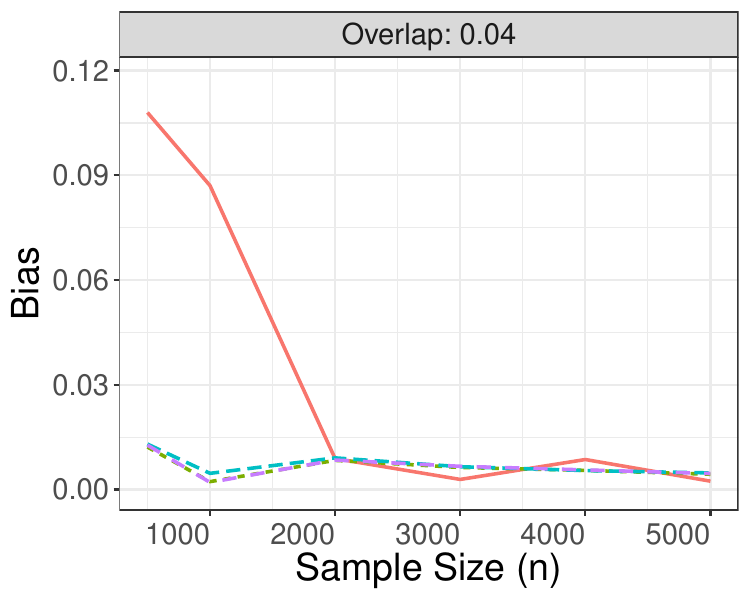}\includegraphics[width=0.5\linewidth]{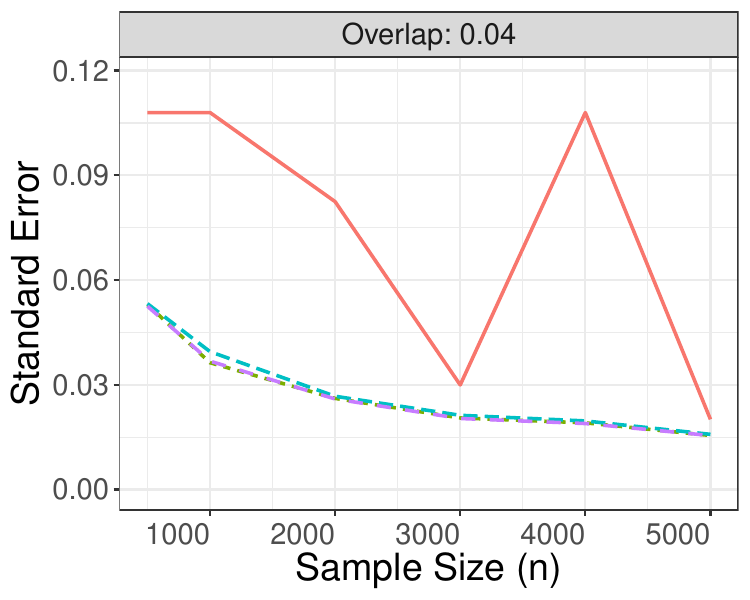}
     \includegraphics[width=0.5\linewidth]{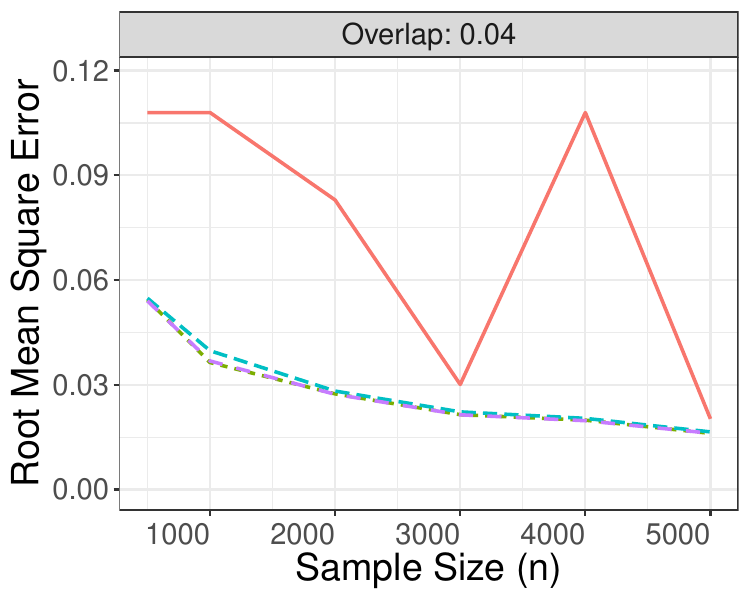}\includegraphics[width=0.5\linewidth]{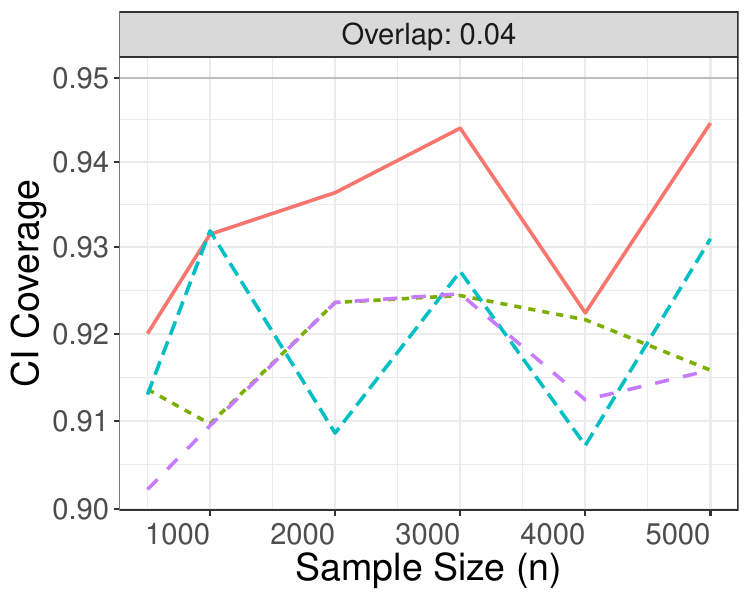}
     \subcaption{Nonlinearity and moderate overlap  ($c_0 \approx 0.04$)}
     \end{subfigure} \hfill \begin{subfigure}[b]{0.48\linewidth} 
     \includegraphics[width=0.5\linewidth]{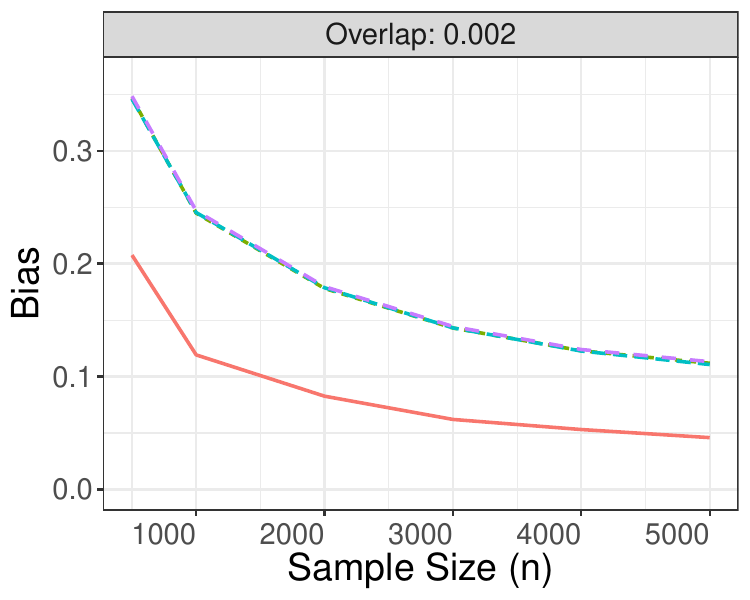}\includegraphics[width=0.5\linewidth]{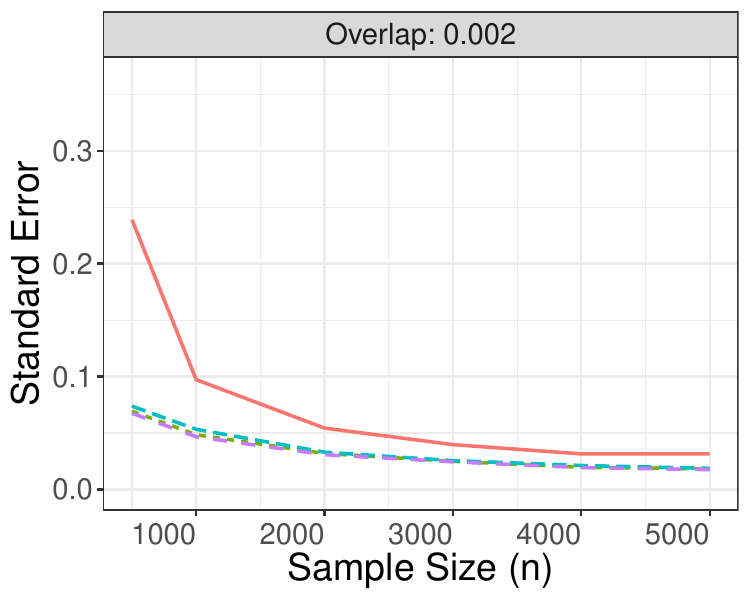}
     \includegraphics[width=0.5\linewidth]{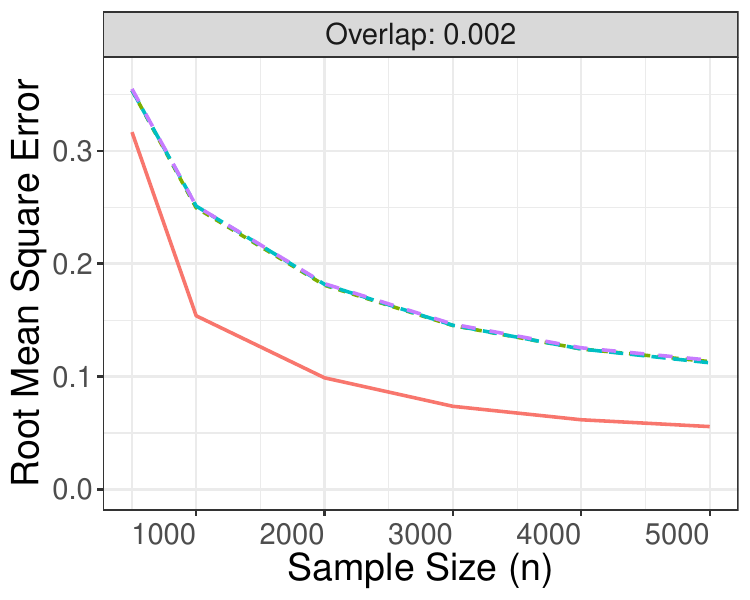}\includegraphics[width=0.5\linewidth]{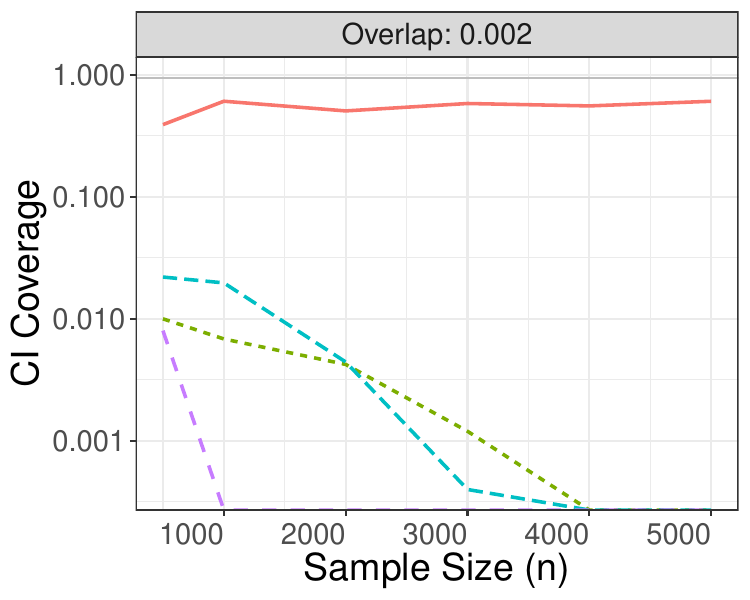}
      \subcaption{Nonlinearity and limited overlap  ($c_0 \approx 0.002$)}
     \end{subfigure}  

     \begin{subfigure}[b]{0.48\linewidth} 
         \includegraphics[width=0.5\linewidth]{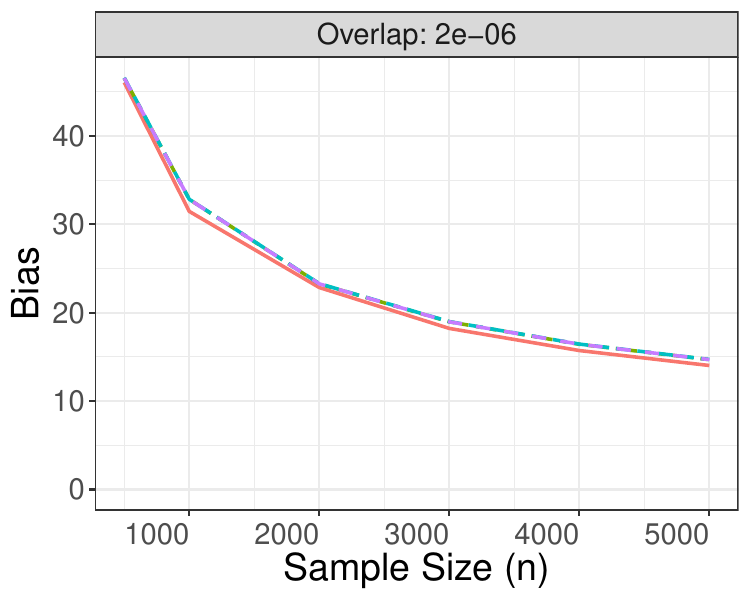}\includegraphics[width=0.5\linewidth]{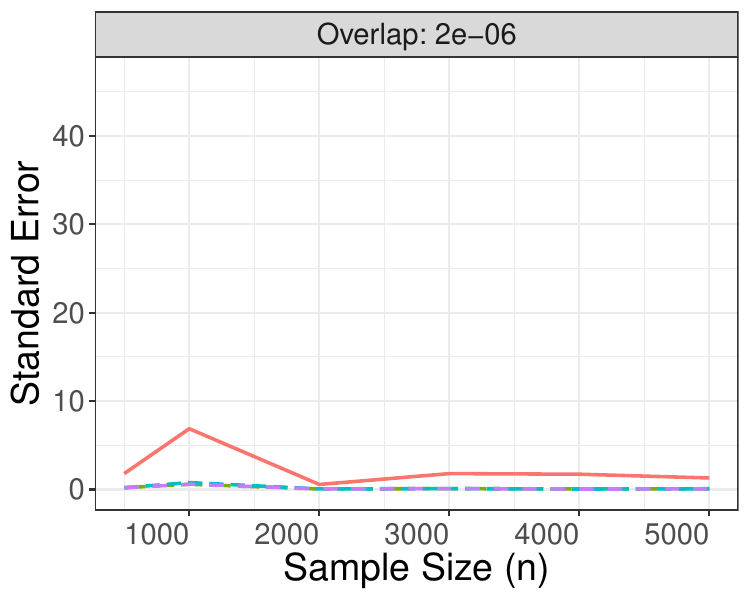}
     \includegraphics[width=0.5\linewidth]{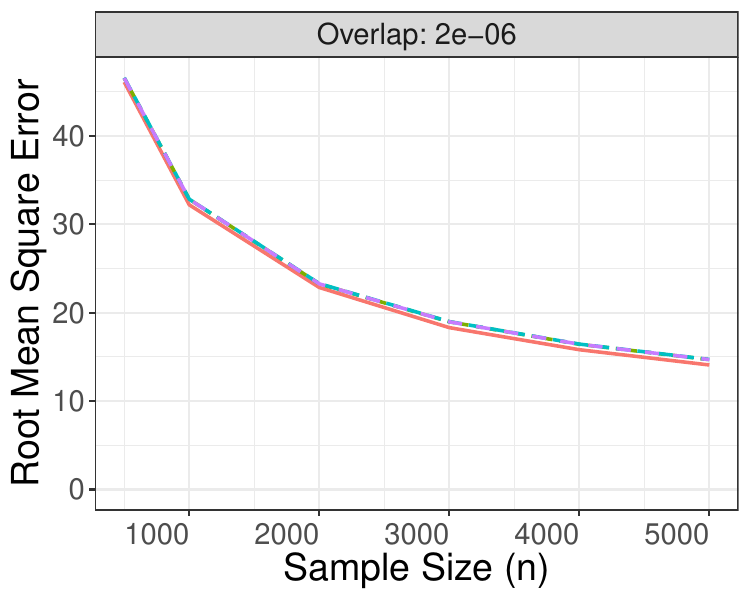}\includegraphics[width=0.5\linewidth]{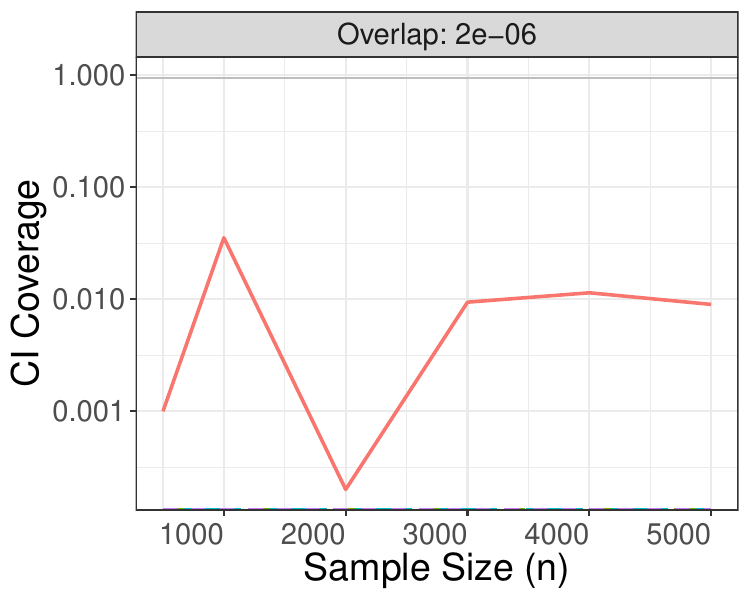}
     \subcaption{Linearity and no overlap ($c_0 \approx 10^{-6}$)}
     \end{subfigure} \begin{subfigure}[b]{0.48\linewidth} 
         \includegraphics[width=0.5\linewidth]{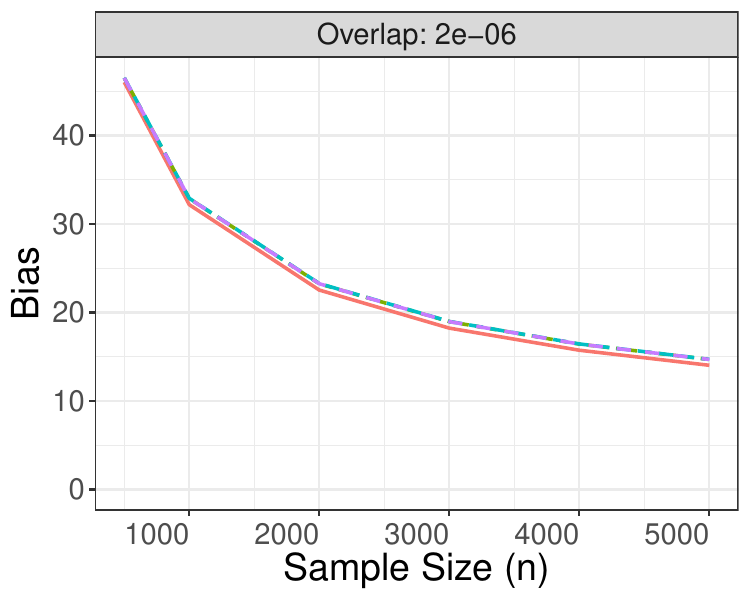}\includegraphics[width=0.5\linewidth]{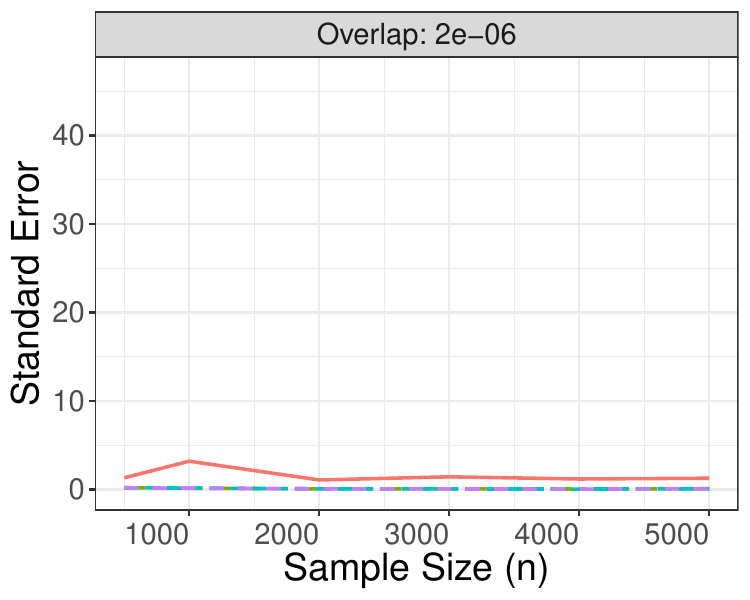}
     \includegraphics[width=0.5\linewidth]{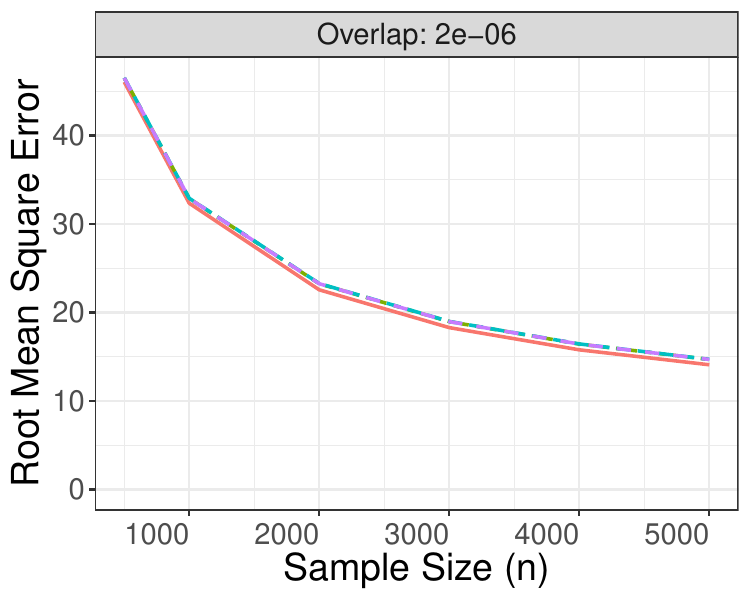}\includegraphics[width=0.5\linewidth]{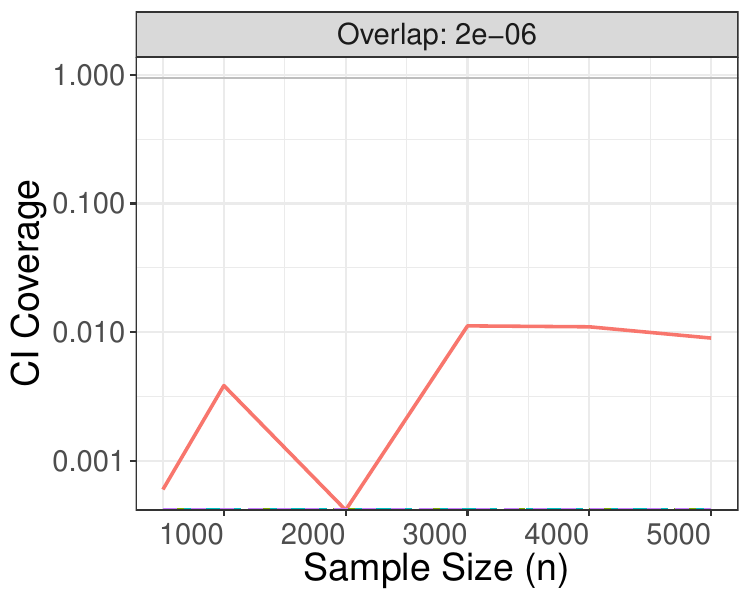}
     \subcaption{Nonlinearity and no overlap ($c_0 \approx 10^{-6}$)}
     \end{subfigure}
            \includegraphics[width=0.8\linewidth]{plots/Legend.png}

     \caption{Comparison of empirical bias, standard error and root mean squared error of estimator, and coverage of nominal 95\% confidence interval for partially linear and plug-in HAL-ADML estimators, prespecified semiparametric estimator (assuming constant CATE), and nonparametric AIPW estimator, under sampling from a least-favorable local perturbation of a distribution satisfying linearity in outcome regression and with varying degrees of treatment overlap. }
     \label{fig:my_label}
 \end{figure}

\section{Proofs for Section \ref{section::oracleparameter}}

 \subsection{Proof of Theorem \ref{theorem::EIFmain}}

\begin{proof}[Proof of Theorem \ref{theorem::EIFmain}]

Let $P \in \mathcal{M}_{\mathrm{np}}$ be arbitrary.  By  \ref{cond::losssmooth}, a minimizing solution $Q_P \in  \argmin_{Q \in \mathcal{M}_0} P \ell (\cdot, Q)$ satisfies the first-order optimality conditions:
$$ \frac{d}{dt} P  \ell(\cdot, Q_t)   \big|_{t=0} = 0,$$
for all regular paths $(Q_t: t \in (-\varepsilon, \varepsilon)) \subseteq \mathcal{M}_0$ with $Q_t = Q_P$ at $t=0$.

For some $\delta > 0$, let $(P_t : t \in (-\delta, \delta)) \subseteq \mathcal{M}_{\mathrm{np}}$ be a regular path through $P$ such that $dP_t = (1+ ts) dP$ for a bounded score $s \in L^2_0(P)$ orthogonal to the loss-based tangent space $\mathcal{S}_{\mathcal{M}_0}(P)$. Since ${\cal M}_{\mathrm{np}}$ is a convex nonparametric model and the score $s$ is bounded, such a path necessarily exists for sufficiently small $\delta > 0$. By \ref{cond::smoothprojection}, there exists a smooth submodel $t \mapsto Q_{P_t} \in \Pi_0 P_t$ with score $s_{0}$ at $t = 0$. By the first order optimality condition of $Q_{P_t} \in \Pi_0 P_t$, we have, for all $t$,
\begin{align*}
 P_t  \dot{\ell}(\cdot, Q_{P_t})(v)   \big|_{t=0} =  0.
\end{align*}
Taking the derivative of both sides and applying the chain rule, we find, for all $v \in T_{\mathcal{M}_0}(Q_P)$, that
\begin{align}
 P \{s \dot{\ell}(\cdot, Q_{P})(v)\}   &=  - P \ddot{\ell}(\cdot, Q_{P})(v, s_0) \nonumber\\ 
  &=  - P \ddot{\ell}_{\Pi_0 P}(v, s_0), \label{eqn::pathwiseofloss}
\end{align}
where the final equality uses invariance of $Q \mapsto P \dot{\ell}(\cdot, Q) $ over the solution set such that $P \dot{\ell}(\cdot, \Pi_0 P)(s_{0}, s_{0,P})$ is well-defined and invariant of the choice of $Q_P \in \Pi_0 P$.

By assumption \ref{cond::oraclepathwise} , $\Psi$ is pathwise differentiable on $\mathcal{M}_0$ with respect to the Hessian inner-product $(v,s) \mapsto P\ddot{\ell}_{\Pi_0 P}(v, s) $. Thus, by definition, for any smooth submodel $(Q_t : t) \subset \mathcal{M}_0$ satisfying $Q_t = \Pi_0 P$ at $t = 0$ with score function $s$, the derivative $\frac{d}{dt} \Psi(Q_t) \Big|_{t=0}$ is a bounded linear operator in the score $s$. Note that the inner product $P\ddot{\ell}_{\Pi_0 P}$ is positive definite with respect to the Hilbert space $\{s + N \mid s \in T_{\mathcal{M}_0}(P)\}$ induced by the equivalence relation $s_1 \sim s_2$ if and only if $s_1 \in s_2 + N$, where $N = \{h \in T_{\mathcal{M}_0}(P) \mid P\ddot{\ell}_{\Pi_0 P}(h, h) = 0\}$ denotes the null space of the inner product. Hence, by Riesz representation theorem, there exists a (potentially non-unique) Riesz representer $s_{0,P} \in T_{\mathcal{M}_0}(\Pi_0 P)$ such that the pathwise derivative satisfies:
\[
\frac{d}{dt} \Psi(Q_t) \Big|_{t=0} = P\ddot{\ell}_{\Pi_0 P}(s, s_{0,P}).
\]

Now, taking the pathwise derivative along the path $t \mapsto Q_{P_t}$, it follows that
\begin{align*}
    \frac{d}{dt} \Psi(\Pi_0 P_t) \big|_{t = 0} &=   \frac{d}{dt} \Psi(Q_t) \big|_{t = 0} \\
    &= P \ddot{\ell}(\cdot, Q_P)(s_{0}, s_{0,P})\\
    &= P \ddot{\ell}_{\Pi_0 P}(s_{0}, s_{0,P}),
\end{align*}
where the final equality uses invariance of $Q \mapsto P \ddot{\ell}_{Q} $ over the solution set. Thus, by \eqref{eqn::pathwiseofloss}, we have that
$$\frac{d}{dt} \Psi(\Pi_0 P_t) \big|_{t = 0} =  -  P \{s \dot{\ell}(\cdot,  \Pi_0 P)(s_{0,P}))\}.$$
The first-order optimality conditions of the risk minimizer $Q_P \in \Pi_0 P$ imply that $P\dot{\ell}(\cdot, \Pi_0 P)(s_{0,P})) = P\dot{\ell}(\cdot, Q_P)(s_{0,P}))$ is mean-zero and is, thus, an element of $L^2_0(P)$. The score $s$ was an arbitrary bounded score in $T_{\mathcal{M}_{\mathrm{np}}}(P) = L^2_0(P)$. Since $T_{\mathcal{M}_{\mathrm{np}}}(P)$ is a closed linear space and bounded scores are dense in $L^2_0(P)$, we have 
$$\frac{d}{dt} \Psi(\Pi_0 P_t) \big|_{t = 0} =  -  P \{s \dot{\ell}(\cdot,  \Pi_0 P)(s_{0,P}))\},\, \forall s \in L^2_0(P).$$
Consequently, we conclude that $\dot{\ell}(\cdot, \Pi_0 P)(s_{0,P})$ is a gradient of $\Psi_0$ with respect to the nonparametric statistical model and that $\Psi_0$ is pathwise differentiable at $P$. Since $\dot{\ell}(\cdot, \Pi_0 P)(s_{0,P})$ is an element of the tangent space, it follows that it is the canonical gradient (or efficient influence function) of $\Psi_0$ at $P$. Although the Riesz representer $s_{0,P}$ may be nonunique, the EIF $\dot{\ell}(\cdot, \Pi_0 P)(s_{0,P})$ is necessarily unique by the Riesz representation theorem applied to the pathwise derivative operator $s \mapsto \frac{d}{dt} \Psi(\Pi_0 P_t) \big|_{t = 0}$ and the $L^2_0(P)$ inner product. Hence, the EIF $\dot{\ell}(\cdot, \Pi_0 P)(s_{0,P})$ is invariant with respect to the choice of $s_{0,P}$.

\end{proof}

\subsection{A characterization of the EIF as a loss derivative}

Theorem \ref{theorem::EIFmain} shows that, under suitable conditions, \(\Psi_0\) is pathwise differentiable and that its efficient influence function is given by a score of the loss \(\ell\) in the direction of the Hessian gradient \(s_{0,P}\). The next theorem provides a partial converse: assuming \(\Psi_0\) is pathwise differentiable, its efficient influence function must be representable as a score of the loss \(\ell\).

 \begin{enumerate}[label=(D\arabic*), ref=D\arabic*,series=cond]
\item \textit{Invariance of $\Psi$ over solution set:} For all $P \in \mathcal{M}_{\mathrm{np}}$, $\argmin_{Q \in \mathcal{M}_0} P\ell(\cdot, Q)$ is nonempty and $\Psi(Q) = \Psi(Q')$ for all $Q, Q' \in \argmin_{Q \in \mathcal{M}_0} P\ell(\cdot, Q) $.   \label{cond::lossProjIdent}
\item  \textit{Pathwise differentiability of $\Psi_0$ at $P_0$:} The oracle projection parameter $\Psi_0:\mathcal{M}_{\mathrm{np}}\rightarrow\mathbb{R}$  is pathwise differentiable at $P_0$ with efficient influence function $o \mapsto D_0(o; P_0)$.
\label{cond::oracleParamPathwise}
\item \textit{Risk minimizer determined by score equations:} For each $P \in \mathcal{M}_{\mathrm{np}}$, $Q_P \in \argmin_{Q \in \mathcal{M}_0} P \ell(\cdot, Q)$ if and only if $\frac{d}{dt} P \ell(\cdot,Q_t)\big |_{t=0} = 0$ for each regular path $\{Q_t: t \in \mathbb{R}\} \subseteq \mathcal{M}_0$ with $Q_t = Q_P$ at $t=0$.   \label{cond::lossProjDeriv}
\end{enumerate}

In the following theorem, we define the loss-based tangent space $\mathcal{S}_{\mathcal{M}_0}(P)\subseteq L^2_0(P_0)$ of the oracle submodel $\mathcal{M}_0$ at any $P\in {\cal M}$ as the closure of the linear span of $P$--weak G\^{a}teaux derivatives (i.e., $\ell$-scores) of the form $\frac{d}{dt}  \ell(\cdot, Q_t)  \big|_{t = 0}$, where $\{Q_{t}: t \in \mathbb{R}\} \subseteq \mathcal{M}_0$ is a regular path with $Q_{t} = \Pi_0P$ at $t=0$.

 \begin{theorem}[Efficient influence function of oracle parameter]
 Under Conditions \ref{cond::lossProjIdent}-\ref{cond::lossProjDeriv},  the efficient influence function $D_{0,P_{0}}$ of the oracle projection parameter $\Psi_0 :{\cal M}_{\mathrm{np}}\rightarrow\mathbb{R}$ at $P_0$ is an element of $\mathcal{S}_{\mathcal{M}_0}(P_0)$.
As a consequence, if $\mathcal{S}_{{\cal M}_0}(P_0)$ is a subspace of the tangent space $T_{{\cal M}_0}(P_0)$ at $P_0$ for model ${\cal M}_0$, then $D_{0,P_0}$ equals the $P_0$--efficient influence function of $\Psi:{\cal M}_0\rightarrow\mathbb{R}$. 
\label{theorem::lossbasedEIF}
\end{theorem}

 The loss-based tangent space $\mathcal{S}_{\mathcal{M}_0}(P)$ consists of loss-based scores of paths through $\Pi_0P$ that remain in the oracle model, and so, it is a subspace of $L^2_0(P)$. For the loglikelihood loss $\ell(\cdot, Q)=-\log\left(\frac{dQ}{d\mu}\right)$, the loss-based tangent space $\mathcal{S}_{{\cal M}_0}(P_0)$ equals the tangent space $T_{P_0}({\cal M}_0)$ at $P_0$ for the model ${\cal M}_0$. Thus, as a consequence of Theorem \ref{theorem::lossbasedEIF}, the efficient influence function of the oracle parameter $\Psi_0$ for the loglikelihood loss at $P_0 \in \mathcal{M}_0$ is equal to the efficient influence function of the parameter $\Psi: \mathcal{M}_0 \rightarrow \mathbb{R}$ for the oracle model $\mathcal{M}_0$. In such cases, an efficient estimator for $\Psi_0$ at $P_0$ performs as well in a local asymptotic minimax sense as an efficient estimator that knew the oracle model $\mathcal{M}_0$ beforehand.

 \begin{proof}[Proof of Theorem \ref{theorem::lossbasedEIF}]
 
Let $P \in \mathcal{M}_{\mathrm{np}}$ be arbitrary. By \ref{cond::lossProjDeriv}, a minimizing solution $Q_P \in  \argmin_{Q \in \mathcal{M}_0} P \ell (\cdot, Q)$ satisfies
$$ \frac{d}{dt} P  \ell(\cdot, Q_t)   \big|_{t=0}  = P \left\{ \frac{d}{dt} \ell(\cdot, Q_t) \big|_{t=0} \right\} = 0,$$
for all regular paths $(Q_t: t \in (-\varepsilon, \varepsilon)) \subseteq \mathcal{M}_0$ with $Q_t = Q_P$ at $t=0$. 

For some $\delta > 0$, let $(P_t : t \in (-\delta, \delta)) \subseteq \mathcal{M}_{\mathrm{np}}$ be a regular path through $P$ such that $dP_t = (1+ ts) dP$ for a bounded score $s \in L^2_0(P)$ orthogonal to the loss-based tangent space $\mathcal{S}_{\mathcal{M}_0}(P)$. Since ${\cal M}_{\mathrm{np}}$ is a convex nonparametric model and the score $s$ is bounded, such a path necessarily exists for sufficiently small $\delta > 0$. By Condition \ref{cond::lossProjDeriv}, for all regular paths $(Q_u: u \in (-\varepsilon, \varepsilon)) \subseteq \mathcal{M}_0$ through $Q_P$ at $u = 0$, we have
\begin{align*}
P_t \left\{ \frac{d}{du} \ell(\cdot, Q_u) \big|_{u=0} \right\}& = \int \left\{ \frac{d}{du} \ell(o, Q_u) \big|_{u=0} \right\}\{ 1+ t s(o)\} P(do)\\
&= \int \left\{ \frac{d}{du} \ell(o, Q_u) \big|_{u=0} \right\} P(do)  + t \int \left\{ \frac{d}{du} \ell(o, Q_u) \big|_{u=0} \right\} s(o) P(do)\ .
\end{align*}
 The first term on the right-hand side is $0$ since $Q_P$ is a minimizer of $Q \mapsto P \ell(\cdot, Q)$ over $Q \in \mathcal{M}_0$. The second term on the right-hand side is also $0$ since $s$ is centered under $P$ and, by construction, orthogonal to $\frac{d}{du} \ell(\cdot, Q_u) \big|_{u=0}  \in \mathcal{S}_{\mathcal{M}_0}(P)$. It follows that $P_t \left\{ \frac{d}{du} \ell(\cdot, Q_u) \big|_{u=0} \right\} = 0$ for all such paths $(Q_u: u \in (-\varepsilon, \varepsilon)) \subseteq \mathcal{M}_0$ and $t$ sufficiently small. By Condition \ref{cond::lossProjDeriv}, this can only occur if $Q_P \in \argmin_{Q \in \mathcal{M}_0} P_t \ell (\cdot, Q)$ for all $t$ sufficiently small.

 By Condition \ref{cond::oracleParamPathwise}, $\Psi_0 = \Psi \circ \Pi_0$ is pathwise differentiable at $P_0$; thus, its efficient influence function $D_{0,P_0}$ exists and is contained in $ L^2_0(P_0)$. For some sufficiently small $\delta > 0$, let $(P_t  : t \in (-\delta, \delta)) \subseteq \mathcal{M}_{\mathrm{np}}$ be a regular path through $P_0$ such that $dP_t = (1+ ts) dP_0$ with score $s \in L^2_0(P_0)$ orthogonal to the loss-based tangent space $\mathcal{S}_{\mathcal{M}_0}(P_0)$. Then, by the above and \ref{cond::lossProjIdent}, our chosen path $(P_t: t \in (-\varepsilon, \varepsilon)$ is such that
 $\Psi_0(P_t) = \Psi(\Pi_0 P_t) = \Psi(Q_{P_0})$
 for all $t$ sufficiently small. Thus, upon differentiation, we find that
 \begin{align*}
   0 =  \frac{d}{dt} \Psi_0(P_t) \big |_{t=0} = \langle D_{0,P_0}, s \rangle_{L^2(P_0)}.
 \end{align*}
 Thus, $D_{0,P_0}$ is necessarily orthogonal to the score $s$ in $L^2_0(P_0)$. However, $s$ was an arbitrary bounded score taken to be orthogonal to the loss-based tangent space $\mathcal{S}_{\mathcal{M}_0}(P_0)$. Since $\mathcal{S}_{\mathcal{M}_0}(P)$ is a closed linear space and bounded scores are dense in $L^2_0(P_0)$, we must have that $D_{0,P_0} \in \mathcal{S}_{\mathcal{M}_0}(P)$. The result then follows.

 \end{proof}

\subsection{Proof of Theorem~\ref{theorem::boundedlinearEIF}}
\begin{proof}[Proof of Theorem~\ref{theorem::boundedlinearEIF}]
This theorem is a direct corollary of the efficient influence function derivations in the proofs of Theorems~4.1 and 4.2 of \cite{chernoRegRiesz}; see also \cite{chernozhukov2018auto} for additional details.

We provide an alternative proof by applying Theorem~\ref{theorem::EIFmain} directly. We begin by verifying its conditions. Although the loss-based projection \(\Pi_0 P\) is nonunique, both the parameter \(\Psi(P)\) and the loss \(\ell(\cdot,P)\) depend only on the covariate distribution and outcome regression of \(P\), so the invariance condition \ref{cond::invariance} holds. In addition, the loss function is smooth at each \(Q \in \mathcal{M}_0\), so Condition~\ref{cond::losssmooth} is satisfied.

More specifically, \(\dot{\ell}_Q(s)\) and \(P\ddot{\ell}_Q(s,v)\) are given by
\[
\dot{\ell}_Q(s): o \mapsto -E_Q[\{Y-\mu_Q(A,W)\}s(O)\mid A=a,W=w]\{y-\mu_Q(a,w)\} - E_Q[s(O)\mid A=a,W=w]
\]
and
\[
P\ddot{\ell}_Q(s,v): o \mapsto E_Q[E_Q[\{Y-\mu_Q(A,W)\}s(O)\mid A,W]E_Q[\{Y-\mu_Q(A,W)\}v(O)\mid A,W]] + E_Q[E_Q[s(O)\mid A,W]E_Q[v(O)\mid A,W]].
\]
Furthermore, \(P\ddot{\ell}_Q\) is a positive semidefinite bilinear form on \(T_{\mathcal{M}_0}(Q)\), so Condition~\ref{cond::innerproductloss} holds. In this example, Condition~\ref{cond::smoothprojection} can be weakened to require only that \(P \mapsto \mu_{P,\mathcal{H}_0}\) be a smooth Hilbert-valued map. This follows because \(P \mapsto \mu_{P,\mathcal{H}_0}\) is the composition of the smooth maps \(P \mapsto \mu_P\) and \(\mu_P \mapsto \mu_{P,\mathcal{H}_0}\) \citep{luedtke2024one}. Finally, Condition~\ref{cond::oraclepathwise} holds whenever the regression functional \(\mu \mapsto E_P[m(W,\mu)]\) is a bounded linear functional on \(\mathcal{H}_0\). The conclusion then follows by applying Theorem~\ref{theorem::EIFmain}.

 We now proceed with the formal proof. Recall that  $\ell(O,Q) := \frac{1}{2}\{Y - \mu_Q(A,W)\}^2 - \log \frac{dQ_X}{d\mu_X}(X).$. We claim that $\dot{\ell}_Q(s): o \mapsto -E_Q[\{Y - \mu_Q(A,W)\}s(O) \mid A = a , W = w] \cdot \{y - \mu_Q(a,w)\} -  E_Q[s(O) \mid A = a, W = w]$ and $P\ddot{\ell}_Q(s,v): o \mapsto E_Q[E_Q[\{Y - \mu_Q(A,W)\}s(O) \mid A, W]E_Q[\{Y - \mu_Q(A,W)\}v(O) \mid A, W]]  + E_Q[E_Q[s(O) \mid A, W]E_Q[v(O) \mid A, W]]$. To see this, first note that
 \begin{align*}
     \frac{d}{dt} \mu_{Q_t} \big |_{t=0} &= \frac{d}{dt} \int y Q_t(dy \mid A=a, W = w) \big |_{t=0}  \\
     &= \int y  \frac{d}{dt}  \log Q_t(dy \mid A=a, W = w) \big |_{t=0} dQ(y \mid A =a, W =w)\\
      &= \int y \{s(O) - E_Q[s(O) \mid A, W]\} dQ(y \mid A =a, W =w)\\
       &= \int \{y - \mu_Q(A,W)s(O) \} dQ(y \mid A =a, W =w)\\
        &= E_Q[\{Y - \mu_Q(A,W)\}s(O) \mid A = a , W = w] ,
 \end{align*}
 where we used that the score component $\frac{d}{dt}  \log Q_t(dy \mid A=a, W = w)$ is equal to the projection of $S(O)$ onto the subtangent space $L^2_0(A,W)$ consisting of functions of $O$ that are mean zero.
 Hence,
 \begin{align*}
\frac{d}{dt} \frac{1}{2} \{y - \mu_{Q_t}(a,w)\}^2 \big |_{t=0} &= - \{y - \mu_{Q}(a,w)\}   \frac{d}{dt} \mu_{Q_t} \big |_{t=0}\\
 &= - \{y - \mu_{Q}(a,w)\} E_Q[\{Y - \mu_Q(A,W)\}s(O) \mid A = a , W = w] .
 \end{align*}
 Next, note that $\frac{d}{dt} \log \frac{dQ_{t,X}}{d\mu_X}(X) \big |_{t=0} = E[s(O) \mid X] - E[s(O)] = E[s(O) \mid X] $. Hence, $\dot{\ell}_Q(s): o \mapsto -E_Q[\{Y - \mu_Q(A,W)\}s(O) \mid A = a , W = w] \cdot \{y - \mu_Q(a,w)\} -  E_Q[s(O) \mid A = a, W = w]$. The second derivative is derived analogously and given by $P\ddot{\ell}_Q(s,v): O \mapsto E_Q[E_Q[\{Y - \mu_Q(A,W)\}s(O) \mid A, W]E_Q[\{Y - \mu_Q(A,W)\}v(O) \mid A, W]]  + E_Q[E_Q[s(O) \mid A, W]E_Q[v(O) \mid A, W]]$. It is immediate that $(s,v) \mapsto P\ddot{\ell}_Q(s,v)$ is a bilinear functional with $P\ddot{\ell}_Q(s,v) \geq 0$. Hence, it defines a semidefinite inner product and \ref{cond::innerproductloss} holds. By Riesz representation theorem, there exists a (potentially non-unique) Riesz representer $s_{0,P}$ in the closure of $T_{\mathcal{M}_0}(\Pi_0 P)$ such that $d\Psi(\Pi_0 P)(s_0) = P\ddot{\ell}_{\Pi_0 P}(s_0, s_{0,P})$ for all $s_0 \in T_{\mathcal{M}_0}(\Pi_0P)$.

First, we determine the precise form of $\{s \in T_{\mathcal{M}_0}(P): d\Psi(\Pi_0 P)(s) = 0\}^{\perp}$. Note that $Q_t \in \mathcal{M}_0$ if and only if $\mu_{Q_t} \in \mathcal{H}_0$ for all $Q_t$. Let $h \in \mathcal{H}_0$ and consider the score $s: o \mapsto \frac{h(a,w)}{v_Q(a,w)}\{y - \mu_Q(a,w)\} + s'(A,W)$, where $v_Q(a,w) := E_Q[\{Y - \mu_Q(A,W)\}^2 \mid A = a, W = w]$ and $s' \in L^2_0(A,W)$. Then, from our previous calculation, 
\begin{align*}
  \frac{d}{dt} \mu_{Q_t}(o) \Big|_{t=0} &= E_Q[\{Y - \mu_Q(A,W)\}s(O) \mid A = a , W = w]  \\
  &=  E_Q[v_Q(A,W)^{-1}\{Y - \mu_Q(A,W)\}^2h(A,W) \mid A = a , W = w]  \\
  &= h(A,W).
\end{align*}
Since $h \in \mathcal{H}_0$, it follows that $\mu_{Q_t}$ locally remains in the model $\mathcal{H}_0$ at $t = 0$ and, hence, $Q_t$ locally remains in the model $\mathcal{M}_0$. We conclude that
\[
 \left\{ o \mapsto s'(a,w) + \frac{h(a,w)}{v_Q(a,w)}\{y - \mu_Q(a,w)\}: h \in \mathcal{H}_0, s' \in L^2_0(A,W) \right\} \subseteq T_{\mathcal{M}_0}(Q).
\]
Moreover, since these submodels generate all possible fluctuation directions $h \in \mathcal{H}_0$ for fluctuating $\mu_Q$, we can show that
\[
\{s \in T_{\mathcal{M}_0}(Q): d\Psi(Q)(s) = 0\}^{\perp} \subseteq \left\{ o \mapsto s'(a,w) + \frac{h(a,w)}{v_Q(a,w)}\{y - \mu_Q(a,w)\}: h \in \mathcal{H}_0, s' \in L^2_0(A,W) \right\} .
\]
For any $s, v \in \{s \in T_{\mathcal{M}_0}(Q): d\Psi(Q)(s) = 0\}^{\perp}$, we can show that
\begin{align*}
P\ddot{\ell}_Q(s,v)   & = E_Q[E_Q[\{Y - \mu_Q(A,W)\}^2 \frac{h_s(a,w)}{v_Q(a,w)} \mid A, W]E_Q[\{Y - \mu_Q(A,W)\}^2  \frac{h_v(a,w)}{v_Q(a,w)} \mid A, W]]   + E_Q[s'(A,W) v'(A,W)] \\
  &=  E_Q[h_s(A,W)h_v(A,W)]   + E_Q[s'(A,W) v'(A,W)],
\end{align*}
which can be verified to be a positive definite inner product on $\{s \in T_{\mathcal{M}_0}(Q): d\Psi(Q)(s) = 0\}^{\perp}$, since $h_s = 0$ and $s' = 0$ if and only if $s = 0$.  We will compute the unique Riesz representer restricted to this subspace and show that it equals $s_{0,P}$.

By \ref{cond::boundedlinearfun}, there exists a representer $\alpha_{P, \mathcal{H}_0} \in \mathcal{H}_0$ such that $E_P[m(W,\mu)] = E_P[\mu(A,W)\alpha_{P, \mathcal{H}_0}(A,W)]$ for all $\mu \in \mathcal{H}_0$. Hence, for any $s \in \left\{ o \mapsto s'(a,w) + \frac{h(a,w)}{v_Q(a,w)}\{y - \mu_Q(a,w)\}: h \in \mathcal{H}_0, s' \in L^2_0(A,W) \right\}$, we have
\begin{align*}
     P\ddot{\ell}_{\Pi_0 P}(s,s) - 2\, d\Psi(\Pi_0 P)(s) &= E_Q[h_s(A,W)^2]   - 2E_P[\alpha_{P, \mathcal{H}_0}(A,W)h_s(A,W)] \\
     & \quad + E_Q[s'(A,W)^2]  - 2E_P[s'(A,W)\{m(W, \mu_{\Pi_0 P}) - \Psi_0(P)\}].
\end{align*}
The above risk is minimized when $h_s = \alpha_{P, \mathcal{H}_0}$ and $s' = m(\cdot, \mu_{\Pi_0 P}) - \Psi_0(P)$. Thus, we conclude that the Riesz representer over this subspace is given by
$$s_{0,P}: o \mapsto \frac{\alpha_{P, \mathcal{H}_0}(a,w)}{v_Q(a,w)}\{y - \mu_Q(a,w)\} +  m(w, \mu_{\Pi_0 P}) - \Psi_0(P). $$
Moreover, for any $v \in T_{\mathcal{M}_0}(\Pi_0 P)$, we can verify that
\begin{align*}
    P\ddot{\ell}_{\Pi_0 P}(s_{0,P},v) &= E_P[\alpha_{P, \mathcal{H}_0}(A,W)\dot{\mu}_{Q}(v)]   + E_Q[s'(A,W) v'(A,W)] \\
    & = E_P[m(W, \dot{\mu}_{Q}(v))]  + E_Q[s'(A,W) v'(A,W)] \\
    &= d\Psi(\Pi_0 P)(v),
\end{align*}
where $\dot{\mu}_{Q}(v) = \frac{d}{dt} \mu_{Q_{t,v}} \big |_{t=0}$ with $Q_{t,v}$ being a regular path centered at $Q$ with score $v$ at $Q$. Hence, $s_{0,P}$ is also a Riesz representer for all of $T_{\mathcal{M}_0}(\Pi_0 P)$ and satisfies, for all $s_0 \in T_{\mathcal{M}_0}(\Pi_0 P)$, that
$$d\Psi(\Pi_0 P)(s_0) =  P\ddot{\ell}_{\Pi_0 P}(s_{0,P},s_0). $$
Applying Theorem \ref{theorem::EIFmain}, we conclude that  $D_{0,P} = -\dot{\ell}_{\Pi_0 P}(s_{0,P})$, where $s_{0,P}(z) := v_{P,\mathcal{H}_0}^{-1}(a,w)\alpha_{P,\mathcal{H}_0}(a,w)\{y - \mu_{P,\mathcal{H}_0}(a,w)\} + m(w,\mu_{P,\mathcal{H}_0}) - \Psi_0(P)$ and $v_{P,\mathcal{H}_0}(a,w) := E_P[\{Y - \mu_{P,\mathcal{H}_0}(A,W)\}^2 \mid A = a, W = w]$.

In this example, \ref{cond::smoothprojection} can be relaxed to require only that $P \mapsto \mu_{P,\mathcal{H}_0}$ is a smooth Hilbert-valued map, which holds since it is the composition of the smooth maps $P \mapsto \mu_P$ and $\mu_P \mapsto \mu_{P,\mathcal{H}_0}$ \citep{luedtke2024one}. Finally, we establish that \ref{cond::oraclepathwise} holds as long as the regression functional $\mu \mapsto E_P[m(W,\mu)]$ is a bounded linear functional on $\mathcal{H}_0$. By applying Theorem \ref{theorem::EIFmain}, we obtain the following theorem establishing the pathwise differentiability of $\Psi_0$ and the explicit form of its efficient influence function.  
 
 \end{proof}

\section{Proofs for Section \ref{section::modelapprox}}

\subsection{Proof of Theorem \ref{theorem::exactBiasOracle}}

\begin{proof}[Proof of Theorem \ref{theorem::exactBiasOracle}]
By Condition \ref{cond::oracleProjInModelFinal}, $\Psi_{n,0}$ is pathwise differentiable at $P_{n,0}$ and the efficient influence function $D_{n,0, P_{n,0}}$ exists. We have the exact expansion
\begin{align}
        \Psi_n(P_0) - \Psi_0(P_0)\ &=\
        \Psi_{n,0}(P_{n,0}) -\Psi_{n,0}(P_0) \nonumber\\
        &=\  -  P_0 D_{n, 0,P_{n,0}}  + R_{n,0}   \nonumber\\
        &=\  (P_{n,0} - P_0)  D_{n,0,P_{n,0}}  + R_{n,0}  \label{proof::eqn::exactBiasOracleExpansion}
\end{align}
where $R_{n,0}  := \Psi_{n,0}(P_{n,0}) - \Psi_{n,0}(P_0) + P_{0} D_{n,0,P_{n,0}}$ is the exact second-order remainder. Here, we used the mean zero property that $P_{n,0} D_{n,0,P_{n,0}} = 0$.

Now, by  \ref{cond::lossProjDeriv}, we have that the minimizer $P_{n,0} = \Pi_n P_0$ satisfies
\begin{equation}
     P_0  \left\{\frac{d}{dt} \ell(\cdot,Q_t) \big|_{t=0} \right\}= 0
     \label{proof::eqn::exactBiasOracle1}
\end{equation}
for all regular paths $(Q_t: t \in (-\delta, \delta)) \subseteq \mathcal{M}_n$ such that $Q_t = \Pi_n P_0$ at $t=0$. By definition, we have that $ \bar{D}_{n, 0,P_{n,0}}$ is contained in the loss-based tangent space $\mathcal{S}_{\mathcal{M}_n}(P_{n,0}) \subseteq L^2_0(P_{n,0})$. Thus, there exists some regular path $(Q_t: t \in (-\delta, \delta)) \subseteq \mathcal{M}_n$ such that $Q_t = P_{n,0}$ at $t=0$ and $\frac{d}{dt} \ell(\cdot,Q_t) \big|_{t=0} =  \bar{D}_{n, 0,P_{n,0}}$. Moreover, by Equation \eqref{proof::eqn::exactBiasOracle1}, we must have that
$$P_0 ( \bar{D}_{n, 0,P_{n,0}} ) = P_0 \left\{ \frac{d}{dt} \ell(\cdot,Q_t) \big|_{t=0}\right\}  = 0\ .$$
Since $\bar{D}_{n, 0,P_{n,0}} \in L^2_0(P_{n,0})$ is centered under $P_{n,0}$, we also have $(P_{n,0} - P_0) \bar{D}_{n, 0,P_{n,0}} = 0$. Combining this with Equation \eqref{proof::eqn::exactBiasOracleExpansion}, we obtain the expansion
\begin{align*}
        \Psi_n(P_0) - \Psi_0(P_0) &=  (P_{n,0}- P_0) ( D_{n, 0,P_{n,0}} ) +  R_{n,0} + 0 \\
        &=  (P_{n,0}- P_0) ( D_{n, 0,P_{n,0}} - \bar{D}_{n, 0,P_{n,0}} ) + R_{n,0} + (P_{n,0} - P_0) \bar{D}_{n, 0,P_{n,0}} \\
        &=  (P_{n,0}- P_0) ( D_{n, 0,P_{n,0}} - \bar{D}_{n, 0,P_{n,0}} ) + R_{n,0}
\end{align*}
as desired.
\end{proof}

\subsection{Proof of Theorem \ref{theorem::oraclebiaslinear}}

\begin{proof}[Proof of Theorem \ref{theorem::oraclebiaslinear}]
By Condition~\ref{cond::boundedlinearfun} and the Riesz representation theorem, there exist Riesz representers $\alpha_{0,\mathcal{H}_{n}}$ and $\alpha_{0,\mathcal{H}_{n,0}}$ for the linear functional $\mu \mapsto E_0[m(W, \mu_{0, \mathcal{H}_n})]$ in the models $\mathcal{H}_{n}$ and $\mathcal{H}_{n,0}$, respectively. Thus, we have
\begin{align*}
    \Psi_n(P_0) - \Psi(P_0) 
    &= E_0[m(W, \mu_{0, \mathcal{H}_n})] - E_0[m(W, \mu_{0})] \\
    &= E_0[\alpha_{0,\mathcal{H}_{n,0}}(A,W)\{\mu_{0, \mathcal{H}_n}(A,W) - \mu_{0}(A,W)\}],
\end{align*}
where we used that $\mu_{0, \mathcal{H}_n} \in \mathcal{H}_n \subseteq \mathcal{H}_{n,0}$ and $\mu_0 \in \mathcal{H}_0 \subseteq \mathcal{H}_{n,0}$. Hence,
\begin{align*}
    \Psi_n(P_0) - \Psi(P_0) 
    &= \langle \alpha_{0,\mathcal{H}_{n,0}}, \mu_{0, \mathcal{H}_n} - \mu_{0} \rangle_{P_0} \\
    &= \langle \alpha_{0,\mathcal{H}_{n,0}} - \alpha_{0,\mathcal{H}_{n}}, \mu_{0, \mathcal{H}_n} - \mu_{0} \rangle_{P_0},
\end{align*}
where we used the orthogonality property of the projection $\mu_{0, \mathcal{H}_n}$ to note that
\[
\langle \alpha_{0,\mathcal{H}_{n}}, \mu_{0, \mathcal{H}_n} - \mu_{0} \rangle_{P_0} = 0.
\]
The result follows.
\end{proof}

Before proving Lemma~\ref{lemma::lipschitzdependent}, we first establish the following lemma. In this lemma, let \(X \in \mathcal{X}\) be a covariate and \(V \in \mathbb{R}\) an outcome. Let \(\varphi_n: \mathcal{X} \to \mathbb{R}^d\) and \(\varphi_0: \mathcal{X} \to \mathbb{R}^d\) be feature maps, and define \(\varphi_{n,0}: \mathcal{X} \to \mathbb{R}^{2d}\) as the combined feature map \(\varphi_{n,0}(x) = (\varphi_n(x), \varphi_0(x))\). For a given feature map \(\varphi\), let \(f_{\varphi}: x \mapsto E_0[V \mid \varphi(X) =  \varphi(x) ]\) denote the corresponds regression function.

\begin{lemma}
      \label{lemma::lipschitzdependent2} Suppose that $(t_1, t_2) \mapsto E_0[f_{(\varphi_n, \varphi_0)}(X) \mid \varphi_n(X) = t_1, \varphi_0(X) = t_2, \mathcal{D}_n]$ is almost surely $L$-Lipschitz continuous. Then,
   $$\|f_{\varphi_n} - f_{(\varphi_n, \varphi_0)} \|_{P_0} \lesssim  \|\|\varphi_0(\cdot) - \varphi_n(\cdot)\|_{\mathbb{R}^d} \|_{L^2(P_0)}  \text{ and } \|f_{\varphi_n} - f_{\varphi_0} \|_{P_0} \lesssim  \|\|\varphi_0(\cdot) - \varphi_n(\cdot)\|_{\mathbb{R}^d} \|_{L^2(P_0)}.$$
\end{lemma}

\begin{proof}

Let $g: \mathbb{R}^k \times \mathbb{R}^k \rightarrow \mathbb{R}$ be a Lipschitz continuous function with constant $L > 0$.  By Lipschitz continuity, we have that
\begin{align*}
  \left|  g(\varphi_n(x), \varphi_0(x)) - E_0[g(\varphi_n(X), \varphi_0(X))  |  \varphi_n(X) = \varphi_n(x)] \right| &= \left| E_0[ g(\varphi_n(x), \varphi_0(x))  - g(\varphi_n(x), \varphi_0(X))  |  \varphi_n(X) = \varphi_n(x)] \right|\\
  & \leq  E_0[ \left|g(\varphi_n(x), \varphi_0(x))  - g(\varphi_n(x), \varphi_0(X)) \right|  |  \varphi_n(X) = \varphi_n(x)]  \\
    & \leq L E_0[\left\|\varphi_0(x) - \varphi_0(X) \right\|_{\mathbb{R}^d}  |  \varphi_n(X) = \varphi_n(x)]  .
\end{align*}
On the event $\{ \varphi_n(X) = \varphi_n(x)\}$, we know 
\begin{align*}
    \|\varphi_0(x) - \varphi_0(X) \|_{\mathbb{R}^d} &\leq \|\varphi_0(x) - \varphi_n(x)\|_{\mathbb{R}^d} + \|\varphi_n(x) - \varphi_n(X)\|_{\mathbb{R}^d} + \|\varphi_0(X) - \varphi_n(X)\|_{\mathbb{R}^d} \\
    & \leq \|\varphi_0(x) - \varphi_n(x)\|_{\mathbb{R}^d} +  \|\varphi_0(X) - \varphi_n(X)\|_{\mathbb{R}^d}.
\end{align*} 
Therefore, 
\begin{align*}
    \left|  g(\varphi_n(x), \varphi_0(x)) - E_0[g(\varphi_n(X), \varphi_0(X))  |  \varphi_n(X) = \varphi_n(x)] \right| & \lesssim E_0[\left\|\varphi_0(x) - \varphi_0(X) \right\|_{\mathbb{R}^d}  \mid    \varphi_n(X) = \varphi_n(x)] \\ 
     & \lesssim \|\varphi_0(x) - \varphi_n(x)\|_{\mathbb{R}^d}    + E_0[ \|\varphi_0(X) - \varphi_n(X)\|_{\mathbb{R}^d}  \mid   \varphi_n(X) = \varphi_n(x)] .
\end{align*}

By assumption, the map $(t_1, t_2) \mapsto E_0[f_{(\varphi_n, \varphi_0)}(X) \mid \varphi_n(X) = t_1, \varphi_0(X) = t_2, \mathcal{D}_n]$ is almost surely $L$-Lipschitz continuous. Moreover, by the law of total expectation,
\begin{align*}
    f_{\varphi_n}(X)& = E_0[E_0[f_{(\varphi_n, \varphi_0)}(X) \mid \varphi_n(X), \varphi_0(X), \mathcal{D}_n] \mid \varphi_n(X), \mathcal{D}_n]  \\
     f_{\varphi_n}(X) &= E_0[E_0[f_{(\varphi_n, \varphi_0)}(X) \mid \varphi_n(X), \varphi_0(X), \mathcal{D}_n] \mid \varphi_0(X), \mathcal{D}_n]  \\
\end{align*}
Hence, taking $g(t_1, t_2) = E_0[f_{(\varphi_n, \varphi_0)}(X) \mid \varphi_n(X) = t_1, \varphi_0(X) = t_2, \mathcal{D}_n]$, we obtain the pointwise bound:
\begin{align*}
    \left|  f_{(\varphi_n, \varphi_0)}(x)  - f_{\varphi_n}(x)  \right| & \lesssim \|\varphi_0(x) - \varphi_n(x)\|_{\mathbb{R}^d}     + E_0[ \|\varphi_0(X) - \varphi_n(X)\|_{\mathbb{R}^d}  \mid   \varphi_n(X) = \varphi_n(x)] .
\end{align*}  
Taking the $L^2(P_0)$ norm of both sides and using that conditional expectations are projections in $L^2(P_0)$, we find
\begin{align*}
   \| f_{(\varphi_n, \varphi_0)} - f_{\varphi_n}\|_{L^2(P_0)} & \lesssim \|\|\varphi_0(\cdot) - \varphi_n(\cdot)\|_{\mathbb{R}^d} \|_{L^2(P_0)}      +  \|E_0[ \|\varphi_0(X) - \varphi_n(X)\|_{\mathbb{R}^d}  \mid   \varphi_n(X) = \varphi_n(\cdot)] \|_{L^2(P_0)} \\
  & \lesssim \|\|\varphi_0(\cdot) - \varphi_n(\cdot)\|_{\mathbb{R}^d} \|_{L^2(P_0)} . 
\end{align*}
By symmetry, an identical argument gives:
\begin{align*}
   \| f_{(\varphi_n, \varphi_0)} - f_{\varphi_0}\|_{L^2(P_0)} & \lesssim \|\|\varphi_0(\cdot) - \varphi_n(\cdot)\|_{\mathbb{R}^d} \|_{L^2(P_0)} .
\end{align*}
By the triangle inequality, we also have
\begin{align*}
   \| f_{\varphi_n} - f_{\varphi_0}\|_{L^2(P_0)} & \lesssim \|\|\varphi_0(\cdot) - \varphi_n(\cdot)\|_{\mathbb{R}^d} \|_{L^2(P_0)} .
\end{align*}

\end{proof}

\subsection{Proof of Lemma \ref{lemma::lipschitzdependent}}

\begin{proof}[Proof of Theorem \ref{lemma::lipschitzdependent}]
By Condition~\ref{cond::boundedlinearfun} and Theorem~\ref{theorem::oraclebiaslinear}, we have
\[
|\Psi_n(P_0) - \Psi(P_0)| \leq \|\operatorname{Proj}_n(\alpha_{0,\mathcal{H}_{n,0}}) - \alpha_{0,\mathcal{H}_{n,0}}\|_{P_0} \, \| \operatorname{Proj}_n(\mu_0) - \mu_0 \|_{P_0}.
\]
By the definitions of $\mathcal{H}_n$, $\mathcal{H}_0$, and $\mathcal{H}_{n,0}$, and using Condition~\ref{cond::lipschitzFeature}, we may apply Lemma~\ref{lemma::lipschitzdependent2} with \(X = (A,W)\), \(V = Y\), and \(V = \alpha_{0,\mathcal{H}}(A,W)\) to conclude that
\begin{align*}
    \|\operatorname{Proj}_n(\mu_0) - \mu_0\|_{P_0} 
    &\lesssim \sqrt{\int \|\varphi_n(a,w) - \varphi_0(a,w)\|_{\mathbb{R}^d}^2 \, P_0(dw, da)}, \\
    \|\operatorname{Proj}_n(\alpha_{0,\mathcal{H}_{n,0}}) - \alpha_{0,\mathcal{H}_{n,0}}\|_{P_0} 
    &\lesssim \sqrt{\int \|\varphi_n(a,w) - \varphi_0(a,w)\|_{\mathbb{R}^d}^2 \, P_0(dw, da)}.
\end{align*}
For the latter bound, we used the projection properties: \(E_0[\alpha_{0,\mathcal{H}}(A,W) \mid \varphi_n(A,W), \varphi_0(A,W), \mathcal{D}_n] = \alpha_{0,\mathcal{H}_{n,0}}(A,W)\), \(E_0[\alpha_{0,\mathcal{H}}(A,W) \mid \varphi_0(A,W), \mathcal{D}_n] = \alpha_{0,\mathcal{H}_0}(A,W)\), and \(E_0[\alpha_{0,\mathcal{H}}(A,W) \mid \varphi_n(A,W), \mathcal{D}_n] = \alpha_{n,\mathcal{H}_{n,0}}(A,W)\). Hence,
\[
|\Psi_n(P_0) - \Psi(P_0)| \lesssim \int \|\varphi_n(a,w) - \varphi_0(a,w)\|_{\mathbb{R}^d}^2 \, P_0(dw, da).
\]
\end{proof}

\subsection{Proof sketch for marginal correlation example}
\label{appendix:example:proof}
 \begin{proof}[Proof sketch for Example~\ref{example::screening}]
We first show that the screening rule introduces no false positives in the
truncated dictionary. Suppose that \(Y\) is subgaussian and that
\(\sup_{j\le k(n)} \|\varphi_j\|_\infty < \infty\). Then each product
\(\varphi_j(A,W)Y\) is subgaussian up to constants, with mean
\[
E_0\{\varphi_j(A,W)Y\} = \beta_j.
\]
Standard maximal inequalities therefore give
\[
\sup_{j\le k(n)}
\left|
\frac{1}{n}\sum_{i=1}^n \varphi_j(A_i,W_i)Y_i - \beta_j
\right|
=
O_p\!\left(\sqrt{\frac{\log k(n)}{n}}\right)
=
O_p(c_n).
\]
Let
\[
\mathcal E_n
:=
\left\{
\sup_{j\le k(n)}
\left|
\frac{1}{n}\sum_{i=1}^n \varphi_j(A_i,W_i)Y_i - \beta_j
\right|
\le C c_n
\right\},
\]
where \(P(\mathcal E_n)\to 1\). Since the screening threshold is
\((\log n)c_n\), on \(\mathcal E_n\), for every \(j\le k(n)\) with
\(\beta_j=0\),
\[
\left|
\frac{1}{n}\sum_{i=1}^n \varphi_j(A_i,W_i)Y_i
\right|
\le C c_n
< (\log n)c_n
\]
for all sufficiently large \(n\). Hence \(\widehat S_n\subseteq S_{0,k}\)
with probability tending to one, so asymptotically
\(\mathcal H_n\subseteq \mathcal H_0\).

We now bound the approximation error for \(\mu_0\). Because \(\Phi\) is an
orthonormal basis and \(\operatorname{Proj}_n\) is the \(L^2(P_0)\)-orthogonal
projection onto \(\mathcal H_n\), on the event
\(\widehat S_n\subseteq S_{0,k}\),
\[
\|\operatorname{Proj}_n(\mu_0)-\mu_0\|_{P_0}^2
=
\sum_{j\le k(n)}
1\{j\notin \widehat S_n\}\beta_j^2
+
\sum_{j>k(n)}\beta_j^2 .
\]
The second term is the truncation bias. On \(\mathcal E_n\), if
\(j\le k(n)\) and \(j\notin \widehat S_n\), then
\[
\left|
\frac{1}{n}\sum_{i=1}^n \varphi_j(A_i,W_i)Y_i
\right|
\le (\log n)c_n,
\]
so
\[
|\beta_j|
\le
\left|
\frac{1}{n}\sum_{i=1}^n \varphi_j(A_i,W_i)Y_i - \beta_j
\right|
+
\left|
\frac{1}{n}\sum_{i=1}^n \varphi_j(A_i,W_i)Y_i
\right|
\le
C c_n + (\log n)c_n
\lesssim
(\log n)c_n.
\]
Therefore,
\[
\sum_{j\le k(n)}
1\{j\notin \widehat S_n\}\beta_j^2
\lesssim
\sum_{j\le k(n)} \min\{\beta_j^2,((\log n)c_n)^2\}.
\]
Under the weak sparsity condition
\(\sum_{j=1}^\infty |\beta_j|^{q_\mu}<\infty\) with \(q_\mu\in(0,2)\), the
elementary inequality
\[
\min\{x^2,t^2\}\le t^{2-q_\mu}|x|^{q_\mu}
\qquad (x\in\mathbb R,\ t>0)
\]
implies that
\[
\sum_{j\le k(n)} \min\{\beta_j^2,((\log n)c_n)^2\}
=
O\!\left(\{(\log n)c_n\}^{\,2-q_\mu}\right).
\]
Taking square roots yields
\[
\|\operatorname{Proj}_n(\mu_0)-\mu_0\|_{P_0}
=
O_p\!\left(\{(\log n)c_n\}^{\,1-q_\mu/2}\right)
+
O\!\left(\left(\sum_{j>k(n)}\beta_j^2\right)^{1/2}\right).
\]

Finally, consider the Riesz representer. Since
\(\alpha_{0,\mathcal H_0}=\sum_{j\in S_0}\gamma_j\varphi_j\), on the event
\(\widehat S_n\subseteq S_{0,k}\),
\[
\|\operatorname{Proj}_n(\alpha_{0,\mathcal H_0})
-\alpha_{0,\mathcal H_0}\|_{P_0}^2
=
\sum_{j\in S_0\cap\{1,\ldots,k(n)\}}
1\{j\notin \widehat S_n\}\gamma_j^2
+
\sum_{\substack{j>k(n)\\ j\in S_0}}\gamma_j^2 .
\]
Using the domination condition
\(|\gamma_j|\lesssim |\beta_j|^a\) for \(j\in S_0\), we obtain on
\(\mathcal E_n\) that
\[
\sum_{j\in S_0\cap\{1,\ldots,k(n)\}}
1\{j\notin \widehat S_n\}\gamma_j^2
\lesssim
\sum_{j\le k(n)} \min\{|\beta_j|^{2a},((\log n)c_n)^{2a}\}.
\]
Again using \(\sum_{j=1}^\infty |\beta_j|^{q_\mu}<\infty\) and the inequality
\[
\min\{|x|^{2a},t^{2a}\}\le t^{2a-q_\mu}|x|^{q_\mu},
\qquad a>q_\mu/2,
\]
we get
\[
\sum_{j\le k(n)} \min\{|\beta_j|^{2a},((\log n)c_n)^{2a}\}
=
O\!\left(\{(\log n)c_n\}^{\,2a-q_\mu}\right).
\]
Therefore,
\[
\|\operatorname{Proj}_n(\alpha_{0,\mathcal H_0})
-\alpha_{0,\mathcal H_0}\|_{P_0}
=
O_p\!\left(\{(\log n)c_n\}^{\,a-q_\mu/2}\right)
+
O\!\left(\left(\sum_{\substack{j>k(n)\\ j\in S_0}}\gamma_j^2\right)^{1/2}\right).
\]
Thus both approximation errors vanish provided \(k(n)\to\infty\) and the
truncation tails vanish.
\end{proof}

\section{Proofs for Section \ref{section::theory}}

\subsection{Proofs of Theorem  \ref{theorem::oracleEff} }

\begin{proof}[Proof of Theorem  \ref{theorem::oracleEff}]

By \ref{cond::debiased}--\ref{cond::Donsker}, we have
\begin{align*}
    \widehat{\psi}_n - \Psi_n(P_0)\ &=\ (P_n-P_0 )D_{n, P_0} + o_p(n^{-1/2}) \\  
    & =\  (P_n-P_0 )D_{0,P_0} + (P_n-P_0 )(D_{n, P_0} - D_{0,P_0}) + o_p(n^{-1/2}) \\
     & =\  (P_n-P_0 )D_{0,P_0} + o_p(n^{-1/2})\ ,
\end{align*}
where the final equality uses,  by \ref{cond::Donsker}, that $(P_n-P_0 )(D_{n, P_0} - D_{0,P_0}) = o_p(n^{-1/2})$. The result now follows. We note that \ref{cond::consDn}, while not used in this proof, is typically required to establish \ref{cond::Donsker}.

Under \ref{cond::oracleProjInModelFinal} and Theorem \ref{theorem::exactBiasOracle}, we have the bound
\begin{align*}
        \Psi_n(P_0) - \Psi_0(P_0)=  (P_{n,0}- P_0) ( D_{0,P_{n,0}} - \bar{D}_{n, 0,P_{n,0}}) + R_{n,0}  + o_p(n^{-1/2})\ .   
\end{align*}
 Combining this our earlier expansion, we obtain the expansion
\begin{align*}
    \widehat{\psi}_n - \Psi_0(P_0)\ & =\ (P_n-P_0 )D_{0,P_0}  + o_p(n^{-1/2})\\
    & \qquad + (P_{n,0} - P_0)( D_{0,P_{n,0}} - \bar{D}_{n, 0,P_{n,0}}) + R_{n,0} + o_p(n^{-1/2})\ .
\end{align*}
By \ref{cond::doubleRemrootnFinal}, we have that $ (P_{n,0} - P_0)( D_{0,P_{n,0}} - \bar{D}_{n, 0,P_{n,0}}) + R_{n,0} = o_p(n^{-1/2})$, and so, $  \widehat{\psi}_n - \Psi_0(P_0)  = (P_n-P_0 )D_{0,P_0}  + o_p(n^{-1/2})$ as desired. It follows that $\widehat{\psi}_n $ is asymptotically linear at $P_0$ for $\Psi_0:\mathcal{M}_{\mathrm{np}} \rightarrow \mathbb{R}$ with influence function being the efficient influence function of $\Psi_0$ at $P_0$. Hence, $\widehat{\psi}_n $ is $P_0$--efficient for $\Psi_0$ relative to $\mathcal{M}_{\mathrm{np}}$ \citep{bickel1993efficient}. Moreover, since  $\widehat{\psi}_n $ is efficient, it is necessarily regular for $\Psi_0$ relative to $\mathcal{M}_{\mathrm{np}}$ \citep{vandervaart2000asymptotic}. The limiting distribution result follows immediately from the central limit theorem and Slutsky's lemma.\end{proof}

  \subsection{Proof of Lemma \ref{lemma::stableLogLik}}

\label{appendix::proofstable}
 \begin{proof}[Proof of Lemma \ref{lemma::stableLogLik}]
Recall that $D_{n,0,P_0}$ is the $P_0$-EIF of the projection parameter $\Psi_{n,0}$ induced by the union model $\mathcal{M}_{n,0}$. 
Let $s_{n,P_{n,0}}$ and $s_{0,P_{0}}$ denote the Riesz representers of the projection parameters $\Psi_n$ and $\Psi_0$, respectively, 
with respect to the Hessian inner products $P_{0}\ddot{\ell}_{P_{n,0}}$ and $P_{0}\ddot{\ell}_{P_{0}}$.

By the triangle inequality and Condition~\ref{cond::weakconstencyOracle}, 
\begin{align*}
  \norm{D_{n,P_0} - D_{0,P_0}}_{P_0}
  &\leq\ \norm{D_{n,P_{n,0}} - D_{0,P_{0}}}_{P_0}
       + \norm{D_{n,P_{n,0}} - D_{n,P_0}}_{P_0} \\
  &\leq\ \norm{D_{n,P_{n,0}} - D_{0,P_{0}}}_{P_0} + o_p(1).
\end{align*}
Furthermore, by the triangle inequality and the fact that $\norm{D_{n,0,P_{n,0}} - D_{n,0,P_0}}_{P_0} = o_p(1)$ by Condition~\ref{cond::weakconstencyOracle}, 
\begin{align*}
    \norm{D_{n,P_{n,0}} - D_{0,P_{0}}}_{P_0}
    &\leq\ \norm{D_{n,P_{n,0}} - D_{n,0,P_{n,0}}}_{P_0} 
           + \norm{D_{0,P_0} - D_{n,0,P_0}}_{P_0} 
           + \norm{D_{n,0,P_{n,0}} - D_{n,0,P_0}}_{P_0} \\
    &\leq\ \norm{D_{n,P_{n,0}} - D_{n,0,P_{n,0}}}_{P_0} 
           + \norm{D_{0,P_0} - D_{n,0,P_0}}_{P_0} 
           + o_p(1).
\end{align*}
Hence, using the representation of the EIFs in Theorem~\ref{theorem::lossbasedEIF} and the Lipschitz property in Condition~\ref{cond::lipschitzloss}, 
\begin{align*}
    \norm{D_{n,P_{n,0}} - D_{0,P_{0}}}_{P_0}
    &\leq\ \norm{\dot{\ell}_{P_{n,0}}(s_{n,P_{n,0}}) - \dot{\ell}_{P_{n,0}}(s_{n,0,P_{n,0}})}_{P_0}
      + \norm{\dot{\ell}_{P_{0}}(s_{0,P_{0}}) - \dot{\ell}_{P_{0}}(s_{n,0,P_{0}})}_{P_0} \\
    &\lesssim\ \Bigl\{ P_0\,\ddot{\ell}_{P_{n,0}}\bigl(s_{n,P_{n,0}} - s_{n,0,P_{n,0}},\, s_{n,P_{n,0}} - s_{n,0,P_{n,0}}\bigr) \Bigr\}^{1/2} \\
    &\quad + \Bigl\{ P_0\,\ddot{\ell}_{P_{0}}\bigl(s_{0,P_{0}} - s_{n,0,P_{0}},\, s_{0,P_{0}} - s_{n,0,P_{0}}\bigr) \Bigr\}^{1/2}.
\end{align*}
By \eqref{eqn::autodml}, using that $\Pi_0 P_0 = P_0$, $s_{0,P_{0}}$ satisfies 
\begin{align*}
 P_0\ddot{\ell}_{P_0}(s_{0,P_{0}},s_{0,P_{0}}) - 2\, d\Psi(P_0)(s_{0,P_{0}})   
   &= \min_{s \in \overline{T}_{\mathcal{M}_0}(P_0)} \Bigl\{ P_0\ddot{\ell}_{P_0}(s,s) - 2\, d\Psi(P_0)(s) \Bigr\} \\
   &= \min_{s \in \overline{T}_{\mathcal{M}_0}(P_0)} \Bigl\{ P_0\ddot{\ell}_{P_0}(s,s) - 2\, P_0\ddot{\ell}_{P_0}(s,s_{n,0,P_0}) \Bigr\} \\
   &= \min_{s \in \overline{T}_{\mathcal{M}_0}(P_0)} P_0\ddot{\ell}_{P_0}(s - s_{n,0,P_0},\, s - s_{n,0,P_0}).
\end{align*}
Similarly,  using that $\Pi_n P_{n,0} = P_{n,0}$, $s_{n,P_{n,0}}$ satisfies
\begin{align*}
 P_0\ddot{\ell}_{P_{n,0}}(s_{n,P_{n,0}},s_{n,P_{n,0}}) - 2\, d\Psi(P_{n,0})(s_{n,P_{n,0}}) 
   &= \min_{s \in \overline{T}_{\mathcal{M}_n}(P_{n,0})} P_0\ddot{\ell}_{P_{n,0}}(s - s_{n,0,P_{n,0}},\, s - s_{n,0,P_{n,0}}).
\end{align*}
Hence, by \eqref{cond::consistentTangentSpace},
\begin{align*}
    \norm{D_{n,P_{n,0}} - D_{0,P_{0}}}_{P_0}
    &\lesssim\ \min_{s \in \overline{T}_{\mathcal{M}_0}(P_0)} P_0\ddot{\ell}_{P_0}(s - s_{n,0,P_0},\, s - s_{n,0,P_0}) \\
    &\quad + \min_{s \in \overline{T}_{\mathcal{M}_n}(P_{n,0})} P_0\ddot{\ell}_{P_{n,0}}(s - s_{n,0,P_{n,0}},\, s - s_{n,0,P_{n,0}})\\
     &\lesssim o_p(1)
\end{align*}
Finally, combining all our displays, we conclude that
\begin{align*}
    \norm{D_{n,P_0} - D_{0,P_0}}_{P_0} = o_p(1),
\end{align*}
as desired.

\end{proof}

\begin{proof}[Proof of Theorem \ref{theorem::oracleRegularity}]
    By Theorem \ref{theorem::oracleEff} and that $\Psi_0(P_0) = \Psi(P_0)$, the ADML estimator $\widehat{\psi}_n$ is asymptotically linear for $\Psi$ at $P_0$ with influence function being the efficient influence function of $\Psi_0$. Since $\Psi_0(P) = \Psi(P)$ for all $P \in \mathcal{M}_0$, regularity of the ADML estimator for $\Psi_0$ over $\mathcal{M}_{\mathrm{np}}$ implies that the ADML estimator is regular for $\Psi$ over the oracle submodel $\mathcal{M}_0$. Under the loss-based tangent space conditions, we have by Theorem \ref{theorem::lossbasedEIF} that $D_{0,P_0}$ is the efficient influence function of $\Psi$ relative to $\mathcal{M}_0$. Thus, the ADML estimator is asymptotically linear with influence function equal to the efficient influence function of $\Psi$ relative to $\mathcal{M}_0$. It follows that the ADML estimator is asymptotically efficient \citep{vandervaart2000asymptotic}.\end{proof}

 \begin{proof}[Proof of Theorem \ref{theorem::irreg}]

     We assumed that $\Psi$ is pathwise differentiable at $P_0$ relative to $\mathcal{M}$ with efficient influence function $D_{\mathcal{M},P_0}$. Using that $\Psi_0(P_0) = \Psi(P_0)$, we first observe that 
     \begin{align*}
        \Psi_0(P_{0,hn^{-1/2}}) - \Psi(P_{0,hn^{-1/2}})\ &=\ \Psi_0(P_{0,hn^{-1/2}}) - \Psi_0(P_0)  + \Psi(P_0)- \Psi(P_{0,hn^{-1/2}})\\
         &=\ hn^{-1/2}\{ d\Psi_0(P_0)(S) - d\Psi(P_0)(S)\} + o(n^{-1/2})\\
          &=\ hn^{-1/2}\langle S, D_{0,P_0}- D_{\mathcal{M},P_0} \rangle_{P_0} + o(n^{-1/2})\ ,
     \end{align*}
where the final two equalities use pathwise differentiability of $\Psi:\mathcal{M} \rightarrow \mathbb{R}$ and $\Psi_0:\mathcal{M}_{\mathrm{np}} \rightarrow \mathbb{R}$. Hence, we find that
\begin{align*}
    \sqrt{n}\left\{\widehat{\psi}_n - \Psi(P_{0,hn^{-1/2}}) \right\}\ &=\ \sqrt{n}\left\{\widehat{\psi}_n  - \Psi_0(P_{0,hn^{-1/2}})\right\} + \sqrt{n}\left\{ \Psi_0(P_{0,hn^{-1/2}}) - \Psi(P_{0,hn^{-1/2}}) \right\} \\
   & =\ \sqrt{n}\,(P_n - P_0)D_{0,P_0} +  h\langle S, D_{0,P_0}- D_{\mathcal{M},P_0} \rangle_{P_0} + o(1)\ .
\end{align*} 
The result then follows from Slutsky's theorem.\end{proof}

\section{Proofs for Section \ref{adml::linearfunc}}
\begin{proof}[Proof of Theorem \ref{theorem::ALlinear}]

We first consider the model approximation error $\Psi_n(P_0) - \Psi(P_0)$. By Lemma~\ref{theorem::oraclebiaslinear} and Condition~\ref{cond::oraclebiaslinear}, we have
\[
|\Psi_n(P_0) - \Psi(P_0)| \leq \|\operatorname{Proj}_n (\alpha_{0,\mathcal{H}_{n,0}}) - \alpha_{0,\mathcal{H}_{n,0}}\|_{P_0}\, \|\operatorname{Proj}_n(\mu_0) - \mu_0\|_{P_0} = o_p(n^{-1/2}).
\]

Next, we consider the estimation error $\widehat{\psi}_n - \Psi_n(P_0)$.
 Denote the estimated influence function 
$$\widehat{D}_n(o) := \alpha_n(a,x)\{y - \mu_n(a,x)\} + m(w, \mu_n) - \frac{1}{n} \sum_{i=1}^n m(W_i, \mu_n).$$ The estimator error for the working parameter $\Psi_n(P_0)$ may be decomposed as:
\begin{align*}
    \widehat{\psi}_n  - \Psi_n(P_0) &= \frac{1}{n} \sum_{i=1}^n m(W_i, \mu_n) + \frac{1}{n} \sum_{i=1}^n  \widehat{D}_n(O_i) - \Psi_n(P_0) \\
    &= (P_n - P_0) \widehat{D}_n + \frac{1}{n} \sum_{i=1}^n m(W_i, \mu_n)  - \Psi_n(P_0) + P_0 \widehat{D}_n\\
      &= P_n D_{0, P_0} + (P_n - P_0) (\widehat{D}_n - D_{0, P_0}) + \text{Rem}_{n,0},
\end{align*}
where we define the second-order remainder $\text{Rem}_{n,0} := \frac{1}{n} \sum_{i=1}^n m(W_i, \mu_n)  - \Psi_n(P_0) + P_0 \widehat{D}_n$. 

By \ref{cond::bounded}, \ref{cond::linearsamplesplit}, and Markov's inequality applied conditional on the independent training dataset, we have that $ (P_n - P_0) (\widehat{D}_n - D_{0, P_0}) =  O_p(n^{-1/2}\|\widehat{D}_n - D_{0, P_0}\|_{P_0})$. Applying the triangle inequality and \ref{cond::bounded}, the consistency of the nuisances in \ref{cond::linearnuisancerate}(i) and \ref{cond::oraclebiaslinear}(i) imply that 
$$\|\widehat{D}_n - D_{0, P_0}\|_{P_0} \leq \|\widehat{D}_n - D_{n, P_0}\|_{P_0} + \|D_{n, P_0} - D_{0, P_0}\|_{P_0} = o_p(1).$$
Hence, $ (P_n - P_0) (\widehat{D}_n - D_{0, P_0}) =  o_p(n^{-1/2}).$ 

Next we show that $\text{Rem}_{n,0} = o_p(n^{-1/2})$. In what follows, we abuse notation and implicitly condition on the training data of the nuisances in our expectations. To this end, observe that
\begin{align*}
    \text{Rem}_{n,0} &= E_0 \left[m(W, \mu_n) + \alpha_n(A, W)\{Y - \mu_n(A, W)\}\right] - \Psi_n(P_0)  \\
    &= E_0 \left[m(W, \mu_n) + \alpha_n(A, W)\{Y - \mu_n(A, W)\}\right] - E_0 \left[m(W, \operatorname{Proj}_n(\mu_0))\right] \\
     &= E_0 \left[m(W, \mu_n - \operatorname{Proj}_n(\mu_0)) + \alpha_n(A, W)\{Y - \mu_n(A, W)\}\right].
\end{align*}
Hence, by the orthogonality properties of the projection $\operatorname{Proj}_n(\mu_0)$, we have $E_0 \left[\alpha_n(A, W)\{\mu_0(A,W) -  \operatorname{Proj}_n(\mu_0)\}\right] = 0$, since $\alpha_n \in \mathcal{H}_n$. Thus, 
\begin{align*}
    \text{Rem}_{n,0}  &= E_0 \left[m(W, \mu_n - \operatorname{Proj}_n(\mu_0))\right] + E_0 \left[\alpha_n(A, W)\{\mu_0(A,W) - \mu_n(A, W)\}\right]\\
&= E_0 \left[m(W, \mu_n - \operatorname{Proj}_n(\mu_0))\right] + E_0 \left[\alpha_n(A, W)\{\operatorname{Proj}_n(\mu_0)(A,W) - \mu_n(A, W)\}\right] \\
&\quad + E_0 \left[\alpha_n(A, W)\{\mu_0(A,W) -  \operatorname{Proj}_n(\mu_0)\}\right] \\
&= E_0 \left[m(W, \mu_n - \operatorname{Proj}_n(\mu_0))\right] + E_0 \left[\alpha_n(A, W)\{\operatorname{Proj}_n(\mu_0)(A,W) - \mu_n(A, W)\}\right].
\end{align*}
Furthermore, by the Riesz representation property of $\alpha_{0, \mathcal{H}_n}$ and the fact that $\mu_n - \operatorname{Proj}_n(\mu_0) \in \mathcal{H}_n$, we have
\[
E_0 \left[m(W, \mu_n - \operatorname{Proj}_n(\mu_0))\right] =  E_0 \left[\alpha_{0, \mathcal{H}_n}(A,W)(\mu_n(A,W) - \operatorname{Proj}_n(\mu_0)(A,W))\right].
\]
Thus,
\begin{align*}
    \text{Rem}_{n,0}
      &=  \left\langle \alpha_n - \alpha_{0,\mathcal{H}_n},\ \operatorname{Proj}_n(\mu_0) - \mu_n \right\rangle_{P_0}.
\end{align*}
By the Cauchy--Schwarz inequality and \ref{cond::linearnuisancerate}(ii),
\[
\left|\text{Rem}_{n,0}\right|
\ \le\
\|\alpha_n - \alpha_{0,\mathcal{H}_n}\|_{P_0,2}\;
\|\operatorname{Proj}_n(\mu_0) - \mu_n\|_{P_0,2} = o_p(n^{-1/2}).
\]
Thus,  $\widehat{\psi}_n  - \Psi_n(P_0) =  P_n D_{0, P_0} + o_p(n^{-1/2}) $.

Combining all of the above, we conclude that
\[
    \widehat{\psi}_n - \Psi_0(P_0) = P_n D_{0,P_0} + o_p(n^{-1/2}).
\]
Hence, $\widehat{\psi}_n$ is $P_0$-asymptotically linear for $\Psi_0(P_0)$ with influence function given by the EIF for $\Psi_0$. It is therefore regular and efficient for $\Psi_0$ \citep{vanderVaartWellner}. In addition, if $\mu_0 \in \mathcal{H}_0$, then $\Psi_0(P_0) = \Psi(P_0)$, and
\[
    \widehat{\psi}_n - \Psi(P_0) = P_n D_{0,P_0} + o_p(n^{-1/2}).
\]
It is regular for $\Psi$ over the oracle submodel $\mathcal{M}_0$ by Theorem~\ref{theorem::oracleRegularity}. Finally, if the conditional variance of the outcome is almost surely constant, then $D_{0,P_0}$ is also the efficient influence function of $\Psi$ relative to the oracle model $\mathcal{M}_0$ \citep{chernoRegRiesz}. The result then follows.

\end{proof}

\section{Proofs for Appendix \ref{appendix::PLM}}

\begin{proof}[Proof of Theorem  \ref{example::theorem::RlearnerEIF}]
    
    Let $\mathcal{H}_0 := \{(a,w) \mapsto m(w) + a \tau(w): m \in L^2(P_0), \tau \in \mathcal{T}_0\}$ be the corresponding oracle model for the outcome regression. By Theorem \ref{theorem::boundedlinearEIF}, the efficient influence function $D_{0,P}$ for $\Psi_0$ at $P$ under the prespecified statistical model $\mathcal{M}$ is given by $o\mapsto \Pi_0 \tau_P(w) - E_0\{\Pi_0 \tau_P(W)\} + \alpha_P(w)\{a - \pi_P(w)\}\{y - \mu_P(w)\}$, where $\alpha_P := \argmin_{\alpha \in \mathcal{H}_0} E_P\left[\alpha(A,W)^2 - 2\{\alpha(1,W)-\alpha(0,W)\}\right]$ is the Riesz representer of the linear functional $\mu \mapsto E_P\{\mu(1,W) - \mu(0,W)\}$ \citep{chernoRegRiesz, chernozhukov2018auto}. We claim that $\alpha_P(A,W) = \gamma_{P, \mathcal{T}_0}(W)\{A-\pi_P(W)\}$  $P$--almost surely. To see this, we observe that $\argmin_{\alpha \in \mathcal{H}_0} E_P\left[\alpha(A,W)^2 - 2\{\alpha(1,W)-\alpha(0,W)\}\right]$ coincides with 
    \begin{equation}\argmin_{m \in L_2(P_0),\gamma \in \mathcal{T}_0} E_P\left[\{m(W) + A \gamma(W) \}^2 - 2\gamma(W)\right].\label{eq:proofthm2}\end{equation} Next, we expand the objective function in \eqref{eq:proofthm2} as
    $$(m,\gamma)\mapsto E_P\left\{m(W)^2 - 2\pi_P(W)m(W)\gamma(W) + A^2 \gamma(W) - 2\gamma(W)  \right\},$$
    and observe that $m$ and $\gamma$ are able to vary freely in the minimization problem. Holding $\gamma$ fixed, we find that the above is minimized by $m_{P,\gamma}:w\mapsto -\pi_P(w)\gamma(w)$. With the choice $m=m_{P,\gamma}$, we obtain the profiled objective function 
    $\gamma\mapsto E_P\left[\{A-\pi_P(W)\}^2 \gamma(W)^2 - 2\gamma(W)\right]$, which is exactly minimized  over $\gamma \in \mathcal{T}_0$ by $\gamma_{P, \mathcal{T}_0}$. Thus, plugging in these optimizers, we conclude that $\alpha_P(a,w) = \{a-\pi_P(w)\}\gamma_{P, \mathcal{T}_0}(w)$. Plugging this expression for $\alpha_P$ in $D_{0, P}$, we obtain the efficient influence function given in the theorem.\end{proof}

 \begin{proof}[Proof of part (i) of Theorem   \ref{example::theorem::RlearnerLimitDistORACLE}]

We have that
\begin{align*}
    \Pi_n \gamma_{0, \mathcal{T}_0}\ :=&\ \ \argmin_{\gamma \in \mathcal{T}_n} E_0\left[\{A-\pi_0(W)\}^2\gamma(W)^2 -  2 \gamma(W) \right]\\
    =&\ \ \argmin_{\gamma \in \mathcal{T}_n} E_0\left( \{A-\pi_0(W)\}^2 \left[ \gamma(W) -  \{A-\pi_0(W)\}^{-2}\right]\right).
\end{align*}  
It follows that $\Pi_n \gamma_{0, \mathcal{T}_0}$ is the overlap-weighted projection of $(a,w) \mapsto \{a-\pi_0( w)\}^{-2}$ onto $\mathcal{T}_n$. Since $\mathcal{T}_n \subseteq \mathcal{T}_{n,0}$, we also have that
\begin{align*}
    \Pi_n \gamma_{0, \mathcal{T}_0} := \argmin_{\gamma \in \mathcal{T}_n} E_0\left[ (A-\pi_0(W))^2 \left\{ \gamma(W) - \gamma_{0, \mathcal{T}_{n,0}}(W) \right\}\right],
\end{align*} 
 where we note that $\Pi_n \gamma_{0, \mathcal{T}_0}$ appears in the efficient influence function $D_{n,P_0}$ of $\Psi_n$ as indicated in Theorem \ref{example::theorem::RlearnerEIF}.

 Let $o \mapsto \widehat D_{n,0}(o) := \tau_n(w) - P_n \tau_n + \Pi_n\gamma_{0}(w)\{a-\pi_n(w)\}\{y - \mu_n(a,w)\}$ be an estimator of the efficient influence function $D_{0,P_0}$ of $\Psi_0$ provided in Theorem \ref{example::theorem::RlearnerEIF}. 
 By the first-order conditions characterizing the minimizer $\tau_n$,  we have that
\begin{align*}
     \widehat{\psi}_n \ &=\ \frac{1}{n}\sum_{i=1}^n \tau_n(W_i)\ =\  \frac{1}{n}\sum_{i=1}^n \tau_n(W_i)   + 
    \frac{1}{n}\sum_{i=1}^n \Pi_n \gamma_{0, \mathcal{T}_0}(W_i)\left\{A-\pi_n(W_i)\right\} \left\{Y_i -\mu_n(A_i,W_i) \right\} \\
    & =\  \widehat{\psi}_n + P_n \widehat D_{n,0}\ .
\end{align*}
As a consequence, we have the bias expansion
\begin{align*}
    \widehat{\psi}_n - \Psi_n(P_0)\ &=\  \widehat{\psi}_n + P_n \widehat D_{n,0}  - \Psi_n(P_0)\\
    &=\ (P_n-P_0)D_{0,P_{0}}  + (P_n-P_0)(\widehat D_{n,0} - D_{0,P_{0}}) + R_{n,0}
\end{align*} 
with
$R_{n,0} :=  \widehat{\psi}_n - \Psi_n(P_0) + P_0 \widehat D_{n,0}$. 

We first show that $ (P_n-P_0)( \widehat D_{n,0}- D_{0,P_{0}} ) = o_p(n^{-1/2})$. By \ref{cond::CATE::DonskerMLE} and preservation of the Donsker property under Lipschitz transformations  \citep{vanderVaartWellner}, $\widehat D_{n,0}- D_{0,P_{0}}$ falls in a Donsker class with probability tending to one. Hence, it suffices to show that $\|  \widehat D_{n,0}- D_{0,P_{0}} \|_{P_0} = o_p(1)$. To this end, we define $\alpha_n$ and $\alpha_0$ pointwise as $\alpha_n(a,w):= \Pi_n\gamma_{0, \mathcal{T}_0}(w) \{a- \pi_n(w)\}$ and $\alpha_0(a,w):= \gamma_{0, \mathcal{T}_0}(w) \{a- \pi_0(w)\}$. Then, we can write $\widehat D_{n,0}(o) = \tau_n(w) - P_n \tau_n + \alpha_n(a,w)\{y - \mu_n(a,w)\}$ and $D_{0,P_0}(o) = \tau_0(w) - P_0 \tau_0 + \alpha_0(a,w)\{y - \mu_0(a,w)\}$. To show that $\|\widehat D_{n,0} - D_{0,P_0}\|_{P_0} = o_p(1)$, we show that $\|\tau_n - \tau_0 - P_n \tau_n + P_0 \tau_0 \|_{P_0} = o_P(1)$ and $\|\alpha_n(\mathcal{I}_Y - \mu_n) -\alpha_0(\mathcal{I}_Y - \mu_0) \|_{P_0} = o_p(1)$ with $\mathcal{I}_Y: o \mapsto y$. By \ref{cond::CATE::consistentNuis}, \ref{cond::CATE::consistentWorking}, and the triangle inequality, we have that $\|\tau_n -   \tau_0 \|_{P_0} \leq \|\tau_n -   \Pi_n \tau_0 \|_{P_0} + \|\Pi_n \tau_0 -   \tau_0 \|_{P_0} = o_p(1)$. Moreover, $P_n \tau_n - P_0 \tau_0 = (P_n - P_0) \tau_n + P_0(\tau_n - \tau_0) = o_p(1)$ since $P_0(\tau_n - \tau_0) \leq \|\tau_n - \tau_0 \|_{P_0}  = o_p(1)$ and $(P_n - P_0) \tau_n = \mathcal{O}_p(n^{-1/2})$ given that $\tau_n$ falls in a Donsker class by \ref{cond::CATE::DonskerMLE}. Hence, $\|\tau_n - \tau_0 - P_n \tau_n + P_0 \tau_0 \|_{P_0} = o_P(1)$ by the triangle inequality. Next, we note that 
\begin{align*}
    \alpha_n(\mathcal{I}_Y - \mu_n) -\alpha_0(\mathcal{I}_Y - \mu_0)   = \alpha_n(\mu_0 - \mu_n)  + (\alpha_n - \alpha_0)(\mathcal{I}_Y - \mu_0)\ .
\end{align*} 
By \ref{cond::CATE::DonskerMLE}, $\mathcal{I}_Y - \mu_0$ and $\alpha_n$ are uniformly bounded so that, by the triangle inequality, the norm of the right-hand side is upper bounded by $\|\alpha_n - \alpha_0 \|_{P_0} + \|\mu_n - \mu_0\|_{P_0}$ up to a constant. We first show that $\|\alpha_n - \alpha_0 \|_{P_0} = o_p(1)$. We note that
\begin{align*}
    \alpha_n - \alpha_0\ =\ \Pi_n \gamma_{0, \mathcal{T}_0}(\mathcal{I}_A - \pi_n) - \gamma_{0, \mathcal{T}_0}(\mathcal{I}_A - \pi_0)\  =\  \Pi_n \gamma_{0, \mathcal{T}_0}(\pi_0 - \pi_n) - (\Pi_n \gamma_n - \gamma_{0, \mathcal{T}_0})(\mathcal{I}_A- \pi_0)
\end{align*} 
with $\mathcal{I}_A: o \mapsto a$, and that $\|(\Pi_n \gamma_n - \gamma_{0, \mathcal{T}_0})(\mathcal{I}_A- \pi_0)\|_{P_0} = \|\Pi_n \gamma_n - \gamma_{0, \mathcal{T}_0}\|_{w_0 P_0} = o_p(1)$ by \ref{cond::CATE::consistentNuis}. Moreover, since $\Pi_n \gamma_{0, \mathcal{T}_0}$ is bounded with probability tending to one by \ref{cond::CATE::DonskerMLE}, we have that $\| \Pi_n \gamma_{0, \mathcal{T}_0}(\pi_0 - \pi_n)  \|_{P_0} =\mathcal{O}_p\left( \|\pi_0 - \pi_n  \|_{P_0}\right) = o_p(1)$ by \ref{cond::CATE::consistentNuis}. We now show that $\|\mu_n - \mu_0\|_{P_0} = o_p(1)$. We note that, by the triangle inequality, \begin{align*}
    \|\mu_n - \mu_0\|_{P_0}\ &\leq\ \|m_n - m_0\|_{P_0} + \| (\mathcal{I}_A-\pi_n)\tau_n +  (\mathcal{I}_A - \pi_0)\tau_0 \|_{P_0}\\
    &\leq\ \|m_n - m_0\|_{P_0} +\| (\mathcal{I}_A-\pi_n)(\tau_n-\tau_0) +  (\pi_n - \pi_0)\tau_0 \|_{P_0}\\
    &=\ \mathcal{O}_p\left(\|m_n - m_0\|_{P_0} +\|\tau_n - \tau_0\|_{P_0} + \| \pi_n - \pi_0\|_{P_0}\right)\ ,
\end{align*} where we use that $\mathcal{I}_A$, $\pi_n$ and $\tau_0$ are bounded with probability tending to one by \ref{cond::CATE::DonskerMLE}.  Hence, by \ref{cond::CATE::consistentNuis}, we have that $\|\mu_n - \mu_0\|_{P_0}= o_p(1)$, and thus,  $\|\widehat D_{n,0} - D_{0,P_0}\|_{P_0} = o_p(1)$, as desired. 

It remains to show that $  R_{n,0} = o_p(n^{-1/2})$. First, we observe that
\begin{align*}
   R_{n,0}\ & =\  \widehat{\psi}_n - \Psi_n(P_{0}) + P_0 \widehat D_{n,0} \\
    &=\  E_0\left[\Pi_n \gamma_{0, \mathcal{T}_0}(W)\{A-\pi_n(W)\}\{Y - \mu_n(A,W) \} \right]   + E_0\{\tau_n(W) - \Pi_n \tau_0(W)\}\\
     &=\  E_0\left[\Pi_n \gamma_{0, \mathcal{T}_0}(W)\{A-\pi_n(W)\}\{\mu_0(A,W) - \mu_n(A,W) \} \right]   + E_0\{\tau_n(W) - \Pi_n \tau_0(W)\}\ ,\end{align*}
where we used the law of iterated expectations. Next, substituting $\mu_n := m_n + (\mathcal{I}_A - \pi_n)\tau_n$ and $\mu_0 := m_0 + (\mathcal{I}_A - \pi_0)\tau_0$, we find that $R_{n,0}=\text{(I)}+\text{(II)}+\text{(III)}$ with \begin{align*}
    \text{(I)}\ &:=\ E_0\left[\Pi_n \gamma_{0, \mathcal{T}_0}(W)\{A-\pi_n(W)\}\{m_0(W) - m_n(W)\}   \right]\\ 
    \text{(II)}\ &:=\ E_0\left(\Pi_n \gamma_{0, \mathcal{T}_0}(W)\{A-\pi_n(W)\}\left[ \{A-\pi_0(W)\}\tau_0(W) - \{A - \pi_n(W)\}\tau_n(W) \right] \right)\\ 
    \text{(III)}\ &:=\ E_0\{\tau_n(W) - \Pi_n \tau_0(W)\}\ .
\end{align*}
First, we note that $\text{(I)}=E_0\left[\Pi_n \gamma_{0, \mathcal{T}_0}(W)\{\pi_0(W)-\pi_n(W)\}\{m_0(W) - m_n(W)\}\right]$ by the law of iterated expectations,  and so,
since $\Pi_n \gamma_{0, \mathcal{T}_0}$ is bounded by \ref{cond::CATE::DRterm}, the Cauchy-Schwarz inequality implies that $\text{(I)}$ is of order $\mathcal{O}_p \left(\|\pi_n - \pi_0\|_{P_0}\|m_n - m_0\|_{P_0}\right)$. Next, we can write $\text{(II)}=\text{(IIa)}+\text{(IIb)}$ with \begin{align*}
    \text{(IIa)}\ &:=\ E_0\left[\Pi_n \gamma_{0, \mathcal{T}_0}(W)\{A-\pi_0(W)\}[  \{A-\pi_0(W)\}  \tau_0(W) - \{A - \pi_n(W)\}\tau_n(W) ] \right]\\
    \text{(IIb)}\ &:=\ E_0\left[\Pi_n \gamma_{0, \mathcal{T}_0}(W)\{\pi_0(W)-\pi_n(W)\}[  \{A-\pi_0(W)\}  \tau_0(W) - \{A - \pi_n(W)\}\tau_n(W) ] \right].
\end{align*}On one hand, we can write
 \begin{align*} 
 \text{(IIa)}\ &=\  E_0\left[\Pi_n \gamma_{0, \mathcal{T}_0}(W)\{A-\pi_0(W)\}^2  \tau_0(W)\right]  - E_0\left[\Pi_n \gamma_{0, \mathcal{T}_0}(W)\{A-\pi_0(W)\}\{A - \pi_n(W)\}\tau_n(W) ] \right]  \\
 &=\  E_0\left[\Pi_n \gamma_{0, \mathcal{T}_0}(W)\{A-\pi_0(W)\}^2  \Pi_n \tau_0(W)\right]  - E_0\left[\Pi_n \gamma_{0, \mathcal{T}_0}(W)\{A-\pi_0(W)\}\{A - \pi_n(W)\}\tau_n(W) ] \right]\\
 &=\ E_0\left[\Pi_n \gamma_{0, \mathcal{T}_0}(W)\{A-\pi_0(W)\}^2 \{\Pi_n\tau_0(W)  - \tau_n(W)\}\right]\\
 &\hspace{.5in}- E_0\left[\Pi_n \gamma_{0, \mathcal{T}_0}(W)\{A-\pi_0(W)\}\{\pi_0(W) - \pi_n(W)\}\Pi_n\tau_0(W)\right]\\
 &=\ E_0\left[\Pi_n \gamma_{0, \mathcal{T}_0}(W)\{A-\pi_0(W)\}^2 \{\Pi_n\tau_0(W)  - \tau_n(W)\}\right]  ,
 \end{align*}where, in particular, we have used the definition of the overlap-weighted projection $\Pi_n$ to replace $\tau_0$ by $\Pi_n \tau_0$,
On the other hand, using that $E_0\left\{A - \pi_0(W)\,|\, W\right\} = 0$ almost surely, we can write
 \begin{align*}
 \text{(IIb)}\ &=\  E_0\left[\Pi_n \gamma_{0, \mathcal{T}_0}(W)\{\pi_0(W)-\pi_n(W)\}[  \{A-\pi_0(W)\}  \tau_0(W) - \{A - \pi_n(W)\}\tau_n(W) ] \right]\\
 &=\ -E_0\left[\Pi_n \gamma_{0, \mathcal{T}_0}(W)\{\pi_0(W)-\pi_n(W)\}^2\tau_n(W) \right].
 \end{align*}
Hence, by the Cauchy-Schwarz inequality, we find that $\text{(II)}$ can be written as
 \begin{align*}
   &E_0\left[\Pi_n \gamma_{0, \mathcal{T}_0}(W)\{A-\pi_0(W)\}^2 \{\Pi_n\tau_0(W)  - \tau_n(W)\}\right]-E_0\left[\Pi_n \gamma_{0, \mathcal{T}_0}(W)\{\pi_0(W)-\pi_n(W)\}^2\tau_n(W) \right]\\
    &=\ E_0\left[\Pi_n \gamma_{0, \mathcal{T}_0}(W)\{A-\pi_0(W)\}^2 \{\Pi_n\tau_0(W)  - \tau_n(W)\}\right]+\mathcal{O}_p\left(\| \pi_n-  \pi_0 \|_{P_0}^2\right).
 \end{align*}
Combining \text{(III)} with our bounds for \text{(I)} and \text{(II)}, we finally find that 
\begin{align*}
     R_{n,0}\ &=\ E_0\left[\Pi_n \gamma_{0, \mathcal{T}_0}(W)\{A-\pi_0(W)\}^2 \{\Pi_n\tau_0(W)  - \tau_n(W)\}\right] + E_0\{\tau_n(W) - \Pi_n \tau_0(W)\}\\
     &\qquad +\mathcal{O}_p\left(\norm{\pi_n-\pi_0}_{P_0}\norm{m_n- m_0}_{P_0}\right) + \mathcal{O}_p\left(\| \pi_n-  \pi_0 \|_{P_0}^2 \right).
\end{align*}
Since $\tau_n -\Pi_n \tau_0 \in \mathcal{T}_n$ by \ref{cond::CATE::nestedModel}, and in view of the proof of Theorem \ref{example::theorem::RlearnerEIF}, we have that
\begin{align*}
    E_0\left\{\tau_n(W) - \Pi_n \tau_0(W)\right\}\ &=\ E_0\left[\Pi_n \gamma_{0, \mathcal{T}_0}(W)\{A-\pi_0(W)\}^2\{\tau_n(W) - \Pi_n \tau_0(W) \}\right]\\
    &=\  -E_0\left[\Pi_n \gamma_{0, \mathcal{T}_0}(W)\{A-\pi_0(W)\}^2\{\Pi_n \tau_0(W) - \tau_n(W) \}\right].
\end{align*}
Thus, by Condition \ref{cond::CATE::DRterm}, we conclude that $R_{n,0} = \mathcal{O}_p\left(\norm{\pi_0-\pi_n}_{P_0}\norm{m_n- m_0}_{P_0}+\norm{\pi_0-\pi_n}_{P_0}^2\right)$ is of order $o_p(n^{-1/2})$.\end{proof}

 \begin{proof}[Proof of part (ii)-(iii) of Theorem   \ref{example::theorem::RlearnerLimitDistORACLE}]
     
Under the stated conditions, Theorem \ref{example::theorem::RlearnerLimitDist} implies that the ADML estimator is asymptotically linear for $\Psi_n(P_0)$ with $\widehat{\psi}_n - \Psi_n(P_0) = (P_n - P_0)D_{0,P_0} + o_p(n^{-1/2})$, where $D_{0,P_0}$ is the efficient influence function of the oracle parameter $\Psi_0$. We will verify  that $\Psi_n(P_0) - \Psi_0(P_0)$ is $o_p(n^{-1/2})$ under the stated conditions. The result then follows from the proof of Theorem \ref{theorem::oracleEff}.

 Let $(a,w) \mapsto \mu_{n,0}(a,w) := m_0(w) + \{a - \pi_0(w)\}\Pi_n \tau_0(w)$ be an oracle approximation of $\mu_0$ compatible with $\Pi_n \tau_0$. By the Riesz representation theorem, we have that $\Psi_n(P_0)=E_0\left\{\Pi_n \tau_0(W)\right\}=E_0\left\{\mu_{n,0}(1,W) - \mu_{n,0}(0,W)\right\}=E_0\left\{\alpha_{n,0}(W) \mu_{n,0}(A,W)\right\}$,
where $\alpha_{n,0}(a,w) := \gamma_{0, \mathcal{T}_{n,0}}(w)\{a - \pi_0(w)\} \in \mathcal{H}_{n,0}$ is the regression Riesz representer implied by $\mathcal{T}_{n,0}$ (see the proof of Theorem \ref{example::theorem::RlearnerEIF}). We can write
\begin{align*}
    \Psi_n(P_0)\ &=\ E_0\left[\gamma_{0, \mathcal{T}_{n,0}}(W)\{A - \pi_0(W)\}\mu_{n,0}(A,W)\right] \\
    &=\  E_0\left[\gamma_{0, \mathcal{T}_{n,0}}(W)\{A - \pi_0(W)\}[m_0(W) + \{A - \pi_0(W)\}\Pi_n \tau_0(W)]\right]\\
    & =\  E_0\left[\gamma_{0, \mathcal{T}_{n,0}}(W)\{A - \pi_0(W)\}^2 \Pi_n \tau_0(W)\right].
\end{align*}
Similarly, since $\tau_0 \in \mathcal{T}_{n,0}$, we have that $ \Psi_0(P_0) =E_0\left[\gamma_{0, \mathcal{T}_{n,0}}(W)\{A - \pi_0(W)\}^2 \tau_{0}(W)\right] $. Therefore, we can write 
$\Psi_n(P_0) - \Psi_0(P_0)  =   E_0\left[\gamma_{0, \mathcal{T}_{n,0}}(W)\{A - \pi_0(W)\}^2 \{\Pi_n \tau_0(W) - \tau_0(W)\}\right]$.
Since $\Pi_n \tau_0$ is the $\pi_0(1-\pi_0)$--weighted projection of $\tau_0$ onto $\mathcal{T}_n$, the orthogonality condition \[E_0\left[\gamma(W)\{A - \pi_0(W)\}^2 \{\Pi_n \tau_0(W) - \tau_0(W)\}\right] = 0\] holds for each $\gamma \in \mathcal{T}_n$. In particular, for the choice $\gamma = \Pi_n \gamma_{0, \mathcal{T}_0} $, we find that
\begin{align*}
    \Psi_n(P_0) - \Psi_0(P_0)\ &=\ E_0\left[\{\gamma_{0, \mathcal{T}_{n,0}}(W) - \Pi_n \gamma_{0, \mathcal{T}_{n,0}} (W)\}\{A - \pi_0(W)\}^2 \{\Pi_n \tau_0(W) - \tau_0(W)\}\right]\\
    &=\ \mathcal{O}_p\left(\norm{\Pi_n \gamma_{0, \mathcal{T}_0}  - \gamma_{0, \mathcal{T}_{n,0}}}_{w_0P_0}\norm{\Pi_n \tau_0 - \tau_0}_{w_0P_0}\right)
\end{align*}
by the Cauchy-Schwarz inequality. Thus, as desired, $\Psi_n(P_0)-\Psi_0(P_0)$ is of order $o_p(n^{-1/2})$ by \ref{cond::CATE::critBias}.

We have shown that $  \widehat{\psi}_n - \Psi_0(P_0) = (P_n - P_0) D_{0,P_0} + o_p(n^{-1/2}) $, so that $\widehat{\psi}_n$ is an asymptotically linear estimator with influence function $D_{0,P_0}$ equal to the efficient influence function of $\Psi_0$ under $\mathcal{M}_{\mathrm{np}}$. Thus, $\widehat{\psi}_n$ is efficient for $\Psi_0(P)$. Moreover, since efficient estimators are necessarily regular \citep{vanderVaartWellner}, we also have that $\widehat{\psi}_n$ is regular for $\Psi_0$. Finally, if the conditional variance of the outcome is almost surely constant, then $D_{0,P_0}$ is also the efficient influence function of $\Psi$ relative to the oracle model $\mathcal{M}_0$ \citep{chernoRegRiesz}. The result then follows.\end{proof}

\begin{proof}[Proof of Corollary \ref{cor::CATEinf}]
    This is a consequence of Theorem \ref{theorem::oracleRegularity} and regularity of the ADML estimator for $\Psi_0$.\end{proof}

\end{document}